\documentclass[journal,a4paper]{IEEEtran}

\usepackage{amsthm}
\usepackage{bm}
\usepackage{thm-restate}
\usepackage{cite}
\usepackage{caption}
\usepackage{subcaption}
\usepackage[cmex10]{amsmath}
\usepackage{amssymb,amsfonts}
\usepackage{algorithmic}
\usepackage{graphicx}
\usepackage{textcomp}
\usepackage{xcolor}
\usepackage{tabularray}
\usepackage{contour}
\usepackage{mathtools}
\usepackage{eqnarray}
\usepackage{hyperref}

\usepackage[utf8]{inputenc}
\usepackage{mleftright}
\mleftright

\newlength{\figwidth}
\usepackage{pgfplots}
 \usetikzlibrary{arrows.meta}
 \usetikzlibrary{backgrounds}
 \usepgfplotslibrary{patchplots}
 \usepgfplotslibrary{fillbetween}
 \pgfplotsset{%
      compat=1.18,
     layers/standard/.define layer set={%
         background,axis background,axis grid,axis ticks,axis lines,axis tick labels,pre main,main,axis descriptions,axis foreground%
     }{
         grid style={/pgfplots/on layer=axis grid},%
         tick style={/pgfplots/on layer=axis ticks},%
         axis line style={/pgfplots/on layer=axis lines},%
         label style={/pgfplots/on layer=axis descriptions},%
         legend style={/pgfplots/on layer=axis descriptions},%
         title style={/pgfplots/on layer=axis descriptions},%
         colorbar style={/pgfplots/on layer=axis descriptions},%
         ticklabel style={/pgfplots/on layer=axis tick labels},%
         axis background@ style={/pgfplots/on layer=axis background},%
         3d box foreground style={/pgfplots/on layer=axis foreground},%
     },
 }

\makeatletter
\pgfplotsset{
  log x thin/.code={%
    \gdef\logxthin@last{-1000}%
    \xdef\logxthin@step{\the\dimexpr 2.302585pt/#1\relax}%
    \pgfplotsset{filter discard warning=false, x filter/.code={%
      \let\logxthin@x\pgfmathresult
      \ifdim\dimexpr\logxthin@x pt-\logxthin@last pt\relax<\logxthin@step
        \let\pgfmathresult\pgfutil@empty
      \else
        \global\let\logxthin@last\logxthin@x
      \fi}}%
  },
}
\makeatother

 \pgfplotsset{
 colormap={plots1}{rgb(0.00000000)=(0.30980000,0.10090000,0.23840000)
 rgb(0.00392157)=(0.30690000,0.10310000,0.24320000)
 rgb(0.00784314)=(0.30410000,0.10530000,0.24800000)
 rgb(0.01176471)=(0.30120000,0.10760000,0.25300000)
 rgb(0.01568627)=(0.29840000,0.11020000,0.25800000)
 rgb(0.01960784)=(0.29550000,0.11280000,0.26310000)
 rgb(0.02352941)=(0.29270000,0.11540000,0.26820000)
 rgb(0.02745098)=(0.28980000,0.11830000,0.27350000)
 rgb(0.03137255)=(0.28690000,0.12120000,0.27890000)
 rgb(0.03529412)=(0.28410000,0.12430000,0.28430000)
 rgb(0.03921569)=(0.28120000,0.12750000,0.28990000)
 rgb(0.04313725)=(0.27830000,0.13090000,0.29550000)
 rgb(0.04705882)=(0.27540000,0.13430000,0.30120000)
 rgb(0.05098039)=(0.27250000,0.13790000,0.30700000)
 rgb(0.05490196)=(0.26960000,0.14160000,0.31290000)
 rgb(0.05882353)=(0.26670000,0.14550000,0.31890000)
 rgb(0.06274510)=(0.26370000,0.14940000,0.32490000)
 rgb(0.06666667)=(0.26080000,0.15350000,0.33110000)
 rgb(0.07058824)=(0.25780000,0.15770000,0.33730000)
 rgb(0.07450980)=(0.25490000,0.16210000,0.34360000)
 rgb(0.07843137)=(0.25190000,0.16650000,0.34990000)
 rgb(0.08235294)=(0.24900000,0.17110000,0.35630000)
 rgb(0.08627451)=(0.24600000,0.17590000,0.36280000)
 rgb(0.09019608)=(0.24310000,0.18070000,0.36930000)
 rgb(0.09411765)=(0.24020000,0.18580000,0.37590000)
 rgb(0.09803922)=(0.23730000,0.19080000,0.38250000)
 rgb(0.10196078)=(0.23440000,0.19610000,0.38920000)
 rgb(0.10588235)=(0.23160000,0.20140000,0.39590000)
 rgb(0.10980392)=(0.22870000,0.20690000,0.40260000)
 rgb(0.11372549)=(0.22600000,0.21250000,0.40940000)
 rgb(0.11764706)=(0.22330000,0.21820000,0.41610000)
 rgb(0.12156863)=(0.22070000,0.22410000,0.42290000)
 rgb(0.12549020)=(0.21820000,0.23000000,0.42970000)
 rgb(0.12941176)=(0.21580000,0.23610000,0.43650000)
 rgb(0.13333333)=(0.21340000,0.24220000,0.44330000)
 rgb(0.13725490)=(0.21130000,0.24850000,0.45000000)
 rgb(0.14117647)=(0.20920000,0.25480000,0.45680000)
 rgb(0.14509804)=(0.20740000,0.26130000,0.46350000)
 rgb(0.14901961)=(0.20570000,0.26780000,0.47030000)
 rgb(0.15294118)=(0.20430000,0.27440000,0.47700000)
 rgb(0.15686275)=(0.20300000,0.28110000,0.48360000)
 rgb(0.16078431)=(0.20200000,0.28790000,0.49020000)
 rgb(0.16470588)=(0.20130000,0.29480000,0.49680000)
 rgb(0.16862745)=(0.20080000,0.30170000,0.50340000)
 rgb(0.17254902)=(0.20070000,0.30870000,0.50990000)
 rgb(0.17647059)=(0.20080000,0.31580000,0.51640000)
 rgb(0.18039216)=(0.20130000,0.32290000,0.52280000)
 rgb(0.18431373)=(0.20210000,0.33010000,0.52920000)
 rgb(0.18823529)=(0.20330000,0.33740000,0.53550000)
 rgb(0.19215686)=(0.20480000,0.34470000,0.54180000)
 rgb(0.19607843)=(0.20680000,0.35200000,0.54800000)
 rgb(0.20000000)=(0.20900000,0.35940000,0.55420000)
 rgb(0.20392157)=(0.21170000,0.36690000,0.56030000)
 rgb(0.20784314)=(0.21470000,0.37430000,0.56640000)
 rgb(0.21176471)=(0.21810000,0.38190000,0.57250000)
 rgb(0.21568627)=(0.22180000,0.38940000,0.57850000)
 rgb(0.21960784)=(0.22590000,0.39700000,0.58450000)
 rgb(0.22352941)=(0.23040000,0.40460000,0.59040000)
 rgb(0.22745098)=(0.23530000,0.41230000,0.59620000)
 rgb(0.23137255)=(0.24040000,0.42000000,0.60210000)
 rgb(0.23529412)=(0.24590000,0.42770000,0.60790000)
 rgb(0.23921569)=(0.25170000,0.43540000,0.61360000)
 rgb(0.24313725)=(0.25780000,0.44320000,0.61930000)
 rgb(0.24705882)=(0.26420000,0.45100000,0.62500000)
 rgb(0.25098039)=(0.27090000,0.45870000,0.63060000)
 rgb(0.25490196)=(0.27790000,0.46650000,0.63620000)
 rgb(0.25882353)=(0.28510000,0.47440000,0.64170000)
 rgb(0.26274510)=(0.29260000,0.48220000,0.64720000)
 rgb(0.26666667)=(0.30030000,0.49000000,0.65270000)
 rgb(0.27058824)=(0.30830000,0.49780000,0.65810000)
 rgb(0.27450980)=(0.31650000,0.50560000,0.66350000)
 rgb(0.27843137)=(0.32480000,0.51350000,0.66880000)
 rgb(0.28235294)=(0.33340000,0.52130000,0.67410000)
 rgb(0.28627451)=(0.34220000,0.52910000,0.67930000)
 rgb(0.29019608)=(0.35120000,0.53690000,0.68440000)
 rgb(0.29411765)=(0.36030000,0.54460000,0.68950000)
 rgb(0.29803922)=(0.36960000,0.55240000,0.69460000)
 rgb(0.30196078)=(0.37900000,0.56010000,0.69960000)
 rgb(0.30588235)=(0.38860000,0.56780000,0.70450000)
 rgb(0.30980392)=(0.39830000,0.57540000,0.70930000)
 rgb(0.31372549)=(0.40820000,0.58300000,0.71410000)
 rgb(0.31764706)=(0.41820000,0.59050000,0.71870000)
 rgb(0.32156863)=(0.42830000,0.59800000,0.72330000)
 rgb(0.32549020)=(0.43850000,0.60550000,0.72780000)
 rgb(0.32941176)=(0.44890000,0.61280000,0.73220000)
 rgb(0.33333333)=(0.45930000,0.62010000,0.73640000)
 rgb(0.33725490)=(0.46980000,0.62730000,0.74060000)
 rgb(0.34117647)=(0.48040000,0.63450000,0.74460000)
 rgb(0.34509804)=(0.49100000,0.64150000,0.74850000)
 rgb(0.34901961)=(0.50170000,0.64840000,0.75230000)
 rgb(0.35294118)=(0.51250000,0.65520000,0.75590000)
 rgb(0.35686275)=(0.52330000,0.66200000,0.75930000)
 rgb(0.36078431)=(0.53410000,0.66850000,0.76260000)
 rgb(0.36470588)=(0.54490000,0.67490000,0.76560000)
 rgb(0.36862745)=(0.55580000,0.68120000,0.76850000)
 rgb(0.37254902)=(0.56660000,0.68740000,0.77120000)
 rgb(0.37647059)=(0.57750000,0.69330000,0.77370000)
 rgb(0.38039216)=(0.58830000,0.69910000,0.77600000)
 rgb(0.38431373)=(0.59900000,0.70470000,0.77800000)
 rgb(0.38823529)=(0.60980000,0.71010000,0.77980000)
 rgb(0.39215686)=(0.62040000,0.71530000,0.78140000)
 rgb(0.39607843)=(0.63100000,0.72030000,0.78260000)
 rgb(0.40000000)=(0.64140000,0.72500000,0.78360000)
 rgb(0.40392157)=(0.65180000,0.72950000,0.78440000)
 rgb(0.40784314)=(0.66200000,0.73380000,0.78480000)
 rgb(0.41176471)=(0.67210000,0.73780000,0.78500000)
 rgb(0.41568627)=(0.68200000,0.74150000,0.78480000)
 rgb(0.41960784)=(0.69170000,0.74490000,0.78430000)
 rgb(0.42352941)=(0.70130000,0.74810000,0.78360000)
 rgb(0.42745098)=(0.71060000,0.75100000,0.78250000)
 rgb(0.43137255)=(0.71980000,0.75360000,0.78100000)
 rgb(0.43529412)=(0.72870000,0.75580000,0.77930000)
 rgb(0.43921569)=(0.73730000,0.75780000,0.77720000)
 rgb(0.44313725)=(0.74570000,0.75940000,0.77480000)
 rgb(0.44705882)=(0.75380000,0.76080000,0.77200000)
 rgb(0.45098039)=(0.76160000,0.76180000,0.76900000)
 rgb(0.45490196)=(0.76910000,0.76250000,0.76560000)
 rgb(0.45882353)=(0.77640000,0.76280000,0.76190000)
 rgb(0.46274510)=(0.78320000,0.76290000,0.75780000)
 rgb(0.46666667)=(0.78980000,0.76260000,0.75350000)
 rgb(0.47058824)=(0.79610000,0.76200000,0.74880000)
 rgb(0.47450980)=(0.80200000,0.76110000,0.74390000)
 rgb(0.47843137)=(0.80760000,0.75980000,0.73870000)
 rgb(0.48235294)=(0.81280000,0.75830000,0.73320000)
 rgb(0.48627451)=(0.81770000,0.75650000,0.72740000)
 rgb(0.49019608)=(0.82220000,0.75430000,0.72130000)
 rgb(0.49411765)=(0.82640000,0.75190000,0.71510000)
 rgb(0.49803922)=(0.83030000,0.74920000,0.70860000)
 rgb(0.50196078)=(0.83390000,0.74630000,0.70180000)
 rgb(0.50588235)=(0.83710000,0.74310000,0.69490000)
 rgb(0.50980392)=(0.84000000,0.73960000,0.68770000)
 rgb(0.51372549)=(0.84250000,0.73590000,0.68040000)
 rgb(0.51764706)=(0.84480000,0.73200000,0.67280000)
 rgb(0.52156863)=(0.84670000,0.72780000,0.66520000)
 rgb(0.52549020)=(0.84840000,0.72340000,0.65730000)
 rgb(0.52941176)=(0.84970000,0.71890000,0.64930000)
 rgb(0.53333333)=(0.85080000,0.71410000,0.64120000)
 rgb(0.53725490)=(0.85160000,0.70920000,0.63300000)
 rgb(0.54117647)=(0.85220000,0.70410000,0.62460000)
 rgb(0.54509804)=(0.85250000,0.69880000,0.61620000)
 rgb(0.54901961)=(0.85260000,0.69340000,0.60760000)
 rgb(0.55294118)=(0.85240000,0.68780000,0.59900000)
 rgb(0.55686275)=(0.85200000,0.68210000,0.59030000)
 rgb(0.56078431)=(0.85140000,0.67620000,0.58150000)
 rgb(0.56470588)=(0.85050000,0.67030000,0.57270000)
 rgb(0.56862745)=(0.84950000,0.66420000,0.56380000)
 rgb(0.57254902)=(0.84830000,0.65800000,0.55490000)
 rgb(0.57647059)=(0.84690000,0.65170000,0.54590000)
 rgb(0.58039216)=(0.84530000,0.64530000,0.53690000)
 rgb(0.58431373)=(0.84350000,0.63880000,0.52780000)
 rgb(0.58823529)=(0.84160000,0.63220000,0.51880000)
 rgb(0.59215686)=(0.83950000,0.62550000,0.50970000)
 rgb(0.59607843)=(0.83720000,0.61880000,0.50050000)
 rgb(0.60000000)=(0.83480000,0.61190000,0.49140000)
 rgb(0.60392157)=(0.83230000,0.60500000,0.48230000)
 rgb(0.60784314)=(0.82960000,0.59800000,0.47310000)
 rgb(0.61176471)=(0.82670000,0.59090000,0.46400000)
 rgb(0.61568627)=(0.82370000,0.58380000,0.45480000)
 rgb(0.61960784)=(0.82060000,0.57650000,0.44570000)
 rgb(0.62352941)=(0.81730000,0.56920000,0.43660000)
 rgb(0.62745098)=(0.81390000,0.56190000,0.42750000)
 rgb(0.63137255)=(0.81040000,0.55450000,0.41830000)
 rgb(0.63529412)=(0.80670000,0.54690000,0.40930000)
 rgb(0.63921569)=(0.80290000,0.53940000,0.40020000)
 rgb(0.64313725)=(0.79890000,0.53170000,0.39110000)
 rgb(0.64705882)=(0.79480000,0.52400000,0.38210000)
 rgb(0.65098039)=(0.79060000,0.51620000,0.37310000)
 rgb(0.65490196)=(0.78620000,0.50840000,0.36420000)
 rgb(0.65882353)=(0.78170000,0.50040000,0.35530000)
 rgb(0.66274510)=(0.77710000,0.49240000,0.34640000)
 rgb(0.66666667)=(0.77230000,0.48440000,0.33760000)
 rgb(0.67058824)=(0.76730000,0.47620000,0.32880000)
 rgb(0.67450980)=(0.76220000,0.46800000,0.32010000)
 rgb(0.67843137)=(0.75700000,0.45970000,0.31140000)
 rgb(0.68235294)=(0.75160000,0.45140000,0.30280000)
 rgb(0.68627451)=(0.74600000,0.44300000,0.29430000)
 rgb(0.69019608)=(0.74030000,0.43450000,0.28590000)
 rgb(0.69411765)=(0.73450000,0.42600000,0.27760000)
 rgb(0.69803922)=(0.72850000,0.41740000,0.26940000)
 rgb(0.70196078)=(0.72230000,0.40880000,0.26130000)
 rgb(0.70588235)=(0.71600000,0.40010000,0.25340000)
 rgb(0.70980392)=(0.70960000,0.39140000,0.24550000)
 rgb(0.71372549)=(0.70300000,0.38260000,0.23780000)
 rgb(0.71764706)=(0.69630000,0.37380000,0.23030000)
 rgb(0.72156863)=(0.68940000,0.36500000,0.22290000)
 rgb(0.72549020)=(0.68250000,0.35620000,0.21570000)
 rgb(0.72941176)=(0.67540000,0.34740000,0.20870000)
 rgb(0.73333333)=(0.66820000,0.33850000,0.20190000)
 rgb(0.73725490)=(0.66090000,0.32970000,0.19530000)
 rgb(0.74117647)=(0.65350000,0.32090000,0.18900000)
 rgb(0.74509804)=(0.64600000,0.31220000,0.18290000)
 rgb(0.74901961)=(0.63850000,0.30350000,0.17700000)
 rgb(0.75294118)=(0.63090000,0.29490000,0.17140000)
 rgb(0.75686275)=(0.62330000,0.28630000,0.16600000)
 rgb(0.76078431)=(0.61560000,0.27790000,0.16100000)
 rgb(0.76470588)=(0.60790000,0.26950000,0.15620000)
 rgb(0.76862745)=(0.60020000,0.26120000,0.15160000)
 rgb(0.77254902)=(0.59250000,0.25310000,0.14750000)
 rgb(0.77647059)=(0.58490000,0.24500000,0.14350000)
 rgb(0.78039216)=(0.57720000,0.23720000,0.13990000)
 rgb(0.78431373)=(0.56960000,0.22950000,0.13660000)
 rgb(0.78823529)=(0.56210000,0.22190000,0.13360000)
 rgb(0.79215686)=(0.55460000,0.21450000,0.13080000)
 rgb(0.79607843)=(0.54720000,0.20730000,0.12840000)
 rgb(0.80000000)=(0.53990000,0.20020000,0.12620000)
 rgb(0.80392157)=(0.53270000,0.19340000,0.12430000)
 rgb(0.80784314)=(0.52560000,0.18680000,0.12270000)
 rgb(0.81176471)=(0.51860000,0.18030000,0.12140000)
 rgb(0.81568627)=(0.51170000,0.17410000,0.12030000)
 rgb(0.81960784)=(0.50500000,0.16810000,0.11950000)
 rgb(0.82352941)=(0.49840000,0.16230000,0.11890000)
 rgb(0.82745098)=(0.49190000,0.15670000,0.11850000)
 rgb(0.83137255)=(0.48560000,0.15130000,0.11840000)
 rgb(0.83529412)=(0.47940000,0.14620000,0.11840000)
 rgb(0.83921569)=(0.47330000,0.14130000,0.11870000)
 rgb(0.84313725)=(0.46740000,0.13650000,0.11910000)
 rgb(0.84705882)=(0.46160000,0.13200000,0.11970000)
 rgb(0.85098039)=(0.45600000,0.12770000,0.12050000)
 rgb(0.85490196)=(0.45050000,0.12360000,0.12150000)
 rgb(0.85882353)=(0.44510000,0.11980000,0.12260000)
 rgb(0.86274510)=(0.43990000,0.11620000,0.12390000)
 rgb(0.86666667)=(0.43480000,0.11280000,0.12540000)
 rgb(0.87058824)=(0.42980000,0.10960000,0.12690000)
 rgb(0.87450980)=(0.42500000,0.10650000,0.12860000)
 rgb(0.87843137)=(0.42030000,0.10370000,0.13050000)
 rgb(0.88235294)=(0.41570000,0.10110000,0.13240000)
 rgb(0.88627451)=(0.41120000,0.09870000,0.13450000)
 rgb(0.89019608)=(0.40680000,0.09640000,0.13660000)
 rgb(0.89411765)=(0.40260000,0.09440000,0.13890000)
 rgb(0.89803922)=(0.39840000,0.09250000,0.14130000)
 rgb(0.90196078)=(0.39430000,0.09090000,0.14380000)
 rgb(0.90588235)=(0.39040000,0.08940000,0.14640000)
 rgb(0.90980392)=(0.38650000,0.08810000,0.14900000)
 rgb(0.91372549)=(0.38270000,0.08700000,0.15180000)
 rgb(0.91764706)=(0.37900000,0.08590000,0.15470000)
 rgb(0.92156863)=(0.37530000,0.08510000,0.15760000)
 rgb(0.92549020)=(0.37170000,0.08450000,0.16060000)
 rgb(0.92941176)=(0.36820000,0.08410000,0.16380000)
 rgb(0.93333333)=(0.36480000,0.08370000,0.16700000)
 rgb(0.93725490)=(0.36140000,0.08350000,0.17030000)
 rgb(0.94117647)=(0.35810000,0.08350000,0.17360000)
 rgb(0.94509804)=(0.35480000,0.08360000,0.17710000)
 rgb(0.94901961)=(0.35160000,0.08390000,0.18060000)
 rgb(0.95294118)=(0.34840000,0.08430000,0.18420000)
 rgb(0.95686275)=(0.34530000,0.08480000,0.18790000)
 rgb(0.96078431)=(0.34220000,0.08540000,0.19170000)
 rgb(0.96470588)=(0.33910000,0.08620000,0.19550000)
 rgb(0.96862745)=(0.33610000,0.08720000,0.19940000)
 rgb(0.97254902)=(0.33310000,0.08820000,0.20350000)
 rgb(0.97647059)=(0.33010000,0.08940000,0.20750000)
 rgb(0.98039216)=(0.32720000,0.09070000,0.21170000)
 rgb(0.98431373)=(0.32420000,0.09210000,0.21600000)
 rgb(0.98823529)=(0.32130000,0.09360000,0.22030000)
 rgb(0.99215686)=(0.31840000,0.09530000,0.22470000)
 rgb(0.99607843)=(0.31550000,0.09700000,0.22920000)
 rgb(1.00000000)=(0.31260000,0.09890000,0.23380000)},
 }

 \pgfplotsset{
 colormap={plots1}{rgb(0.00000000)=(0.30980000,0.10090000,0.23840000)
 rgb(0.00392157)=(0.30690000,0.10310000,0.24320000)
 rgb(0.00784314)=(0.30410000,0.10530000,0.24800000)
 rgb(0.01176471)=(0.30120000,0.10760000,0.25300000)
 rgb(0.01568627)=(0.29840000,0.11020000,0.25800000)
 rgb(0.01960784)=(0.29550000,0.11280000,0.26310000)
 rgb(0.02352941)=(0.29270000,0.11540000,0.26820000)
 rgb(0.02745098)=(0.28980000,0.11830000,0.27350000)
 rgb(0.03137255)=(0.28690000,0.12120000,0.27890000)
 rgb(0.03529412)=(0.28410000,0.12430000,0.28430000)
 rgb(0.03921569)=(0.28120000,0.12750000,0.28990000)
 rgb(0.04313725)=(0.27830000,0.13090000,0.29550000)
 rgb(0.04705882)=(0.27540000,0.13430000,0.30120000)
 rgb(0.05098039)=(0.27250000,0.13790000,0.30700000)
 rgb(0.05490196)=(0.26960000,0.14160000,0.31290000)
 rgb(0.05882353)=(0.26670000,0.14550000,0.31890000)
 rgb(0.06274510)=(0.26370000,0.14940000,0.32490000)
 rgb(0.06666667)=(0.26080000,0.15350000,0.33110000)
 rgb(0.07058824)=(0.25780000,0.15770000,0.33730000)
 rgb(0.07450980)=(0.25490000,0.16210000,0.34360000)
 rgb(0.07843137)=(0.25190000,0.16650000,0.34990000)
 rgb(0.08235294)=(0.24900000,0.17110000,0.35630000)
 rgb(0.08627451)=(0.24600000,0.17590000,0.36280000)
 rgb(0.09019608)=(0.24310000,0.18070000,0.36930000)
 rgb(0.09411765)=(0.24020000,0.18580000,0.37590000)
 rgb(0.09803922)=(0.23730000,0.19080000,0.38250000)
 rgb(0.10196078)=(0.23440000,0.19610000,0.38920000)
 rgb(0.10588235)=(0.23160000,0.20140000,0.39590000)
 rgb(0.10980392)=(0.22870000,0.20690000,0.40260000)
 rgb(0.11372549)=(0.22600000,0.21250000,0.40940000)
 rgb(0.11764706)=(0.22330000,0.21820000,0.41610000)
 rgb(0.12156863)=(0.22070000,0.22410000,0.42290000)
 rgb(0.12549020)=(0.21820000,0.23000000,0.42970000)
 rgb(0.12941176)=(0.21580000,0.23610000,0.43650000)
 rgb(0.13333333)=(0.21340000,0.24220000,0.44330000)
 rgb(0.13725490)=(0.21130000,0.24850000,0.45000000)
 rgb(0.14117647)=(0.20920000,0.25480000,0.45680000)
 rgb(0.14509804)=(0.20740000,0.26130000,0.46350000)
 rgb(0.14901961)=(0.20570000,0.26780000,0.47030000)
 rgb(0.15294118)=(0.20430000,0.27440000,0.47700000)
 rgb(0.15686275)=(0.20300000,0.28110000,0.48360000)
 rgb(0.16078431)=(0.20200000,0.28790000,0.49020000)
 rgb(0.16470588)=(0.20130000,0.29480000,0.49680000)
 rgb(0.16862745)=(0.20080000,0.30170000,0.50340000)
 rgb(0.17254902)=(0.20070000,0.30870000,0.50990000)
 rgb(0.17647059)=(0.20080000,0.31580000,0.51640000)
 rgb(0.18039216)=(0.20130000,0.32290000,0.52280000)
 rgb(0.18431373)=(0.20210000,0.33010000,0.52920000)
 rgb(0.18823529)=(0.20330000,0.33740000,0.53550000)
 rgb(0.19215686)=(0.20480000,0.34470000,0.54180000)
 rgb(0.19607843)=(0.20680000,0.35200000,0.54800000)
 rgb(0.20000000)=(0.20900000,0.35940000,0.55420000)
 rgb(0.20392157)=(0.21170000,0.36690000,0.56030000)
 rgb(0.20784314)=(0.21470000,0.37430000,0.56640000)
 rgb(0.21176471)=(0.21810000,0.38190000,0.57250000)
 rgb(0.21568627)=(0.22180000,0.38940000,0.57850000)
 rgb(0.21960784)=(0.22590000,0.39700000,0.58450000)
 rgb(0.22352941)=(0.23040000,0.40460000,0.59040000)
 rgb(0.22745098)=(0.23530000,0.41230000,0.59620000)
 rgb(0.23137255)=(0.24040000,0.42000000,0.60210000)
 rgb(0.23529412)=(0.24590000,0.42770000,0.60790000)
 rgb(0.23921569)=(0.25170000,0.43540000,0.61360000)
 rgb(0.24313725)=(0.25780000,0.44320000,0.61930000)
 rgb(0.24705882)=(0.26420000,0.45100000,0.62500000)
 rgb(0.25098039)=(0.27090000,0.45870000,0.63060000)
 rgb(0.25490196)=(0.27790000,0.46650000,0.63620000)
 rgb(0.25882353)=(0.28510000,0.47440000,0.64170000)
 rgb(0.26274510)=(0.29260000,0.48220000,0.64720000)
 rgb(0.26666667)=(0.30030000,0.49000000,0.65270000)
 rgb(0.27058824)=(0.30830000,0.49780000,0.65810000)
 rgb(0.27450980)=(0.31650000,0.50560000,0.66350000)
 rgb(0.27843137)=(0.32480000,0.51350000,0.66880000)
 rgb(0.28235294)=(0.33340000,0.52130000,0.67410000)
 rgb(0.28627451)=(0.34220000,0.52910000,0.67930000)
 rgb(0.29019608)=(0.35120000,0.53690000,0.68440000)
 rgb(0.29411765)=(0.36030000,0.54460000,0.68950000)
 rgb(0.29803922)=(0.36960000,0.55240000,0.69460000)
 rgb(0.30196078)=(0.37900000,0.56010000,0.69960000)
 rgb(0.30588235)=(0.38860000,0.56780000,0.70450000)
 rgb(0.30980392)=(0.39830000,0.57540000,0.70930000)
 rgb(0.31372549)=(0.40820000,0.58300000,0.71410000)
 rgb(0.31764706)=(0.41820000,0.59050000,0.71870000)
 rgb(0.32156863)=(0.42830000,0.59800000,0.72330000)
 rgb(0.32549020)=(0.43850000,0.60550000,0.72780000)
 rgb(0.32941176)=(0.44890000,0.61280000,0.73220000)
 rgb(0.33333333)=(0.45930000,0.62010000,0.73640000)
 rgb(0.33725490)=(0.46980000,0.62730000,0.74060000)
 rgb(0.34117647)=(0.48040000,0.63450000,0.74460000)
 rgb(0.34509804)=(0.49100000,0.64150000,0.74850000)
 rgb(0.34901961)=(0.50170000,0.64840000,0.75230000)
 rgb(0.35294118)=(0.51250000,0.65520000,0.75590000)
 rgb(0.35686275)=(0.52330000,0.66200000,0.75930000)
 rgb(0.36078431)=(0.53410000,0.66850000,0.76260000)
 rgb(0.36470588)=(0.54490000,0.67490000,0.76560000)
 rgb(0.36862745)=(0.55580000,0.68120000,0.76850000)
 rgb(0.37254902)=(0.56660000,0.68740000,0.77120000)
 rgb(0.37647059)=(0.57750000,0.69330000,0.77370000)
 rgb(0.38039216)=(0.58830000,0.69910000,0.77600000)
 rgb(0.38431373)=(0.59900000,0.70470000,0.77800000)
 rgb(0.38823529)=(0.60980000,0.71010000,0.77980000)
 rgb(0.39215686)=(0.62040000,0.71530000,0.78140000)
 rgb(0.39607843)=(0.63100000,0.72030000,0.78260000)
 rgb(0.40000000)=(0.64140000,0.72500000,0.78360000)
 rgb(0.40392157)=(0.65180000,0.72950000,0.78440000)
 rgb(0.40784314)=(0.66200000,0.73380000,0.78480000)
 rgb(0.41176471)=(0.67210000,0.73780000,0.78500000)
 rgb(0.41568627)=(0.68200000,0.74150000,0.78480000)
 rgb(0.41960784)=(0.69170000,0.74490000,0.78430000)
 rgb(0.42352941)=(0.70130000,0.74810000,0.78360000)
 rgb(0.42745098)=(0.71060000,0.75100000,0.78250000)
 rgb(0.43137255)=(0.71980000,0.75360000,0.78100000)
 rgb(0.43529412)=(0.72870000,0.75580000,0.77930000)
 rgb(0.43921569)=(0.73730000,0.75780000,0.77720000)
 rgb(0.44313725)=(0.74570000,0.75940000,0.77480000)
 rgb(0.44705882)=(0.75380000,0.76080000,0.77200000)
 rgb(0.45098039)=(0.76160000,0.76180000,0.76900000)
 rgb(0.45490196)=(0.76910000,0.76250000,0.76560000)
 rgb(0.45882353)=(0.77640000,0.76280000,0.76190000)
 rgb(0.46274510)=(0.78320000,0.76290000,0.75780000)
 rgb(0.46666667)=(0.78980000,0.76260000,0.75350000)
 rgb(0.47058824)=(0.79610000,0.76200000,0.74880000)
 rgb(0.47450980)=(0.80200000,0.76110000,0.74390000)
 rgb(0.47843137)=(0.80760000,0.75980000,0.73870000)
 rgb(0.48235294)=(0.81280000,0.75830000,0.73320000)
 rgb(0.48627451)=(0.81770000,0.75650000,0.72740000)
 rgb(0.49019608)=(0.82220000,0.75430000,0.72130000)
 rgb(0.49411765)=(0.82640000,0.75190000,0.71510000)
 rgb(0.49803922)=(0.83030000,0.74920000,0.70860000)
 rgb(0.50196078)=(0.83390000,0.74630000,0.70180000)
 rgb(0.50588235)=(0.83710000,0.74310000,0.69490000)
 rgb(0.50980392)=(0.84000000,0.73960000,0.68770000)
 rgb(0.51372549)=(0.84250000,0.73590000,0.68040000)
 rgb(0.51764706)=(0.84480000,0.73200000,0.67280000)
 rgb(0.52156863)=(0.84670000,0.72780000,0.66520000)
 rgb(0.52549020)=(0.84840000,0.72340000,0.65730000)
 rgb(0.52941176)=(0.84970000,0.71890000,0.64930000)
 rgb(0.53333333)=(0.85080000,0.71410000,0.64120000)
 rgb(0.53725490)=(0.85160000,0.70920000,0.63300000)
 rgb(0.54117647)=(0.85220000,0.70410000,0.62460000)
 rgb(0.54509804)=(0.85250000,0.69880000,0.61620000)
 rgb(0.54901961)=(0.85260000,0.69340000,0.60760000)
 rgb(0.55294118)=(0.85240000,0.68780000,0.59900000)
 rgb(0.55686275)=(0.85200000,0.68210000,0.59030000)
 rgb(0.56078431)=(0.85140000,0.67620000,0.58150000)
 rgb(0.56470588)=(0.85050000,0.67030000,0.57270000)
 rgb(0.56862745)=(0.84950000,0.66420000,0.56380000)
 rgb(0.57254902)=(0.84830000,0.65800000,0.55490000)
 rgb(0.57647059)=(0.84690000,0.65170000,0.54590000)
 rgb(0.58039216)=(0.84530000,0.64530000,0.53690000)
 rgb(0.58431373)=(0.84350000,0.63880000,0.52780000)
 rgb(0.58823529)=(0.84160000,0.63220000,0.51880000)
 rgb(0.59215686)=(0.83950000,0.62550000,0.50970000)
 rgb(0.59607843)=(0.83720000,0.61880000,0.50050000)
 rgb(0.60000000)=(0.83480000,0.61190000,0.49140000)
 rgb(0.60392157)=(0.83230000,0.60500000,0.48230000)
 rgb(0.60784314)=(0.82960000,0.59800000,0.47310000)
 rgb(0.61176471)=(0.82670000,0.59090000,0.46400000)
 rgb(0.61568627)=(0.82370000,0.58380000,0.45480000)
 rgb(0.61960784)=(0.82060000,0.57650000,0.44570000)
 rgb(0.62352941)=(0.81730000,0.56920000,0.43660000)
 rgb(0.62745098)=(0.81390000,0.56190000,0.42750000)
 rgb(0.63137255)=(0.81040000,0.55450000,0.41830000)
 rgb(0.63529412)=(0.80670000,0.54690000,0.40930000)
 rgb(0.63921569)=(0.80290000,0.53940000,0.40020000)
 rgb(0.64313725)=(0.79890000,0.53170000,0.39110000)
 rgb(0.64705882)=(0.79480000,0.52400000,0.38210000)
 rgb(0.65098039)=(0.79060000,0.51620000,0.37310000)
 rgb(0.65490196)=(0.78620000,0.50840000,0.36420000)
 rgb(0.65882353)=(0.78170000,0.50040000,0.35530000)
 rgb(0.66274510)=(0.77710000,0.49240000,0.34640000)
 rgb(0.66666667)=(0.77230000,0.48440000,0.33760000)
 rgb(0.67058824)=(0.76730000,0.47620000,0.32880000)
 rgb(0.67450980)=(0.76220000,0.46800000,0.32010000)
 rgb(0.67843137)=(0.75700000,0.45970000,0.31140000)
 rgb(0.68235294)=(0.75160000,0.45140000,0.30280000)
 rgb(0.68627451)=(0.74600000,0.44300000,0.29430000)
 rgb(0.69019608)=(0.74030000,0.43450000,0.28590000)
 rgb(0.69411765)=(0.73450000,0.42600000,0.27760000)
 rgb(0.69803922)=(0.72850000,0.41740000,0.26940000)
 rgb(0.70196078)=(0.72230000,0.40880000,0.26130000)
 rgb(0.70588235)=(0.71600000,0.40010000,0.25340000)
 rgb(0.70980392)=(0.70960000,0.39140000,0.24550000)
 rgb(0.71372549)=(0.70300000,0.38260000,0.23780000)
 rgb(0.71764706)=(0.69630000,0.37380000,0.23030000)
 rgb(0.72156863)=(0.68940000,0.36500000,0.22290000)
 rgb(0.72549020)=(0.68250000,0.35620000,0.21570000)
 rgb(0.72941176)=(0.67540000,0.34740000,0.20870000)
 rgb(0.73333333)=(0.66820000,0.33850000,0.20190000)
 rgb(0.73725490)=(0.66090000,0.32970000,0.19530000)
 rgb(0.74117647)=(0.65350000,0.32090000,0.18900000)
 rgb(0.74509804)=(0.64600000,0.31220000,0.18290000)
 rgb(0.74901961)=(0.63850000,0.30350000,0.17700000)
 rgb(0.75294118)=(0.63090000,0.29490000,0.17140000)
 rgb(0.75686275)=(0.62330000,0.28630000,0.16600000)
 rgb(0.76078431)=(0.61560000,0.27790000,0.16100000)
 rgb(0.76470588)=(0.60790000,0.26950000,0.15620000)
 rgb(0.76862745)=(0.60020000,0.26120000,0.15160000)
 rgb(0.77254902)=(0.59250000,0.25310000,0.14750000)
 rgb(0.77647059)=(0.58490000,0.24500000,0.14350000)
 rgb(0.78039216)=(0.57720000,0.23720000,0.13990000)
 rgb(0.78431373)=(0.56960000,0.22950000,0.13660000)
 rgb(0.78823529)=(0.56210000,0.22190000,0.13360000)
 rgb(0.79215686)=(0.55460000,0.21450000,0.13080000)
 rgb(0.79607843)=(0.54720000,0.20730000,0.12840000)
 rgb(0.80000000)=(0.53990000,0.20020000,0.12620000)
 rgb(0.80392157)=(0.53270000,0.19340000,0.12430000)
 rgb(0.80784314)=(0.52560000,0.18680000,0.12270000)
 rgb(0.81176471)=(0.51860000,0.18030000,0.12140000)
 rgb(0.81568627)=(0.51170000,0.17410000,0.12030000)
 rgb(0.81960784)=(0.50500000,0.16810000,0.11950000)
 rgb(0.82352941)=(0.49840000,0.16230000,0.11890000)
 rgb(0.82745098)=(0.49190000,0.15670000,0.11850000)
 rgb(0.83137255)=(0.48560000,0.15130000,0.11840000)
 rgb(0.83529412)=(0.47940000,0.14620000,0.11840000)
 rgb(0.83921569)=(0.47330000,0.14130000,0.11870000)
 rgb(0.84313725)=(0.46740000,0.13650000,0.11910000)
 rgb(0.84705882)=(0.46160000,0.13200000,0.11970000)
 rgb(0.85098039)=(0.45600000,0.12770000,0.12050000)
 rgb(0.85490196)=(0.45050000,0.12360000,0.12150000)
 rgb(0.85882353)=(0.44510000,0.11980000,0.12260000)
 rgb(0.86274510)=(0.43990000,0.11620000,0.12390000)
 rgb(0.86666667)=(0.43480000,0.11280000,0.12540000)
 rgb(0.87058824)=(0.42980000,0.10960000,0.12690000)
 rgb(0.87450980)=(0.42500000,0.10650000,0.12860000)
 rgb(0.87843137)=(0.42030000,0.10370000,0.13050000)
 rgb(0.88235294)=(0.41570000,0.10110000,0.13240000)
 rgb(0.88627451)=(0.41120000,0.09870000,0.13450000)
 rgb(0.89019608)=(0.40680000,0.09640000,0.13660000)
 rgb(0.89411765)=(0.40260000,0.09440000,0.13890000)
 rgb(0.89803922)=(0.39840000,0.09250000,0.14130000)
 rgb(0.90196078)=(0.39430000,0.09090000,0.14380000)
 rgb(0.90588235)=(0.39040000,0.08940000,0.14640000)
 rgb(0.90980392)=(0.38650000,0.08810000,0.14900000)
 rgb(0.91372549)=(0.38270000,0.08700000,0.15180000)
 rgb(0.91764706)=(0.37900000,0.08590000,0.15470000)
 rgb(0.92156863)=(0.37530000,0.08510000,0.15760000)
 rgb(0.92549020)=(0.37170000,0.08450000,0.16060000)
 rgb(0.92941176)=(0.36820000,0.08410000,0.16380000)
 rgb(0.93333333)=(0.36480000,0.08370000,0.16700000)
 rgb(0.93725490)=(0.36140000,0.08350000,0.17030000)
 rgb(0.94117647)=(0.35810000,0.08350000,0.17360000)
 rgb(0.94509804)=(0.35480000,0.08360000,0.17710000)
 rgb(0.94901961)=(0.35160000,0.08390000,0.18060000)
 rgb(0.95294118)=(0.34840000,0.08430000,0.18420000)
 rgb(0.95686275)=(0.34530000,0.08480000,0.18790000)
 rgb(0.96078431)=(0.34220000,0.08540000,0.19170000)
 rgb(0.96470588)=(0.33910000,0.08620000,0.19550000)
 rgb(0.96862745)=(0.33610000,0.08720000,0.19940000)
 rgb(0.97254902)=(0.33310000,0.08820000,0.20350000)
 rgb(0.97647059)=(0.33010000,0.08940000,0.20750000)
 rgb(0.98039216)=(0.32720000,0.09070000,0.21170000)
 rgb(0.98431373)=(0.32420000,0.09210000,0.21600000)
 rgb(0.98823529)=(0.32130000,0.09360000,0.22030000)
 rgb(0.99215686)=(0.31840000,0.09530000,0.22470000)
 rgb(0.99607843)=(0.31550000,0.09700000,0.22920000)
 rgb(1.00000000)=(0.31260000,0.09890000,0.23380000)},
 }

 \pgfplotsset{
 colormap={plots2}{rgb(0.00000000)=(0.00130000,0.06980000,0.37950000)
 rgb(0.00392157)=(0.00240000,0.07650000,0.38350000)
 rgb(0.00784314)=(0.00330000,0.08310000,0.38750000)
 rgb(0.01176471)=(0.00410000,0.08960000,0.39150000)
 rgb(0.01568627)=(0.00490000,0.09590000,0.39550000)
 rgb(0.01960784)=(0.00560000,0.10230000,0.39940000)
 rgb(0.02352941)=(0.00620000,0.10850000,0.40340000)
 rgb(0.02745098)=(0.00670000,0.11470000,0.40730000)
 rgb(0.03137255)=(0.00710000,0.12080000,0.41130000)
 rgb(0.03529412)=(0.00750000,0.12700000,0.41520000)
 rgb(0.03921569)=(0.00780000,0.13310000,0.41920000)
 rgb(0.04313725)=(0.00810000,0.13910000,0.42310000)
 rgb(0.04705882)=(0.00840000,0.14520000,0.42700000)
 rgb(0.05098039)=(0.00860000,0.15110000,0.43090000)
 rgb(0.05490196)=(0.00880000,0.15710000,0.43480000)
 rgb(0.05882353)=(0.00890000,0.16320000,0.43870000)
 rgb(0.06274510)=(0.00910000,0.16910000,0.44260000)
 rgb(0.06666667)=(0.00920000,0.17510000,0.44650000)
 rgb(0.07058824)=(0.00930000,0.18110000,0.45030000)
 rgb(0.07450980)=(0.00940000,0.18710000,0.45420000)
 rgb(0.07843137)=(0.00940000,0.19300000,0.45810000)
 rgb(0.08235294)=(0.00950000,0.19900000,0.46200000)
 rgb(0.08627451)=(0.00960000,0.20500000,0.46580000)
 rgb(0.09019608)=(0.00960000,0.21100000,0.46970000)
 rgb(0.09411765)=(0.00970000,0.21700000,0.47360000)
 rgb(0.09803922)=(0.00970000,0.22310000,0.47750000)
 rgb(0.10196078)=(0.00980000,0.22910000,0.48140000)
 rgb(0.10588235)=(0.00990000,0.23520000,0.48520000)
 rgb(0.10980392)=(0.01000000,0.24130000,0.48920000)
 rgb(0.11372549)=(0.01010000,0.24740000,0.49310000)
 rgb(0.11764706)=(0.01030000,0.25350000,0.49700000)
 rgb(0.12156863)=(0.01050000,0.25970000,0.50100000)
 rgb(0.12549020)=(0.01080000,0.26590000,0.50490000)
 rgb(0.12941176)=(0.01120000,0.27200000,0.50890000)
 rgb(0.13333333)=(0.01170000,0.27830000,0.51290000)
 rgb(0.13725490)=(0.01230000,0.28460000,0.51700000)
 rgb(0.14117647)=(0.01290000,0.29090000,0.52100000)
 rgb(0.14509804)=(0.01380000,0.29720000,0.52510000)
 rgb(0.14901961)=(0.01480000,0.30360000,0.52920000)
 rgb(0.15294118)=(0.01610000,0.31000000,0.53330000)
 rgb(0.15686275)=(0.01770000,0.31650000,0.53750000)
 rgb(0.16078431)=(0.01960000,0.32300000,0.54170000)
 rgb(0.16470588)=(0.02190000,0.32960000,0.54590000)
 rgb(0.16862745)=(0.02470000,0.33610000,0.55020000)
 rgb(0.17254902)=(0.02800000,0.34280000,0.55450000)
 rgb(0.17647059)=(0.03200000,0.34950000,0.55890000)
 rgb(0.18039216)=(0.03680000,0.35630000,0.56330000)
 rgb(0.18431373)=(0.04220000,0.36320000,0.56780000)
 rgb(0.18823529)=(0.04800000,0.37010000,0.57230000)
 rgb(0.19215686)=(0.05430000,0.37710000,0.57690000)
 rgb(0.19607843)=(0.06100000,0.38410000,0.58160000)
 rgb(0.20000000)=(0.06810000,0.39130000,0.58630000)
 rgb(0.20392157)=(0.07550000,0.39850000,0.59100000)
 rgb(0.20784314)=(0.08320000,0.40570000,0.59590000)
 rgb(0.21176471)=(0.09140000,0.41310000,0.60080000)
 rgb(0.21568627)=(0.09980000,0.42050000,0.60570000)
 rgb(0.21960784)=(0.10860000,0.42800000,0.61070000)
 rgb(0.22352941)=(0.11770000,0.43560000,0.61580000)
 rgb(0.22745098)=(0.12700000,0.44320000,0.62090000)
 rgb(0.23137255)=(0.13670000,0.45090000,0.62610000)
 rgb(0.23529412)=(0.14660000,0.45860000,0.63130000)
 rgb(0.23921569)=(0.15680000,0.46650000,0.63660000)
 rgb(0.24313725)=(0.16720000,0.47430000,0.64190000)
 rgb(0.24705882)=(0.17780000,0.48220000,0.64720000)
 rgb(0.25098039)=(0.18860000,0.49020000,0.65260000)
 rgb(0.25490196)=(0.19960000,0.49820000,0.65800000)
 rgb(0.25882353)=(0.21080000,0.50620000,0.66350000)
 rgb(0.26274510)=(0.22210000,0.51430000,0.66890000)
 rgb(0.26666667)=(0.23360000,0.52230000,0.67440000)
 rgb(0.27058824)=(0.24520000,0.53040000,0.67990000)
 rgb(0.27450980)=(0.25700000,0.53850000,0.68540000)
 rgb(0.27843137)=(0.26890000,0.54660000,0.69090000)
 rgb(0.28235294)=(0.28080000,0.55470000,0.69640000)
 rgb(0.28627451)=(0.29290000,0.56280000,0.70190000)
 rgb(0.29019608)=(0.30500000,0.57090000,0.70740000)
 rgb(0.29411765)=(0.31720000,0.57900000,0.71300000)
 rgb(0.29803922)=(0.32940000,0.58710000,0.71840000)
 rgb(0.30196078)=(0.34170000,0.59510000,0.72390000)
 rgb(0.30588235)=(0.35410000,0.60320000,0.72940000)
 rgb(0.30980392)=(0.36650000,0.61120000,0.73490000)
 rgb(0.31372549)=(0.37890000,0.61920000,0.74030000)
 rgb(0.31764706)=(0.39130000,0.62720000,0.74580000)
 rgb(0.32156863)=(0.40380000,0.63510000,0.75120000)
 rgb(0.32549020)=(0.41620000,0.64300000,0.75660000)
 rgb(0.32941176)=(0.42870000,0.65100000,0.76200000)
 rgb(0.33333333)=(0.44120000,0.65880000,0.76730000)
 rgb(0.33725490)=(0.45370000,0.66670000,0.77270000)
 rgb(0.34117647)=(0.46620000,0.67450000,0.77800000)
 rgb(0.34509804)=(0.47870000,0.68230000,0.78340000)
 rgb(0.34901961)=(0.49120000,0.69010000,0.78870000)
 rgb(0.35294118)=(0.50370000,0.69790000,0.79400000)
 rgb(0.35686275)=(0.51620000,0.70570000,0.79930000)
 rgb(0.36078431)=(0.52870000,0.71340000,0.80450000)
 rgb(0.36470588)=(0.54110000,0.72110000,0.80980000)
 rgb(0.36862745)=(0.55360000,0.72880000,0.81500000)
 rgb(0.37254902)=(0.56610000,0.73640000,0.82020000)
 rgb(0.37647059)=(0.57860000,0.74410000,0.82540000)
 rgb(0.38039216)=(0.59100000,0.75170000,0.83060000)
 rgb(0.38431373)=(0.60350000,0.75930000,0.83580000)
 rgb(0.38823529)=(0.61590000,0.76690000,0.84090000)
 rgb(0.39215686)=(0.62840000,0.77450000,0.84610000)
 rgb(0.39607843)=(0.64080000,0.78200000,0.85110000)
 rgb(0.40000000)=(0.65320000,0.78950000,0.85620000)
 rgb(0.40392157)=(0.66560000,0.79690000,0.86120000)
 rgb(0.40784314)=(0.67810000,0.80440000,0.86620000)
 rgb(0.41176471)=(0.69050000,0.81170000,0.87110000)
 rgb(0.41568627)=(0.70290000,0.81900000,0.87590000)
 rgb(0.41960784)=(0.71530000,0.82630000,0.88060000)
 rgb(0.42352941)=(0.72760000,0.83340000,0.88510000)
 rgb(0.42745098)=(0.74000000,0.84050000,0.88960000)
 rgb(0.43137255)=(0.75240000,0.84740000,0.89380000)
 rgb(0.43529412)=(0.76470000,0.85410000,0.89780000)
 rgb(0.43921569)=(0.77690000,0.86070000,0.90160000)
 rgb(0.44313725)=(0.78910000,0.86700000,0.90500000)
 rgb(0.44705882)=(0.80120000,0.87300000,0.90800000)
 rgb(0.45098039)=(0.81310000,0.87870000,0.91070000)
 rgb(0.45490196)=(0.82490000,0.88410000,0.91280000)
 rgb(0.45882353)=(0.83640000,0.88890000,0.91430000)
 rgb(0.46274510)=(0.84760000,0.89330000,0.91520000)
 rgb(0.46666667)=(0.85850000,0.89710000,0.91540000)
 rgb(0.47058824)=(0.86890000,0.90020000,0.91480000)
 rgb(0.47450980)=(0.87870000,0.90260000,0.91340000)
 rgb(0.47843137)=(0.88800000,0.90430000,0.91120000)
 rgb(0.48235294)=(0.89650000,0.90520000,0.90800000)
 rgb(0.48627451)=(0.90420000,0.90520000,0.90400000)
 rgb(0.49019608)=(0.91120000,0.90440000,0.89910000)
 rgb(0.49411765)=(0.91720000,0.90280000,0.89340000)
 rgb(0.49803922)=(0.92230000,0.90040000,0.88690000)
 rgb(0.50196078)=(0.92650000,0.89720000,0.87970000)
 rgb(0.50588235)=(0.92980000,0.89330000,0.87180000)
 rgb(0.50980392)=(0.93220000,0.88870000,0.86340000)
 rgb(0.51372549)=(0.93390000,0.88360000,0.85450000)
 rgb(0.51764706)=(0.93480000,0.87790000,0.84520000)
 rgb(0.52156863)=(0.93500000,0.87180000,0.83550000)
 rgb(0.52549020)=(0.93460000,0.86530000,0.82560000)
 rgb(0.52941176)=(0.93380000,0.85850000,0.81540000)
 rgb(0.53333333)=(0.93240000,0.85150000,0.80510000)
 rgb(0.53725490)=(0.93070000,0.84420000,0.79470000)
 rgb(0.54117647)=(0.92860000,0.83680000,0.78420000)
 rgb(0.54509804)=(0.92630000,0.82920000,0.77360000)
 rgb(0.54901961)=(0.92380000,0.82150000,0.76300000)
 rgb(0.55294118)=(0.92100000,0.81380000,0.75230000)
 rgb(0.55686275)=(0.91810000,0.80600000,0.74170000)
 rgb(0.56078431)=(0.91520000,0.79820000,0.73100000)
 rgb(0.56470588)=(0.91210000,0.79030000,0.72040000)
 rgb(0.56862745)=(0.90890000,0.78240000,0.70980000)
 rgb(0.57254902)=(0.90570000,0.77450000,0.69920000)
 rgb(0.57647059)=(0.90250000,0.76670000,0.68860000)
 rgb(0.58039216)=(0.89920000,0.75880000,0.67810000)
 rgb(0.58431373)=(0.89600000,0.75100000,0.66760000)
 rgb(0.58823529)=(0.89270000,0.74310000,0.65710000)
 rgb(0.59215686)=(0.88940000,0.73530000,0.64670000)
 rgb(0.59607843)=(0.88610000,0.72760000,0.63630000)
 rgb(0.60000000)=(0.88280000,0.71980000,0.62590000)
 rgb(0.60392157)=(0.87960000,0.71210000,0.61560000)
 rgb(0.60784314)=(0.87630000,0.70440000,0.60540000)
 rgb(0.61176471)=(0.87300000,0.69680000,0.59510000)
 rgb(0.61568627)=(0.86980000,0.68910000,0.58500000)
 rgb(0.61960784)=(0.86660000,0.68150000,0.57480000)
 rgb(0.62352941)=(0.86330000,0.67400000,0.56470000)
 rgb(0.62745098)=(0.86010000,0.66650000,0.55470000)
 rgb(0.63137255)=(0.85690000,0.65900000,0.54470000)
 rgb(0.63529412)=(0.85370000,0.65150000,0.53480000)
 rgb(0.63921569)=(0.85060000,0.64410000,0.52480000)
 rgb(0.64313725)=(0.84740000,0.63670000,0.51500000)
 rgb(0.64705882)=(0.84430000,0.62930000,0.50510000)
 rgb(0.65098039)=(0.84110000,0.62200000,0.49540000)
 rgb(0.65490196)=(0.83800000,0.61470000,0.48560000)
 rgb(0.65882353)=(0.83490000,0.60740000,0.47590000)
 rgb(0.66274510)=(0.83180000,0.60010000,0.46630000)
 rgb(0.66666667)=(0.82870000,0.59290000,0.45670000)
 rgb(0.67058824)=(0.82560000,0.58580000,0.44710000)
 rgb(0.67450980)=(0.82260000,0.57860000,0.43760000)
 rgb(0.67843137)=(0.81950000,0.57150000,0.42810000)
 rgb(0.68235294)=(0.81650000,0.56440000,0.41870000)
 rgb(0.68627451)=(0.81350000,0.55730000,0.40930000)
 rgb(0.69019608)=(0.81040000,0.55030000,0.39990000)
 rgb(0.69411765)=(0.80740000,0.54330000,0.39060000)
 rgb(0.69803922)=(0.80440000,0.53630000,0.38130000)
 rgb(0.70196078)=(0.80150000,0.52930000,0.37200000)
 rgb(0.70588235)=(0.79850000,0.52240000,0.36280000)
 rgb(0.70980392)=(0.79550000,0.51550000,0.35370000)
 rgb(0.71372549)=(0.79250000,0.50860000,0.34450000)
 rgb(0.71764706)=(0.78960000,0.50170000,0.33540000)
 rgb(0.72156863)=(0.78660000,0.49480000,0.32630000)
 rgb(0.72549020)=(0.78370000,0.48800000,0.31730000)
 rgb(0.72941176)=(0.78070000,0.48110000,0.30830000)
 rgb(0.73333333)=(0.77770000,0.47430000,0.29930000)
 rgb(0.73725490)=(0.77480000,0.46750000,0.29040000)
 rgb(0.74117647)=(0.77180000,0.46060000,0.28140000)
 rgb(0.74509804)=(0.76880000,0.45380000,0.27250000)
 rgb(0.74901961)=(0.76580000,0.44690000,0.26360000)
 rgb(0.75294118)=(0.76270000,0.44010000,0.25480000)
 rgb(0.75686275)=(0.75960000,0.43310000,0.24590000)
 rgb(0.76078431)=(0.75650000,0.42620000,0.23700000)
 rgb(0.76470588)=(0.75330000,0.41920000,0.22820000)
 rgb(0.76862745)=(0.75010000,0.41220000,0.21930000)
 rgb(0.77254902)=(0.74670000,0.40500000,0.21050000)
 rgb(0.77647059)=(0.74320000,0.39780000,0.20160000)
 rgb(0.78039216)=(0.73970000,0.39050000,0.19270000)
 rgb(0.78431373)=(0.73590000,0.38310000,0.18390000)
 rgb(0.78823529)=(0.73200000,0.37550000,0.17500000)
 rgb(0.79215686)=(0.72790000,0.36770000,0.16600000)
 rgb(0.79607843)=(0.72350000,0.35990000,0.15710000)
 rgb(0.80000000)=(0.71890000,0.35180000,0.14820000)
 rgb(0.80392157)=(0.71400000,0.34350000,0.13930000)
 rgb(0.80784314)=(0.70880000,0.33500000,0.13050000)
 rgb(0.81176471)=(0.70330000,0.32640000,0.12150000)
 rgb(0.81568627)=(0.69740000,0.31750000,0.11280000)
 rgb(0.81960784)=(0.69120000,0.30850000,0.10410000)
 rgb(0.82352941)=(0.68470000,0.29930000,0.09560000)
 rgb(0.82745098)=(0.67770000,0.28990000,0.08740000)
 rgb(0.83137255)=(0.67050000,0.28050000,0.07920000)
 rgb(0.83529412)=(0.66290000,0.27100000,0.07150000)
 rgb(0.83921569)=(0.65500000,0.26150000,0.06410000)
 rgb(0.84313725)=(0.64700000,0.25210000,0.05710000)
 rgb(0.84705882)=(0.63870000,0.24270000,0.05060000)
 rgb(0.85098039)=(0.63030000,0.23350000,0.04480000)
 rgb(0.85490196)=(0.62170000,0.22440000,0.03940000)
 rgb(0.85882353)=(0.61310000,0.21570000,0.03480000)
 rgb(0.86274510)=(0.60450000,0.20710000,0.03110000)
 rgb(0.86666667)=(0.59590000,0.19870000,0.02820000)
 rgb(0.87058824)=(0.58740000,0.19070000,0.02600000)
 rgb(0.87450980)=(0.57890000,0.18290000,0.02440000)
 rgb(0.87843137)=(0.57050000,0.17540000,0.02330000)
 rgb(0.88235294)=(0.56230000,0.16820000,0.02250000)
 rgb(0.88627451)=(0.55410000,0.16120000,0.02210000)
 rgb(0.89019608)=(0.54600000,0.15440000,0.02190000)
 rgb(0.89411765)=(0.53800000,0.14790000,0.02170000)
 rgb(0.89803922)=(0.53020000,0.14150000,0.02170000)
 rgb(0.90196078)=(0.52240000,0.13530000,0.02180000)
 rgb(0.90588235)=(0.51480000,0.12920000,0.02200000)
 rgb(0.90980392)=(0.50720000,0.12330000,0.02220000)
 rgb(0.91372549)=(0.49970000,0.11750000,0.02250000)
 rgb(0.91764706)=(0.49230000,0.11180000,0.02280000)
 rgb(0.92156863)=(0.48500000,0.10620000,0.02310000)
 rgb(0.92549020)=(0.47780000,0.10060000,0.02350000)
 rgb(0.92941176)=(0.47060000,0.09520000,0.02390000)
 rgb(0.93333333)=(0.46350000,0.08970000,0.02430000)
 rgb(0.93725490)=(0.45650000,0.08430000,0.02480000)
 rgb(0.94117647)=(0.44950000,0.07870000,0.02520000)
 rgb(0.94509804)=(0.44260000,0.07340000,0.02560000)
 rgb(0.94901961)=(0.43570000,0.06790000,0.02610000)
 rgb(0.95294118)=(0.42890000,0.06240000,0.02650000)
 rgb(0.95686275)=(0.42210000,0.05680000,0.02700000)
 rgb(0.96078431)=(0.41540000,0.05110000,0.02740000)
 rgb(0.96470588)=(0.40880000,0.04540000,0.02780000)
 rgb(0.96862745)=(0.40210000,0.03940000,0.02820000)
 rgb(0.97254902)=(0.39560000,0.03340000,0.02860000)
 rgb(0.97647059)=(0.38900000,0.02780000,0.02890000)
 rgb(0.98039216)=(0.38250000,0.02260000,0.02930000)
 rgb(0.98431373)=(0.37600000,0.01760000,0.02960000)
 rgb(0.98823529)=(0.36960000,0.01290000,0.02990000)
 rgb(0.99215686)=(0.36320000,0.00820000,0.03010000)
 rgb(0.99607843)=(0.35680000,0.00400000,0.03030000)
 rgb(1.00000000)=(0.35040000,0.00010000,0.03050000)},
 }
 \pgfplotsset{
 colormap={plots2}{rgb(0.00000000)=(0.00130000,0.06980000,0.37950000)
 rgb(0.00392157)=(0.00240000,0.07650000,0.38350000)
 rgb(0.00784314)=(0.00330000,0.08310000,0.38750000)
 rgb(0.01176471)=(0.00410000,0.08960000,0.39150000)
 rgb(0.01568627)=(0.00490000,0.09590000,0.39550000)
 rgb(0.01960784)=(0.00560000,0.10230000,0.39940000)
 rgb(0.02352941)=(0.00620000,0.10850000,0.40340000)
 rgb(0.02745098)=(0.00670000,0.11470000,0.40730000)
 rgb(0.03137255)=(0.00710000,0.12080000,0.41130000)
 rgb(0.03529412)=(0.00750000,0.12700000,0.41520000)
 rgb(0.03921569)=(0.00780000,0.13310000,0.41920000)
 rgb(0.04313725)=(0.00810000,0.13910000,0.42310000)
 rgb(0.04705882)=(0.00840000,0.14520000,0.42700000)
 rgb(0.05098039)=(0.00860000,0.15110000,0.43090000)
 rgb(0.05490196)=(0.00880000,0.15710000,0.43480000)
 rgb(0.05882353)=(0.00890000,0.16320000,0.43870000)
 rgb(0.06274510)=(0.00910000,0.16910000,0.44260000)
 rgb(0.06666667)=(0.00920000,0.17510000,0.44650000)
 rgb(0.07058824)=(0.00930000,0.18110000,0.45030000)
 rgb(0.07450980)=(0.00940000,0.18710000,0.45420000)
 rgb(0.07843137)=(0.00940000,0.19300000,0.45810000)
 rgb(0.08235294)=(0.00950000,0.19900000,0.46200000)
 rgb(0.08627451)=(0.00960000,0.20500000,0.46580000)
 rgb(0.09019608)=(0.00960000,0.21100000,0.46970000)
 rgb(0.09411765)=(0.00970000,0.21700000,0.47360000)
 rgb(0.09803922)=(0.00970000,0.22310000,0.47750000)
 rgb(0.10196078)=(0.00980000,0.22910000,0.48140000)
 rgb(0.10588235)=(0.00990000,0.23520000,0.48520000)
 rgb(0.10980392)=(0.01000000,0.24130000,0.48920000)
 rgb(0.11372549)=(0.01010000,0.24740000,0.49310000)
 rgb(0.11764706)=(0.01030000,0.25350000,0.49700000)
 rgb(0.12156863)=(0.01050000,0.25970000,0.50100000)
 rgb(0.12549020)=(0.01080000,0.26590000,0.50490000)
 rgb(0.12941176)=(0.01120000,0.27200000,0.50890000)
 rgb(0.13333333)=(0.01170000,0.27830000,0.51290000)
 rgb(0.13725490)=(0.01230000,0.28460000,0.51700000)
 rgb(0.14117647)=(0.01290000,0.29090000,0.52100000)
 rgb(0.14509804)=(0.01380000,0.29720000,0.52510000)
 rgb(0.14901961)=(0.01480000,0.30360000,0.52920000)
 rgb(0.15294118)=(0.01610000,0.31000000,0.53330000)
 rgb(0.15686275)=(0.01770000,0.31650000,0.53750000)
 rgb(0.16078431)=(0.01960000,0.32300000,0.54170000)
 rgb(0.16470588)=(0.02190000,0.32960000,0.54590000)
 rgb(0.16862745)=(0.02470000,0.33610000,0.55020000)
 rgb(0.17254902)=(0.02800000,0.34280000,0.55450000)
 rgb(0.17647059)=(0.03200000,0.34950000,0.55890000)
 rgb(0.18039216)=(0.03680000,0.35630000,0.56330000)
 rgb(0.18431373)=(0.04220000,0.36320000,0.56780000)
 rgb(0.18823529)=(0.04800000,0.37010000,0.57230000)
 rgb(0.19215686)=(0.05430000,0.37710000,0.57690000)
 rgb(0.19607843)=(0.06100000,0.38410000,0.58160000)
 rgb(0.20000000)=(0.06810000,0.39130000,0.58630000)
 rgb(0.20392157)=(0.07550000,0.39850000,0.59100000)
 rgb(0.20784314)=(0.08320000,0.40570000,0.59590000)
 rgb(0.21176471)=(0.09140000,0.41310000,0.60080000)
 rgb(0.21568627)=(0.09980000,0.42050000,0.60570000)
 rgb(0.21960784)=(0.10860000,0.42800000,0.61070000)
 rgb(0.22352941)=(0.11770000,0.43560000,0.61580000)
 rgb(0.22745098)=(0.12700000,0.44320000,0.62090000)
 rgb(0.23137255)=(0.13670000,0.45090000,0.62610000)
 rgb(0.23529412)=(0.14660000,0.45860000,0.63130000)
 rgb(0.23921569)=(0.15680000,0.46650000,0.63660000)
 rgb(0.24313725)=(0.16720000,0.47430000,0.64190000)
 rgb(0.24705882)=(0.17780000,0.48220000,0.64720000)
 rgb(0.25098039)=(0.18860000,0.49020000,0.65260000)
 rgb(0.25490196)=(0.19960000,0.49820000,0.65800000)
 rgb(0.25882353)=(0.21080000,0.50620000,0.66350000)
 rgb(0.26274510)=(0.22210000,0.51430000,0.66890000)
 rgb(0.26666667)=(0.23360000,0.52230000,0.67440000)
 rgb(0.27058824)=(0.24520000,0.53040000,0.67990000)
 rgb(0.27450980)=(0.25700000,0.53850000,0.68540000)
 rgb(0.27843137)=(0.26890000,0.54660000,0.69090000)
 rgb(0.28235294)=(0.28080000,0.55470000,0.69640000)
 rgb(0.28627451)=(0.29290000,0.56280000,0.70190000)
 rgb(0.29019608)=(0.30500000,0.57090000,0.70740000)
 rgb(0.29411765)=(0.31720000,0.57900000,0.71300000)
 rgb(0.29803922)=(0.32940000,0.58710000,0.71840000)
 rgb(0.30196078)=(0.34170000,0.59510000,0.72390000)
 rgb(0.30588235)=(0.35410000,0.60320000,0.72940000)
 rgb(0.30980392)=(0.36650000,0.61120000,0.73490000)
 rgb(0.31372549)=(0.37890000,0.61920000,0.74030000)
 rgb(0.31764706)=(0.39130000,0.62720000,0.74580000)
 rgb(0.32156863)=(0.40380000,0.63510000,0.75120000)
 rgb(0.32549020)=(0.41620000,0.64300000,0.75660000)
 rgb(0.32941176)=(0.42870000,0.65100000,0.76200000)
 rgb(0.33333333)=(0.44120000,0.65880000,0.76730000)
 rgb(0.33725490)=(0.45370000,0.66670000,0.77270000)
 rgb(0.34117647)=(0.46620000,0.67450000,0.77800000)
 rgb(0.34509804)=(0.47870000,0.68230000,0.78340000)
 rgb(0.34901961)=(0.49120000,0.69010000,0.78870000)
 rgb(0.35294118)=(0.50370000,0.69790000,0.79400000)
 rgb(0.35686275)=(0.51620000,0.70570000,0.79930000)
 rgb(0.36078431)=(0.52870000,0.71340000,0.80450000)
 rgb(0.36470588)=(0.54110000,0.72110000,0.80980000)
 rgb(0.36862745)=(0.55360000,0.72880000,0.81500000)
 rgb(0.37254902)=(0.56610000,0.73640000,0.82020000)
 rgb(0.37647059)=(0.57860000,0.74410000,0.82540000)
 rgb(0.38039216)=(0.59100000,0.75170000,0.83060000)
 rgb(0.38431373)=(0.60350000,0.75930000,0.83580000)
 rgb(0.38823529)=(0.61590000,0.76690000,0.84090000)
 rgb(0.39215686)=(0.62840000,0.77450000,0.84610000)
 rgb(0.39607843)=(0.64080000,0.78200000,0.85110000)
 rgb(0.40000000)=(0.65320000,0.78950000,0.85620000)
 rgb(0.40392157)=(0.66560000,0.79690000,0.86120000)
 rgb(0.40784314)=(0.67810000,0.80440000,0.86620000)
 rgb(0.41176471)=(0.69050000,0.81170000,0.87110000)
 rgb(0.41568627)=(0.70290000,0.81900000,0.87590000)
 rgb(0.41960784)=(0.71530000,0.82630000,0.88060000)
 rgb(0.42352941)=(0.72760000,0.83340000,0.88510000)
 rgb(0.42745098)=(0.74000000,0.84050000,0.88960000)
 rgb(0.43137255)=(0.75240000,0.84740000,0.89380000)
 rgb(0.43529412)=(0.76470000,0.85410000,0.89780000)
 rgb(0.43921569)=(0.77690000,0.86070000,0.90160000)
 rgb(0.44313725)=(0.78910000,0.86700000,0.90500000)
 rgb(0.44705882)=(0.80120000,0.87300000,0.90800000)
 rgb(0.45098039)=(0.81310000,0.87870000,0.91070000)
 rgb(0.45490196)=(0.82490000,0.88410000,0.91280000)
 rgb(0.45882353)=(0.83640000,0.88890000,0.91430000)
 rgb(0.46274510)=(0.84760000,0.89330000,0.91520000)
 rgb(0.46666667)=(0.85850000,0.89710000,0.91540000)
 rgb(0.47058824)=(0.86890000,0.90020000,0.91480000)
 rgb(0.47450980)=(0.87870000,0.90260000,0.91340000)
 rgb(0.47843137)=(0.88800000,0.90430000,0.91120000)
 rgb(0.48235294)=(0.89650000,0.90520000,0.90800000)
 rgb(0.48627451)=(0.90420000,0.90520000,0.90400000)
 rgb(0.49019608)=(0.91120000,0.90440000,0.89910000)
 rgb(0.49411765)=(0.91720000,0.90280000,0.89340000)
 rgb(0.49803922)=(0.92230000,0.90040000,0.88690000)
 rgb(0.50196078)=(0.92650000,0.89720000,0.87970000)
 rgb(0.50588235)=(0.92980000,0.89330000,0.87180000)
 rgb(0.50980392)=(0.93220000,0.88870000,0.86340000)
 rgb(0.51372549)=(0.93390000,0.88360000,0.85450000)
 rgb(0.51764706)=(0.93480000,0.87790000,0.84520000)
 rgb(0.52156863)=(0.93500000,0.87180000,0.83550000)
 rgb(0.52549020)=(0.93460000,0.86530000,0.82560000)
 rgb(0.52941176)=(0.93380000,0.85850000,0.81540000)
 rgb(0.53333333)=(0.93240000,0.85150000,0.80510000)
 rgb(0.53725490)=(0.93070000,0.84420000,0.79470000)
 rgb(0.54117647)=(0.92860000,0.83680000,0.78420000)
 rgb(0.54509804)=(0.92630000,0.82920000,0.77360000)
 rgb(0.54901961)=(0.92380000,0.82150000,0.76300000)
 rgb(0.55294118)=(0.92100000,0.81380000,0.75230000)
 rgb(0.55686275)=(0.91810000,0.80600000,0.74170000)
 rgb(0.56078431)=(0.91520000,0.79820000,0.73100000)
 rgb(0.56470588)=(0.91210000,0.79030000,0.72040000)
 rgb(0.56862745)=(0.90890000,0.78240000,0.70980000)
 rgb(0.57254902)=(0.90570000,0.77450000,0.69920000)
 rgb(0.57647059)=(0.90250000,0.76670000,0.68860000)
 rgb(0.58039216)=(0.89920000,0.75880000,0.67810000)
 rgb(0.58431373)=(0.89600000,0.75100000,0.66760000)
 rgb(0.58823529)=(0.89270000,0.74310000,0.65710000)
 rgb(0.59215686)=(0.88940000,0.73530000,0.64670000)
 rgb(0.59607843)=(0.88610000,0.72760000,0.63630000)
 rgb(0.60000000)=(0.88280000,0.71980000,0.62590000)
 rgb(0.60392157)=(0.87960000,0.71210000,0.61560000)
 rgb(0.60784314)=(0.87630000,0.70440000,0.60540000)
 rgb(0.61176471)=(0.87300000,0.69680000,0.59510000)
 rgb(0.61568627)=(0.86980000,0.68910000,0.58500000)
 rgb(0.61960784)=(0.86660000,0.68150000,0.57480000)
 rgb(0.62352941)=(0.86330000,0.67400000,0.56470000)
 rgb(0.62745098)=(0.86010000,0.66650000,0.55470000)
 rgb(0.63137255)=(0.85690000,0.65900000,0.54470000)
 rgb(0.63529412)=(0.85370000,0.65150000,0.53480000)
 rgb(0.63921569)=(0.85060000,0.64410000,0.52480000)
 rgb(0.64313725)=(0.84740000,0.63670000,0.51500000)
 rgb(0.64705882)=(0.84430000,0.62930000,0.50510000)
 rgb(0.65098039)=(0.84110000,0.62200000,0.49540000)
 rgb(0.65490196)=(0.83800000,0.61470000,0.48560000)
 rgb(0.65882353)=(0.83490000,0.60740000,0.47590000)
 rgb(0.66274510)=(0.83180000,0.60010000,0.46630000)
 rgb(0.66666667)=(0.82870000,0.59290000,0.45670000)
 rgb(0.67058824)=(0.82560000,0.58580000,0.44710000)
 rgb(0.67450980)=(0.82260000,0.57860000,0.43760000)
 rgb(0.67843137)=(0.81950000,0.57150000,0.42810000)
 rgb(0.68235294)=(0.81650000,0.56440000,0.41870000)
 rgb(0.68627451)=(0.81350000,0.55730000,0.40930000)
 rgb(0.69019608)=(0.81040000,0.55030000,0.39990000)
 rgb(0.69411765)=(0.80740000,0.54330000,0.39060000)
 rgb(0.69803922)=(0.80440000,0.53630000,0.38130000)
 rgb(0.70196078)=(0.80150000,0.52930000,0.37200000)
 rgb(0.70588235)=(0.79850000,0.52240000,0.36280000)
 rgb(0.70980392)=(0.79550000,0.51550000,0.35370000)
 rgb(0.71372549)=(0.79250000,0.50860000,0.34450000)
 rgb(0.71764706)=(0.78960000,0.50170000,0.33540000)
 rgb(0.72156863)=(0.78660000,0.49480000,0.32630000)
 rgb(0.72549020)=(0.78370000,0.48800000,0.31730000)
 rgb(0.72941176)=(0.78070000,0.48110000,0.30830000)
 rgb(0.73333333)=(0.77770000,0.47430000,0.29930000)
 rgb(0.73725490)=(0.77480000,0.46750000,0.29040000)
 rgb(0.74117647)=(0.77180000,0.46060000,0.28140000)
 rgb(0.74509804)=(0.76880000,0.45380000,0.27250000)
 rgb(0.74901961)=(0.76580000,0.44690000,0.26360000)
 rgb(0.75294118)=(0.76270000,0.44010000,0.25480000)
 rgb(0.75686275)=(0.75960000,0.43310000,0.24590000)
 rgb(0.76078431)=(0.75650000,0.42620000,0.23700000)
 rgb(0.76470588)=(0.75330000,0.41920000,0.22820000)
 rgb(0.76862745)=(0.75010000,0.41220000,0.21930000)
 rgb(0.77254902)=(0.74670000,0.40500000,0.21050000)
 rgb(0.77647059)=(0.74320000,0.39780000,0.20160000)
 rgb(0.78039216)=(0.73970000,0.39050000,0.19270000)
 rgb(0.78431373)=(0.73590000,0.38310000,0.18390000)
 rgb(0.78823529)=(0.73200000,0.37550000,0.17500000)
 rgb(0.79215686)=(0.72790000,0.36770000,0.16600000)
 rgb(0.79607843)=(0.72350000,0.35990000,0.15710000)
 rgb(0.80000000)=(0.71890000,0.35180000,0.14820000)
 rgb(0.80392157)=(0.71400000,0.34350000,0.13930000)
 rgb(0.80784314)=(0.70880000,0.33500000,0.13050000)
 rgb(0.81176471)=(0.70330000,0.32640000,0.12150000)
 rgb(0.81568627)=(0.69740000,0.31750000,0.11280000)
 rgb(0.81960784)=(0.69120000,0.30850000,0.10410000)
 rgb(0.82352941)=(0.68470000,0.29930000,0.09560000)
 rgb(0.82745098)=(0.67770000,0.28990000,0.08740000)
 rgb(0.83137255)=(0.67050000,0.28050000,0.07920000)
 rgb(0.83529412)=(0.66290000,0.27100000,0.07150000)
 rgb(0.83921569)=(0.65500000,0.26150000,0.06410000)
 rgb(0.84313725)=(0.64700000,0.25210000,0.05710000)
 rgb(0.84705882)=(0.63870000,0.24270000,0.05060000)
 rgb(0.85098039)=(0.63030000,0.23350000,0.04480000)
 rgb(0.85490196)=(0.62170000,0.22440000,0.03940000)
 rgb(0.85882353)=(0.61310000,0.21570000,0.03480000)
 rgb(0.86274510)=(0.60450000,0.20710000,0.03110000)
 rgb(0.86666667)=(0.59590000,0.19870000,0.02820000)
 rgb(0.87058824)=(0.58740000,0.19070000,0.02600000)
 rgb(0.87450980)=(0.57890000,0.18290000,0.02440000)
 rgb(0.87843137)=(0.57050000,0.17540000,0.02330000)
 rgb(0.88235294)=(0.56230000,0.16820000,0.02250000)
 rgb(0.88627451)=(0.55410000,0.16120000,0.02210000)
 rgb(0.89019608)=(0.54600000,0.15440000,0.02190000)
 rgb(0.89411765)=(0.53800000,0.14790000,0.02170000)
 rgb(0.89803922)=(0.53020000,0.14150000,0.02170000)
 rgb(0.90196078)=(0.52240000,0.13530000,0.02180000)
 rgb(0.90588235)=(0.51480000,0.12920000,0.02200000)
 rgb(0.90980392)=(0.50720000,0.12330000,0.02220000)
 rgb(0.91372549)=(0.49970000,0.11750000,0.02250000)
 rgb(0.91764706)=(0.49230000,0.11180000,0.02280000)
 rgb(0.92156863)=(0.48500000,0.10620000,0.02310000)
 rgb(0.92549020)=(0.47780000,0.10060000,0.02350000)
 rgb(0.92941176)=(0.47060000,0.09520000,0.02390000)
 rgb(0.93333333)=(0.46350000,0.08970000,0.02430000)
 rgb(0.93725490)=(0.45650000,0.08430000,0.02480000)
 rgb(0.94117647)=(0.44950000,0.07870000,0.02520000)
 rgb(0.94509804)=(0.44260000,0.07340000,0.02560000)
 rgb(0.94901961)=(0.43570000,0.06790000,0.02610000)
 rgb(0.95294118)=(0.42890000,0.06240000,0.02650000)
 rgb(0.95686275)=(0.42210000,0.05680000,0.02700000)
 rgb(0.96078431)=(0.41540000,0.05110000,0.02740000)
 rgb(0.96470588)=(0.40880000,0.04540000,0.02780000)
 rgb(0.96862745)=(0.40210000,0.03940000,0.02820000)
 rgb(0.97254902)=(0.39560000,0.03340000,0.02860000)
 rgb(0.97647059)=(0.38900000,0.02780000,0.02890000)
 rgb(0.98039216)=(0.38250000,0.02260000,0.02930000)
 rgb(0.98431373)=(0.37600000,0.01760000,0.02960000)
 rgb(0.98823529)=(0.36960000,0.01290000,0.02990000)
 rgb(0.99215686)=(0.36320000,0.00820000,0.03010000)
 rgb(0.99607843)=(0.35680000,0.00400000,0.03030000)
 rgb(1.00000000)=(0.35040000,0.00010000,0.03050000)},
 }

 \pgfplotsset{
 colormap={tableaucolorblind}{rgb(0.00000000)=(0.06666667,0.43921569,0.66666667)
 rgb(0.11111111)=(0.98823529,0.49019608,0.04313725)
 rgb(0.22222222)=(0.63921569,0.67450980,0.72549020)
 rgb(0.33333333)=(0.34117647,0.37647059,0.42352941)
 rgb(0.44444444)=(0.37254902,0.63529412,0.80784314)
 rgb(0.55555556)=(0.78431373,0.32156863,0.00000000)
 rgb(0.66666667)=(0.48235294,0.51764706,0.56078431)
 rgb(0.77777778)=(0.63921569,0.80000000,0.91372549)
 rgb(0.88888889)=(1.00000000,0.73725490,0.47450980)
 rgb(1.00000000)=(0.78431373,0.81568627,0.85098039)},
 }

\definecolor{tabcblue}      {rgb}{0.06666667,0.43921569,0.66666667}
\definecolor{tabcorange}    {rgb}{0.98823529,0.49019608,0.04313725}
\definecolor{tabcgreylight} {rgb}{0.63921569,0.67450980,0.72549020}
\definecolor{tabcgrey}      {rgb}{0.34117647,0.37647059,0.42352941}
\definecolor{tabcbluelight} {rgb}{0.37254902,0.63529412,0.80784314}
\definecolor{tabcorangedark}{rgb}{0.78431373,0.32156863,0.00000000}
\definecolor{tabcgreymid}   {rgb}{0.48235294,0.51764706,0.56078431}
\definecolor{tabcbluepale}  {rgb}{0.63921569,0.80000000,0.91372549}
\definecolor{tabcorangepale}{rgb}{1.00000000,0.73725490,0.47450980}
\definecolor{tabcgreypale}  {rgb}{0.78431373,0.81568627,0.85098039}

\definecolor{tabblue}  {rgb}{0.12156863,0.46666667,0.70588235}
\definecolor{taborange}{rgb}{1.00000000,0.49803922,0.05490196}
\definecolor{tabgreen} {rgb}{0.17254902,0.62745098,0.17254902}
\definecolor{tabred}   {rgb}{0.83921569,0.15294118,0.15686275}
\definecolor{tabpurple}{rgb}{0.58039216,0.40392157,0.74117647}
\definecolor{tabbrown} {rgb}{0.54901961,0.33725490,0.29411765}
\definecolor{tabpink}  {rgb}{0.89019608,0.46666667,0.76078431}
\definecolor{tabgrey}  {rgb}{0.49803922,0.49803922,0.49803922}
\definecolor{tabolive} {rgb}{0.73725490,0.74117647,0.13333333}
\definecolor{tabcyan}  {rgb}{0.09019608,0.74509804,0.81176471}

\definecolor{vikblue}{rgb}{0.136702,0.450888,0.626062}
\definecolor{vikred} {rgb}{0.613135,0.215657,0.034829}
\definecolor{vikredlight}{rgb}{0.887773,0.258164,0.054834}

\pgfplotsset{
  colormap={tab10}{rgb=(0.12156863,0.46666667,0.70588235)
    rgb=(1.00000000,0.49803922,0.05490196)
    rgb=(0.17254902,0.62745098,0.17254902)
    rgb=(0.83921569,0.15294118,0.15686275)
    rgb=(0.58039216,0.40392157,0.74117647)
    rgb=(0.54901961,0.33725490,0.29411765)
    rgb=(0.89019608,0.46666667,0.76078431)
    rgb=(0.49803922,0.49803922,0.49803922)
    rgb=(0.73725490,0.74117647,0.13333333)
    rgb=(0.09019608,0.74509804,0.81176471)},
}

\pgfplotsset{
  colormap={vik}{
    rgb=(0.001328,0.069836,0.379529)
    rgb=(0.002366,0.076475,0.383518)
    rgb=(0.003304,0.083083,0.387487)
    rgb=(0.004146,0.089590,0.391477)
    rgb=(0.004897,0.095948,0.395453)
    rgb=(0.005563,0.102274,0.399409)
    rgb=(0.006151,0.108500,0.403388)
    rgb=(0.006668,0.114686,0.407339)
    rgb=(0.007119,0.120845,0.411288)
    rgb=(0.007512,0.126958,0.415230)
    rgb=(0.007850,0.133068,0.419166)
    rgb=(0.008141,0.139092,0.423079)
    rgb=(0.008391,0.145171,0.427006)
    rgb=(0.008606,0.151144,0.430910)
    rgb=(0.008790,0.157140,0.434809)
    rgb=(0.008947,0.163152,0.438691)
    rgb=(0.009080,0.169142,0.442587)
    rgb=(0.009193,0.175103,0.446459)
    rgb=(0.009290,0.181052,0.450337)
    rgb=(0.009372,0.187051,0.454212)
    rgb=(0.009443,0.193028,0.458077)
    rgb=(0.009506,0.198999,0.461951)
    rgb=(0.009564,0.205011,0.465816)
    rgb=(0.009619,0.211021,0.469707)
    rgb=(0.009675,0.217047,0.473571)
    rgb=(0.009735,0.223084,0.477461)
    rgb=(0.009802,0.229123,0.481352)
    rgb=(0.009881,0.235206,0.485250)
    rgb=(0.009977,0.241277,0.489161)
    rgb=(0.010098,0.247386,0.493080)
    rgb=(0.010254,0.253516,0.497020)
    rgb=(0.010463,0.259675,0.500974)
    rgb=(0.010755,0.265853,0.504938)
    rgb=(0.011176,0.272037,0.508925)
    rgb=(0.011716,0.278296,0.512923)
    rgb=(0.012286,0.284554,0.516953)
    rgb=(0.012934,0.290865,0.520998)
    rgb=(0.013790,0.297214,0.525074)
    rgb=(0.014838,0.303577,0.529184)
    rgb=(0.016131,0.310015,0.533308)
    rgb=(0.017711,0.316474,0.537485)
    rgb=(0.019630,0.322986,0.541677)
    rgb=(0.021948,0.329550,0.545931)
    rgb=(0.024730,0.336144,0.550210)
    rgb=(0.028047,0.342826,0.554538)
    rgb=(0.031980,0.349543,0.558906)
    rgb=(0.036812,0.356332,0.563341)
    rgb=(0.042229,0.363171,0.567811)
    rgb=(0.048008,0.370086,0.572345)
    rgb=(0.054292,0.377080,0.576933)
    rgb=(0.060963,0.384129,0.581571)
    rgb=(0.068081,0.391265,0.586280)
    rgb=(0.075457,0.398460,0.591042)
    rgb=(0.083246,0.405740,0.595868)
    rgb=(0.091425,0.413088,0.600754)
    rgb=(0.099832,0.420499,0.605697)
    rgb=(0.108595,0.428000,0.610711)
    rgb=(0.117694,0.435566,0.615770)
    rgb=(0.127042,0.443194,0.620895)
    rgb=(0.136702,0.450888,0.626062)
    rgb=(0.146607,0.458643,0.631289)
    rgb=(0.156787,0.466457,0.636560)
    rgb=(0.167187,0.474324,0.641866)
    rgb=(0.177807,0.482238,0.647218)
    rgb=(0.188606,0.490191,0.652599)
    rgb=(0.199580,0.498193,0.658021)
    rgb=(0.210783,0.506201,0.663465)
    rgb=(0.222120,0.514263,0.668924)
    rgb=(0.233602,0.522322,0.674403)
    rgb=(0.245231,0.530414,0.679894)
    rgb=(0.256999,0.538517,0.685405)
    rgb=(0.268867,0.546617,0.690908)
    rgb=(0.280797,0.554717,0.696428)
    rgb=(0.292852,0.562822,0.701935)
    rgb=(0.304985,0.570907,0.707448)
    rgb=(0.317174,0.578997,0.712950)
    rgb=(0.329438,0.587064,0.718447)
    rgb=(0.341729,0.595123,0.723934)
    rgb=(0.354067,0.603164,0.729412)
    rgb=(0.366459,0.611186,0.734877)
    rgb=(0.378862,0.619189,0.740325)
    rgb=(0.391305,0.627159,0.745757)
    rgb=(0.403760,0.635114,0.751183)
    rgb=(0.416227,0.643046,0.756582)
    rgb=(0.428711,0.650956,0.761968)
    rgb=(0.441199,0.658836,0.767341)
    rgb=(0.453697,0.666696,0.772699)
    rgb=(0.466195,0.674537,0.778044)
    rgb=(0.478697,0.682349,0.783369)
    rgb=(0.491208,0.690143,0.788682)
    rgb=(0.503691,0.697910,0.793980)
    rgb=(0.516178,0.705661,0.799260)
    rgb=(0.528677,0.713387,0.804525)
    rgb=(0.541149,0.721090,0.809775)
    rgb=(0.553624,0.728778,0.815010)
    rgb=(0.566096,0.736441,0.820229)
    rgb=(0.578557,0.744089,0.825435)
    rgb=(0.591014,0.751718,0.830626)
    rgb=(0.603468,0.759314,0.835793)
    rgb=(0.615908,0.766896,0.840941)
    rgb=(0.628351,0.774452,0.846058)
    rgb=(0.640779,0.781988,0.851147)
    rgb=(0.653203,0.789485,0.856206)
    rgb=(0.665631,0.796945,0.861214)
    rgb=(0.678051,0.804371,0.866172)
    rgb=(0.690457,0.811742,0.871059)
    rgb=(0.702868,0.819048,0.875866)
    rgb=(0.715265,0.826290,0.880567)
    rgb=(0.727646,0.833439,0.885146)
    rgb=(0.740019,0.840479,0.889570)
    rgb=(0.752354,0.847380,0.893807)
    rgb=(0.764662,0.854125,0.897821)
    rgb=(0.776918,0.860678,0.901565)
    rgb=(0.789096,0.866991,0.904992)
    rgb=(0.801170,0.873031,0.908043)
    rgb=(0.813110,0.878738,0.910653)
    rgb=(0.824870,0.884062,0.912761)
    rgb=(0.836396,0.888934,0.914302)
    rgb=(0.847617,0.893289,0.915195)
    rgb=(0.858470,0.897074,0.915385)
    rgb=(0.868874,0.900206,0.914812)
    rgb=(0.878729,0.902636,0.913418)
    rgb=(0.887965,0.904303,0.911164)
    rgb=(0.896497,0.905178,0.908034)
    rgb=(0.904242,0.905221,0.904013)
    rgb=(0.911151,0.904422,0.899132)
    rgb=(0.917175,0.902800,0.893409)
    rgb=(0.922285,0.900367,0.886911)
    rgb=(0.926482,0.897173,0.879687)
    rgb=(0.929789,0.893256,0.871826)
    rgb=(0.932236,0.888698,0.863396)
    rgb=(0.933880,0.883552,0.854476)
    rgb=(0.934782,0.877893,0.845152)
    rgb=(0.935013,0.871795,0.835493)
    rgb=(0.934644,0.865313,0.825561)
    rgb=(0.933752,0.858522,0.815421)
    rgb=(0.932408,0.851469,0.805112)
    rgb=(0.930682,0.844208,0.794685)
    rgb=(0.928622,0.836778,0.784169)
    rgb=(0.926298,0.829215,0.773579)
    rgb=(0.923752,0.821545,0.762958)
    rgb=(0.921017,0.813795,0.752313)
    rgb=(0.918147,0.805997,0.741659)
    rgb=(0.915156,0.798157,0.731008)
    rgb=(0.912080,0.790294,0.720370)
    rgb=(0.908933,0.782421,0.709752)
    rgb=(0.905741,0.774540,0.699150)
    rgb=(0.902506,0.766670,0.688588)
    rgb=(0.899249,0.758812,0.678051)
    rgb=(0.895973,0.750973,0.667550)
    rgb=(0.892690,0.743148,0.657086)
    rgb=(0.889402,0.735345,0.646657)
    rgb=(0.886118,0.727569,0.636274)
    rgb=(0.882831,0.719826,0.625923)
    rgb=(0.879556,0.712106,0.615618)
    rgb=(0.876289,0.704419,0.605357)
    rgb=(0.873033,0.696764,0.595141)
    rgb=(0.869784,0.689144,0.584972)
    rgb=(0.866551,0.681541,0.574832)
    rgb=(0.863333,0.673985,0.564746)
    rgb=(0.860121,0.666453,0.554708)
    rgb=(0.856920,0.658957,0.544709)
    rgb=(0.853732,0.651500,0.534753)
    rgb=(0.850562,0.644061,0.524842)
    rgb=(0.847402,0.636670,0.514974)
    rgb=(0.844258,0.629296,0.505146)
    rgb=(0.841125,0.621957,0.495369)
    rgb=(0.838005,0.614653,0.485627)
    rgb=(0.834895,0.607392,0.475941)
    rgb=(0.831802,0.600144,0.466284)
    rgb=(0.828715,0.592938,0.456675)
    rgb=(0.825639,0.585758,0.447109)
    rgb=(0.822582,0.578600,0.437595)
    rgb=(0.819528,0.571478,0.428106)
    rgb=(0.816496,0.564388,0.418657)
    rgb=(0.813463,0.557328,0.409260)
    rgb=(0.810446,0.550285,0.399892)
    rgb=(0.807443,0.543274,0.390575)
    rgb=(0.804446,0.536288,0.381299)
    rgb=(0.801454,0.529329,0.372040)
    rgb=(0.798475,0.522380,0.362835)
    rgb=(0.795500,0.515460,0.353660)
    rgb=(0.792535,0.508575,0.344523)
    rgb=(0.789573,0.501692,0.335435)
    rgb=(0.786617,0.494827,0.326343)
    rgb=(0.783657,0.487977,0.317312)
    rgb=(0.780695,0.481123,0.308300)
    rgb=(0.777737,0.474295,0.299327)
    rgb=(0.774763,0.467464,0.290352)
    rgb=(0.771788,0.460620,0.281424)
    rgb=(0.768787,0.453783,0.272508)
    rgb=(0.765776,0.446929,0.263640)
    rgb=(0.762724,0.440055,0.254764)
    rgb=(0.759638,0.433147,0.245872)
    rgb=(0.756510,0.426200,0.237047)
    rgb=(0.753316,0.419216,0.228190)
    rgb=(0.750051,0.412163,0.219330)
    rgb=(0.746698,0.405028,0.210470)
    rgb=(0.743239,0.397819,0.201593)
    rgb=(0.739651,0.390493,0.192739)
    rgb=(0.735899,0.383060,0.183852)
    rgb=(0.731988,0.375473,0.174977)
    rgb=(0.727865,0.367743,0.166045)
    rgb=(0.723516,0.359852,0.157131)
    rgb=(0.718915,0.351766,0.148211)
    rgb=(0.714028,0.343503,0.139282)
    rgb=(0.708841,0.335048,0.130458)
    rgb=(0.703318,0.326354,0.121545)
    rgb=(0.697448,0.317502,0.112841)
    rgb=(0.691227,0.308462,0.104132)
    rgb=(0.684653,0.299264,0.095633)
    rgb=(0.677734,0.289916,0.087350)
    rgb=(0.670476,0.280477,0.079197)
    rgb=(0.662904,0.271015,0.071510)
    rgb=(0.655048,0.261520,0.064079)
    rgb=(0.646969,0.252081,0.057104)
    rgb=(0.638686,0.242711,0.050618)
    rgb=(0.630261,0.233488,0.044750)
    rgb=(0.621722,0.224449,0.039414)
    rgb=(0.613135,0.215657,0.034829)
    rgb=(0.604539,0.207086,0.031072)
    rgb=(0.595947,0.198741,0.028212)
    rgb=(0.587403,0.190700,0.026019)
    rgb=(0.578937,0.182918,0.024396)
    rgb=(0.570545,0.175423,0.023257)
    rgb=(0.562268,0.168171,0.022523)
    rgb=(0.554076,0.161202,0.022110)
    rgb=(0.546007,0.154400,0.021861)
    rgb=(0.538043,0.147854,0.021737)
    rgb=(0.530182,0.141491,0.021722)
    rgb=(0.522424,0.135276,0.021800)
    rgb=(0.514776,0.129209,0.021957)
    rgb=(0.507213,0.123272,0.022179)
    rgb=(0.499733,0.117487,0.022455)
    rgb=(0.492348,0.111818,0.022775)
    rgb=(0.485034,0.106209,0.023130)
    rgb=(0.477801,0.100607,0.023513)
    rgb=(0.470639,0.095156,0.023916)
    rgb=(0.463530,0.089668,0.024336)
    rgb=(0.456494,0.084258,0.024766)
    rgb=(0.449521,0.078741,0.025203)
    rgb=(0.442603,0.073404,0.025644)
    rgb=(0.435737,0.067904,0.026084)
    rgb=(0.428918,0.062415,0.026522)
    rgb=(0.422146,0.056832,0.026954)
    rgb=(0.415437,0.051116,0.027378)
    rgb=(0.408768,0.045352,0.027790)
    rgb=(0.402132,0.039448,0.028189)
    rgb=(0.395562,0.033385,0.028570)
    rgb=(0.389015,0.027844,0.028932)
    rgb=(0.382496,0.022586,0.029271)
    rgb=(0.376028,0.017608,0.029583)
    rgb=(0.369578,0.012890,0.029866)
    rgb=(0.363161,0.008243,0.030115)
    rgb=(0.356785,0.004035,0.030327)
    rgb=(0.350423,0.000061,0.030499)
  }
}

\pgfplotsset{
  colormap={lipari}{
    rgb=(0.011400,0.073200,0.148300)
    rgb=(0.014000,0.079100,0.155400)
    rgb=(0.015900,0.084700,0.162500)
    rgb=(0.017200,0.090000,0.169700)
    rgb=(0.018000,0.095100,0.177000)
    rgb=(0.019000,0.100000,0.184300)
    rgb=(0.020000,0.104900,0.191600)
    rgb=(0.021000,0.110000,0.199000)
    rgb=(0.022100,0.114900,0.206400)
    rgb=(0.023200,0.119900,0.213800)
    rgb=(0.024400,0.125000,0.221200)
    rgb=(0.025700,0.130100,0.228700)
    rgb=(0.027100,0.135200,0.236100)
    rgb=(0.028600,0.140300,0.243500)
    rgb=(0.030200,0.145500,0.251000)
    rgb=(0.032000,0.150600,0.258500)
    rgb=(0.033800,0.155800,0.266000)
    rgb=(0.036100,0.161000,0.273400)
    rgb=(0.038400,0.166300,0.280800)
    rgb=(0.041000,0.171500,0.288300)
    rgb=(0.043500,0.176800,0.295700)
    rgb=(0.046500,0.182100,0.303100)
    rgb=(0.049600,0.187500,0.310500)
    rgb=(0.052800,0.192800,0.317800)
    rgb=(0.056400,0.198100,0.325100)
    rgb=(0.060100,0.203500,0.332300)
    rgb=(0.064100,0.208900,0.339500)
    rgb=(0.068300,0.214300,0.346600)
    rgb=(0.072800,0.219700,0.353700)
    rgb=(0.077500,0.225100,0.360700)
    rgb=(0.082600,0.230500,0.367600)
    rgb=(0.087900,0.235900,0.374400)
    rgb=(0.093400,0.241200,0.381000)
    rgb=(0.099300,0.246600,0.387600)
    rgb=(0.105400,0.251900,0.394000)
    rgb=(0.111800,0.257200,0.400200)
    rgb=(0.118400,0.262500,0.406300)
    rgb=(0.125200,0.267700,0.412300)
    rgb=(0.132300,0.272800,0.418000)
    rgb=(0.139500,0.277900,0.423500)
    rgb=(0.147000,0.282900,0.428800)
    rgb=(0.154600,0.287800,0.433900)
    rgb=(0.162300,0.292600,0.438700)
    rgb=(0.170100,0.297300,0.443300)
    rgb=(0.178100,0.301800,0.447600)
    rgb=(0.186000,0.306300,0.451700)
    rgb=(0.194000,0.310600,0.455400)
    rgb=(0.202000,0.314700,0.458900)
    rgb=(0.210000,0.318600,0.462100)
    rgb=(0.217900,0.322400,0.465100)
    rgb=(0.225800,0.326000,0.467700)
    rgb=(0.233500,0.329500,0.470100)
    rgb=(0.241200,0.332700,0.472200)
    rgb=(0.248700,0.335700,0.474000)
    rgb=(0.256100,0.338600,0.475500)
    rgb=(0.263300,0.341300,0.476900)
    rgb=(0.270400,0.343800,0.478000)
    rgb=(0.277200,0.346100,0.478800)
    rgb=(0.283900,0.348300,0.479500)
    rgb=(0.290400,0.350300,0.480000)
    rgb=(0.296800,0.352100,0.480300)
    rgb=(0.302900,0.353800,0.480400)
    rgb=(0.308900,0.355400,0.480400)
    rgb=(0.314700,0.356800,0.480300)
    rgb=(0.320400,0.358100,0.480000)
    rgb=(0.325800,0.359300,0.479700)
    rgb=(0.331200,0.360400,0.479200)
    rgb=(0.336400,0.361400,0.478700)
    rgb=(0.341500,0.362300,0.478000)
    rgb=(0.346500,0.363100,0.477400)
    rgb=(0.351400,0.363900,0.476600)
    rgb=(0.356200,0.364600,0.475800)
    rgb=(0.360900,0.365200,0.475000)
    rgb=(0.365600,0.365800,0.474100)
    rgb=(0.370200,0.366400,0.473200)
    rgb=(0.374800,0.366900,0.472200)
    rgb=(0.379300,0.367300,0.471200)
    rgb=(0.383700,0.367800,0.470200)
    rgb=(0.388200,0.368200,0.469200)
    rgb=(0.392600,0.368600,0.468200)
    rgb=(0.397100,0.369000,0.467100)
    rgb=(0.401500,0.369300,0.466000)
    rgb=(0.405900,0.369700,0.464900)
    rgb=(0.410400,0.370000,0.463800)
    rgb=(0.414800,0.370300,0.462700)
    rgb=(0.419300,0.370600,0.461500)
    rgb=(0.423800,0.370900,0.460400)
    rgb=(0.428300,0.371200,0.459200)
    rgb=(0.432900,0.371500,0.458000)
    rgb=(0.437500,0.371800,0.456800)
    rgb=(0.442100,0.372100,0.455600)
    rgb=(0.446700,0.372400,0.454400)
    rgb=(0.451400,0.372700,0.453200)
    rgb=(0.456100,0.373000,0.452000)
    rgb=(0.460900,0.373300,0.450700)
    rgb=(0.465600,0.373600,0.449500)
    rgb=(0.470500,0.373900,0.448200)
    rgb=(0.475300,0.374100,0.446900)
    rgb=(0.480200,0.374400,0.445600)
    rgb=(0.485200,0.374700,0.444300)
    rgb=(0.490100,0.375000,0.443000)
    rgb=(0.495200,0.375300,0.441700)
    rgb=(0.500200,0.375600,0.440400)
    rgb=(0.505300,0.375900,0.439100)
    rgb=(0.510400,0.376200,0.437700)
    rgb=(0.515600,0.376500,0.436400)
    rgb=(0.520800,0.376800,0.435000)
    rgb=(0.526100,0.377100,0.433600)
    rgb=(0.531300,0.377400,0.432200)
    rgb=(0.536700,0.377700,0.430800)
    rgb=(0.542000,0.378000,0.429400)
    rgb=(0.547400,0.378300,0.428000)
    rgb=(0.552900,0.378700,0.426600)
    rgb=(0.558400,0.379000,0.425200)
    rgb=(0.563900,0.379300,0.423700)
    rgb=(0.569500,0.379600,0.422300)
    rgb=(0.575100,0.380000,0.420800)
    rgb=(0.580800,0.380300,0.419300)
    rgb=(0.586500,0.380700,0.417900)
    rgb=(0.592300,0.381000,0.416400)
    rgb=(0.598000,0.381400,0.414900)
    rgb=(0.603900,0.381700,0.413400)
    rgb=(0.609800,0.382100,0.411800)
    rgb=(0.615700,0.382500,0.410300)
    rgb=(0.621600,0.382900,0.408800)
    rgb=(0.627700,0.383300,0.407300)
    rgb=(0.633700,0.383700,0.405700)
    rgb=(0.639800,0.384100,0.404200)
    rgb=(0.646000,0.384500,0.402600)
    rgb=(0.652200,0.384900,0.401100)
    rgb=(0.658400,0.385400,0.399500)
    rgb=(0.664700,0.385900,0.398000)
    rgb=(0.671100,0.386400,0.396400)
    rgb=(0.677500,0.386900,0.394800)
    rgb=(0.683900,0.387500,0.393300)
    rgb=(0.690400,0.388100,0.391800)
    rgb=(0.696900,0.388800,0.390200)
    rgb=(0.703400,0.389400,0.388700)
    rgb=(0.710000,0.390200,0.387200)
    rgb=(0.716700,0.391000,0.385700)
    rgb=(0.723400,0.391800,0.384300)
    rgb=(0.730100,0.392700,0.382900)
    rgb=(0.736800,0.393700,0.381500)
    rgb=(0.743600,0.394800,0.380200)
    rgb=(0.750400,0.395900,0.378800)
    rgb=(0.757200,0.397200,0.377600)
    rgb=(0.764000,0.398500,0.376400)
    rgb=(0.770800,0.400000,0.375300)
    rgb=(0.777500,0.401600,0.374300)
    rgb=(0.784300,0.403400,0.373300)
    rgb=(0.791000,0.405300,0.372400)
    rgb=(0.797700,0.407300,0.371600)
    rgb=(0.804400,0.409500,0.370900)
    rgb=(0.811000,0.411900,0.370400)
    rgb=(0.817500,0.414400,0.369900)
    rgb=(0.823900,0.417200,0.369600)
    rgb=(0.830200,0.420100,0.369400)
    rgb=(0.836400,0.423300,0.369400)
    rgb=(0.842400,0.426600,0.369500)
    rgb=(0.848300,0.430200,0.369800)
    rgb=(0.854000,0.434000,0.370200)
    rgb=(0.859500,0.438000,0.370900)
    rgb=(0.864800,0.442200,0.371600)
    rgb=(0.869900,0.446500,0.372600)
    rgb=(0.874800,0.451100,0.373700)
    rgb=(0.879400,0.455900,0.375100)
    rgb=(0.883800,0.460800,0.376600)
    rgb=(0.887900,0.465900,0.378200)
    rgb=(0.891700,0.471200,0.380000)
    rgb=(0.895200,0.476600,0.382000)
    rgb=(0.898500,0.482100,0.384200)
    rgb=(0.901500,0.487700,0.386500)
    rgb=(0.904200,0.493400,0.388900)
    rgb=(0.906600,0.499100,0.391500)
    rgb=(0.908800,0.505000,0.394200)
    rgb=(0.910600,0.510800,0.397000)
    rgb=(0.912200,0.516700,0.399900)
    rgb=(0.913600,0.522600,0.403000)
    rgb=(0.914700,0.528500,0.406100)
    rgb=(0.915500,0.534300,0.409300)
    rgb=(0.916200,0.540200,0.412500)
    rgb=(0.916600,0.546000,0.415800)
    rgb=(0.916900,0.551800,0.419200)
    rgb=(0.917000,0.557500,0.422500)
    rgb=(0.916900,0.563100,0.426000)
    rgb=(0.916600,0.568700,0.429400)
    rgb=(0.916300,0.574300,0.432900)
    rgb=(0.915800,0.579700,0.436400)
    rgb=(0.915200,0.585100,0.439900)
    rgb=(0.914500,0.590500,0.443500)
    rgb=(0.913800,0.595800,0.447000)
    rgb=(0.912900,0.601000,0.450600)
    rgb=(0.912100,0.606100,0.454100)
    rgb=(0.911100,0.611300,0.457700)
    rgb=(0.910200,0.616300,0.461300)
    rgb=(0.909200,0.621300,0.465000)
    rgb=(0.908300,0.626300,0.468600)
    rgb=(0.907300,0.631300,0.472300)
    rgb=(0.906400,0.636200,0.476000)
    rgb=(0.905400,0.641000,0.479700)
    rgb=(0.904500,0.645900,0.483500)
    rgb=(0.903700,0.650800,0.487400)
    rgb=(0.902900,0.655600,0.491300)
    rgb=(0.902100,0.660400,0.495200)
    rgb=(0.901400,0.665300,0.499200)
    rgb=(0.900700,0.670100,0.503300)
    rgb=(0.900100,0.675000,0.507500)
    rgb=(0.899600,0.679900,0.511800)
    rgb=(0.899200,0.684800,0.516100)
    rgb=(0.898900,0.689700,0.520600)
    rgb=(0.898700,0.694600,0.525200)
    rgb=(0.898500,0.699600,0.529900)
    rgb=(0.898500,0.704700,0.534700)
    rgb=(0.898600,0.709800,0.539700)
    rgb=(0.898700,0.714900,0.544800)
    rgb=(0.899000,0.720100,0.550000)
    rgb=(0.899400,0.725300,0.555400)
    rgb=(0.900000,0.730600,0.560900)
    rgb=(0.900600,0.735900,0.566600)
    rgb=(0.901400,0.741300,0.572400)
    rgb=(0.902300,0.746800,0.578400)
    rgb=(0.903400,0.752300,0.584500)
    rgb=(0.904500,0.757900,0.590800)
    rgb=(0.905800,0.763500,0.597200)
    rgb=(0.907200,0.769200,0.603800)
    rgb=(0.908700,0.774900,0.610500)
    rgb=(0.910400,0.780700,0.617400)
    rgb=(0.912100,0.786500,0.624500)
    rgb=(0.914000,0.792400,0.631600)
    rgb=(0.916000,0.798300,0.638900)
    rgb=(0.918100,0.804200,0.646400)
    rgb=(0.920300,0.810200,0.653900)
    rgb=(0.922600,0.816300,0.661600)
    rgb=(0.925000,0.822300,0.669400)
    rgb=(0.927500,0.828400,0.677300)
    rgb=(0.930100,0.834600,0.685300)
    rgb=(0.932800,0.840700,0.693300)
    rgb=(0.935600,0.846900,0.701500)
    rgb=(0.938400,0.853100,0.709800)
    rgb=(0.941300,0.859300,0.718100)
    rgb=(0.944200,0.865500,0.726500)
    rgb=(0.947200,0.871700,0.734900)
    rgb=(0.950300,0.878000,0.743500)
    rgb=(0.953400,0.884200,0.752000)
    rgb=(0.956500,0.890500,0.760600)
    rgb=(0.959700,0.896700,0.769200)
    rgb=(0.962900,0.903000,0.777900)
    rgb=(0.966100,0.909200,0.786600)
    rgb=(0.969400,0.915500,0.795300)
    rgb=(0.972700,0.921700,0.804100)
    rgb=(0.975900,0.928000,0.812800)
    rgb=(0.979200,0.934200,0.821600)
    rgb=(0.982500,0.940400,0.830300)
    rgb=(0.985800,0.946600,0.839100)
    rgb=(0.989000,0.952800,0.847800)
    rgb=(0.992300,0.959000,0.856600)
  }
}

\pgfplotsset{
  colormap={lipari_r}{
    rgb=(0.992300,0.959000,0.856600)
    rgb=(0.989000,0.952800,0.847800)
    rgb=(0.985800,0.946600,0.839100)
    rgb=(0.982500,0.940400,0.830300)
    rgb=(0.979200,0.934200,0.821600)
    rgb=(0.975900,0.928000,0.812800)
    rgb=(0.972700,0.921700,0.804100)
    rgb=(0.969400,0.915500,0.795300)
    rgb=(0.966100,0.909200,0.786600)
    rgb=(0.962900,0.903000,0.777900)
    rgb=(0.959700,0.896700,0.769200)
    rgb=(0.956500,0.890500,0.760600)
    rgb=(0.953400,0.884200,0.752000)
    rgb=(0.950300,0.878000,0.743500)
    rgb=(0.947200,0.871700,0.734900)
    rgb=(0.944200,0.865500,0.726500)
    rgb=(0.941300,0.859300,0.718100)
    rgb=(0.938400,0.853100,0.709800)
    rgb=(0.935600,0.846900,0.701500)
    rgb=(0.932800,0.840700,0.693300)
    rgb=(0.930100,0.834600,0.685300)
    rgb=(0.927500,0.828400,0.677300)
    rgb=(0.925000,0.822300,0.669400)
    rgb=(0.922600,0.816300,0.661600)
    rgb=(0.920300,0.810200,0.653900)
    rgb=(0.918100,0.804200,0.646400)
    rgb=(0.916000,0.798300,0.638900)
    rgb=(0.914000,0.792400,0.631600)
    rgb=(0.912100,0.786500,0.624500)
    rgb=(0.910400,0.780700,0.617400)
    rgb=(0.908700,0.774900,0.610500)
    rgb=(0.907200,0.769200,0.603800)
    rgb=(0.905800,0.763500,0.597200)
    rgb=(0.904500,0.757900,0.590800)
    rgb=(0.903400,0.752300,0.584500)
    rgb=(0.902300,0.746800,0.578400)
    rgb=(0.901400,0.741300,0.572400)
    rgb=(0.900600,0.735900,0.566600)
    rgb=(0.900000,0.730600,0.560900)
    rgb=(0.899400,0.725300,0.555400)
    rgb=(0.899000,0.720100,0.550000)
    rgb=(0.898700,0.714900,0.544800)
    rgb=(0.898600,0.709800,0.539700)
    rgb=(0.898500,0.704700,0.534700)
    rgb=(0.898500,0.699600,0.529900)
    rgb=(0.898700,0.694600,0.525200)
    rgb=(0.898900,0.689700,0.520600)
    rgb=(0.899200,0.684800,0.516100)
    rgb=(0.899600,0.679900,0.511800)
    rgb=(0.900100,0.675000,0.507500)
    rgb=(0.900700,0.670100,0.503300)
    rgb=(0.901400,0.665300,0.499200)
    rgb=(0.902100,0.660400,0.495200)
    rgb=(0.902900,0.655600,0.491300)
    rgb=(0.903700,0.650800,0.487400)
    rgb=(0.904500,0.645900,0.483500)
    rgb=(0.905400,0.641000,0.479700)
    rgb=(0.906400,0.636200,0.476000)
    rgb=(0.907300,0.631300,0.472300)
    rgb=(0.908300,0.626300,0.468600)
    rgb=(0.909200,0.621300,0.465000)
    rgb=(0.910200,0.616300,0.461300)
    rgb=(0.911100,0.611300,0.457700)
    rgb=(0.912100,0.606100,0.454100)
    rgb=(0.912900,0.601000,0.450600)
    rgb=(0.913800,0.595800,0.447000)
    rgb=(0.914500,0.590500,0.443500)
    rgb=(0.915200,0.585100,0.439900)
    rgb=(0.915800,0.579700,0.436400)
    rgb=(0.916300,0.574300,0.432900)
    rgb=(0.916600,0.568700,0.429400)
    rgb=(0.916900,0.563100,0.426000)
    rgb=(0.917000,0.557500,0.422500)
    rgb=(0.916900,0.551800,0.419200)
    rgb=(0.916600,0.546000,0.415800)
    rgb=(0.916200,0.540200,0.412500)
    rgb=(0.915500,0.534300,0.409300)
    rgb=(0.914700,0.528500,0.406100)
    rgb=(0.913600,0.522600,0.403000)
    rgb=(0.912200,0.516700,0.399900)
    rgb=(0.910600,0.510800,0.397000)
    rgb=(0.908800,0.505000,0.394200)
    rgb=(0.906600,0.499100,0.391500)
    rgb=(0.904200,0.493400,0.388900)
    rgb=(0.901500,0.487700,0.386500)
    rgb=(0.898500,0.482100,0.384200)
    rgb=(0.895200,0.476600,0.382000)
    rgb=(0.891700,0.471200,0.380000)
    rgb=(0.887900,0.465900,0.378200)
    rgb=(0.883800,0.460800,0.376600)
    rgb=(0.879400,0.455900,0.375100)
    rgb=(0.874800,0.451100,0.373700)
    rgb=(0.869900,0.446500,0.372600)
    rgb=(0.864800,0.442200,0.371600)
    rgb=(0.859500,0.438000,0.370900)
    rgb=(0.854000,0.434000,0.370200)
    rgb=(0.848300,0.430200,0.369800)
    rgb=(0.842400,0.426600,0.369500)
    rgb=(0.836400,0.423300,0.369400)
    rgb=(0.830200,0.420100,0.369400)
    rgb=(0.823900,0.417200,0.369600)
    rgb=(0.817500,0.414400,0.369900)
    rgb=(0.811000,0.411900,0.370400)
    rgb=(0.804400,0.409500,0.370900)
    rgb=(0.797700,0.407300,0.371600)
    rgb=(0.791000,0.405300,0.372400)
    rgb=(0.784300,0.403400,0.373300)
    rgb=(0.777500,0.401600,0.374300)
    rgb=(0.770800,0.400000,0.375300)
    rgb=(0.764000,0.398500,0.376400)
    rgb=(0.757200,0.397200,0.377600)
    rgb=(0.750400,0.395900,0.378800)
    rgb=(0.743600,0.394800,0.380200)
    rgb=(0.736800,0.393700,0.381500)
    rgb=(0.730100,0.392700,0.382900)
    rgb=(0.723400,0.391800,0.384300)
    rgb=(0.716700,0.391000,0.385700)
    rgb=(0.710000,0.390200,0.387200)
    rgb=(0.703400,0.389400,0.388700)
    rgb=(0.696900,0.388800,0.390200)
    rgb=(0.690400,0.388100,0.391800)
    rgb=(0.683900,0.387500,0.393300)
    rgb=(0.677500,0.386900,0.394800)
    rgb=(0.671100,0.386400,0.396400)
    rgb=(0.664700,0.385900,0.398000)
    rgb=(0.658400,0.385400,0.399500)
    rgb=(0.652200,0.384900,0.401100)
    rgb=(0.646000,0.384500,0.402600)
    rgb=(0.639800,0.384100,0.404200)
    rgb=(0.633700,0.383700,0.405700)
    rgb=(0.627700,0.383300,0.407300)
    rgb=(0.621600,0.382900,0.408800)
    rgb=(0.615700,0.382500,0.410300)
    rgb=(0.609800,0.382100,0.411800)
    rgb=(0.603900,0.381700,0.413400)
    rgb=(0.598000,0.381400,0.414900)
    rgb=(0.592300,0.381000,0.416400)
    rgb=(0.586500,0.380700,0.417900)
    rgb=(0.580800,0.380300,0.419300)
    rgb=(0.575100,0.380000,0.420800)
    rgb=(0.569500,0.379600,0.422300)
    rgb=(0.563900,0.379300,0.423700)
    rgb=(0.558400,0.379000,0.425200)
    rgb=(0.552900,0.378700,0.426600)
    rgb=(0.547400,0.378300,0.428000)
    rgb=(0.542000,0.378000,0.429400)
    rgb=(0.536700,0.377700,0.430800)
    rgb=(0.531300,0.377400,0.432200)
    rgb=(0.526100,0.377100,0.433600)
    rgb=(0.520800,0.376800,0.435000)
    rgb=(0.515600,0.376500,0.436400)
    rgb=(0.510400,0.376200,0.437700)
    rgb=(0.505300,0.375900,0.439100)
    rgb=(0.500200,0.375600,0.440400)
    rgb=(0.495200,0.375300,0.441700)
    rgb=(0.490100,0.375000,0.443000)
    rgb=(0.485200,0.374700,0.444300)
    rgb=(0.480200,0.374400,0.445600)
    rgb=(0.475300,0.374100,0.446900)
    rgb=(0.470500,0.373900,0.448200)
    rgb=(0.465600,0.373600,0.449500)
    rgb=(0.460900,0.373300,0.450700)
    rgb=(0.456100,0.373000,0.452000)
    rgb=(0.451400,0.372700,0.453200)
    rgb=(0.446700,0.372400,0.454400)
    rgb=(0.442100,0.372100,0.455600)
    rgb=(0.437500,0.371800,0.456800)
    rgb=(0.432900,0.371500,0.458000)
    rgb=(0.428300,0.371200,0.459200)
    rgb=(0.423800,0.370900,0.460400)
    rgb=(0.419300,0.370600,0.461500)
    rgb=(0.414800,0.370300,0.462700)
    rgb=(0.410400,0.370000,0.463800)
    rgb=(0.405900,0.369700,0.464900)
    rgb=(0.401500,0.369300,0.466000)
    rgb=(0.397100,0.369000,0.467100)
    rgb=(0.392600,0.368600,0.468200)
    rgb=(0.388200,0.368200,0.469200)
    rgb=(0.383700,0.367800,0.470200)
    rgb=(0.379300,0.367300,0.471200)
    rgb=(0.374800,0.366900,0.472200)
    rgb=(0.370200,0.366400,0.473200)
    rgb=(0.365600,0.365800,0.474100)
    rgb=(0.360900,0.365200,0.475000)
    rgb=(0.356200,0.364600,0.475800)
    rgb=(0.351400,0.363900,0.476600)
    rgb=(0.346500,0.363100,0.477400)
    rgb=(0.341500,0.362300,0.478000)
    rgb=(0.336400,0.361400,0.478700)
    rgb=(0.331200,0.360400,0.479200)
    rgb=(0.325800,0.359300,0.479700)
    rgb=(0.320400,0.358100,0.480000)
    rgb=(0.314700,0.356800,0.480300)
    rgb=(0.308900,0.355400,0.480400)
    rgb=(0.302900,0.353800,0.480400)
    rgb=(0.296800,0.352100,0.480300)
    rgb=(0.290400,0.350300,0.480000)
    rgb=(0.283900,0.348300,0.479500)
    rgb=(0.277200,0.346100,0.478800)
    rgb=(0.270400,0.343800,0.478000)
    rgb=(0.263300,0.341300,0.476900)
    rgb=(0.256100,0.338600,0.475500)
    rgb=(0.248700,0.335700,0.474000)
    rgb=(0.241200,0.332700,0.472200)
    rgb=(0.233500,0.329500,0.470100)
    rgb=(0.225800,0.326000,0.467700)
    rgb=(0.217900,0.322400,0.465100)
    rgb=(0.210000,0.318600,0.462100)
    rgb=(0.202000,0.314700,0.458900)
    rgb=(0.194000,0.310600,0.455400)
    rgb=(0.186000,0.306300,0.451700)
    rgb=(0.178100,0.301800,0.447600)
    rgb=(0.170100,0.297300,0.443300)
    rgb=(0.162300,0.292600,0.438700)
    rgb=(0.154600,0.287800,0.433900)
    rgb=(0.147000,0.282900,0.428800)
    rgb=(0.139500,0.277900,0.423500)
    rgb=(0.132300,0.272800,0.418000)
    rgb=(0.125200,0.267700,0.412300)
    rgb=(0.118400,0.262500,0.406300)
    rgb=(0.111800,0.257200,0.400200)
    rgb=(0.105400,0.251900,0.394000)
    rgb=(0.099300,0.246600,0.387600)
    rgb=(0.093400,0.241200,0.381000)
    rgb=(0.087900,0.235900,0.374400)
    rgb=(0.082600,0.230500,0.367600)
    rgb=(0.077500,0.225100,0.360700)
    rgb=(0.072800,0.219700,0.353700)
    rgb=(0.068300,0.214300,0.346600)
    rgb=(0.064100,0.208900,0.339500)
    rgb=(0.060100,0.203500,0.332300)
    rgb=(0.056400,0.198100,0.325100)
    rgb=(0.052800,0.192800,0.317800)
    rgb=(0.049600,0.187500,0.310500)
    rgb=(0.046500,0.182100,0.303100)
    rgb=(0.043500,0.176800,0.295700)
    rgb=(0.041000,0.171500,0.288300)
    rgb=(0.038400,0.166300,0.280800)
    rgb=(0.036100,0.161000,0.273400)
    rgb=(0.033800,0.155800,0.266000)
    rgb=(0.032000,0.150600,0.258500)
    rgb=(0.030200,0.145500,0.251000)
    rgb=(0.028600,0.140300,0.243500)
    rgb=(0.027100,0.135200,0.236100)
    rgb=(0.025700,0.130100,0.228700)
    rgb=(0.024400,0.125000,0.221200)
    rgb=(0.023200,0.119900,0.213800)
    rgb=(0.022100,0.114900,0.206400)
    rgb=(0.021000,0.110000,0.199000)
    rgb=(0.020000,0.104900,0.191600)
    rgb=(0.019000,0.100000,0.184300)
    rgb=(0.018000,0.095100,0.177000)
    rgb=(0.017200,0.090000,0.169700)
    rgb=(0.015900,0.084700,0.162500)
    rgb=(0.014000,0.079100,0.155400)
    rgb=(0.011400,0.073200,0.148300)
  }
}

\pgfplotsset{ layers/my layer set/.define layer set={
        background, backishground, main, foreground
    }{
    },
    set layers=my layer set,
}
\pgfplotsset{compat=1.18}

\usepackage{tikz}

\usepackage{xspace}
\usepackage{tabularx}
\usepackage{arydshln}
\usepackage[normalem]{ulem}
\usepackage{placeins}
\usepackage[capitalize]{cleveref}
\crefname{appendix}{Appendix}{Appendices}
\Crefname{appendix}{Appendix}{Appendices}
\usepackage{siunitx}
\usepackage{intcalc}
\usepackage[inline]{enumitem}
\usepackage{float}

\makeatletter
\newcommand{\removelatexerror}{\let\@latex@error\@gobble}
\makeatother

\definecolor{DarkGreen}{rgb}{0.1,0.5,0.1}
\definecolor{DarkRed}{rgb}{0.5,0.1,0.1}
\definecolor{DarkBlue}{rgb}{0.1,0.1,0.5}

\newtheorem{theorem}{Theorem}
\newtheorem{corollary}{Corollary}
\newtheorem{lemma}{Lemma}
\newtheorem{observation}{Observation}

\newtheorem{definition}{Definition}

\newtheorem{remark}{Remark}

\newcommand{\defeq}{\vcentcolon=}

\newcommand{\property}{fully symmetric\xspace}

\newcommand{\setint}[2]{\{#1, \dots, #2\}}
\newcommand{\wt}[1]{\ensuremath{\mathrm{wt}_H(#1)}}

\DeclareMathOperator{\Li}{Li}

\newcommand{\dzgen}{\,\mathrm{d}\zgen}
\newcommand{\dr}{\,\mathrm{d}r}
\newcommand{\du}{\,\mathrm{d}u}

\newcommand{\E}{\mathbb{E}}

\newcommand{\optT}{\ensuremath{T^\star}}
\newcommand{\pdT}{\ensuremath{\widehat{T}}}
\newcommand{\optpdT}{\ensuremath{\widehat{T}^\star}}

\newcommand{\zrv}{z}
\newcommand{\zs}{s}
\newcommand{\zE}{Z}
\newcommand{\zgen}{t}

\newcommand{\bmp}{\bm{p}}
\newcommand{\bpk}{\bm{p}^{(k)}}
\newcommand{\pk}{{p}^{(k)}}

\newcommand{\bmq}{\bm{q}}

\newcommand{\bmqd}{\bm{q}^{(\delta)}}
\newcommand{\bqkd}{\bm{q}^{(\delta, k)}}
\newcommand{\qkd}{{q}^{(\delta, k)}}
\newcommand{\qd}{{q}^{(\delta)}}

\newcommand{\res}{\rho}
\newcommand{\maxres}{\widehat{\res}}

\begin{document}

\title{Random Access Expectation in DNA Data Storage: Asymptotic Analysis Through Fountain Codes}

\author{
  \thanks{This work has been submitted to the IEEE for possible publication. Copyright may be transferred without notice, after which this version may no longer be accessible.}
  \thanks{An earlier version of this paper was submitted for presentation at the 2026 IEEE Information Theory Workshop (ITW)~\cite{arxivversion}.}
  \IEEEauthorblockN{Christoph~Hofmeister, Rawad~Bitar, and Eitan~Yaakobi} \\
    \thanks{CH and RB are with the School of Computation, Information and Technology at the Technical University of Munich, Germany. 
    This project has received funding from the German Research Foundation (DFG) under Grant Agreement No. WA 3907/7-1. 
    Emails: \{christoph.hofmeister, rawad.bitar\}@tum.de}
    \thanks{EY is with the CS department of Technion---Israel Institute of Technology, Israel. 
    His research is funded by the European Union (ERC, DNAStorage, 865630). Views and opinions expressed are however those of the author(s) only and do not necessarily reflect those of the European Union or the European Research Council Executive Agency. Neither the European Union nor the granting authority can be held responsible for them.  
    Email: yaakobi@cs.technion.ac.il}
}

\markboth{Submitted to the IEEE Transactions on Information Theory}{Submitted to the IEEE Transactions on Information Theory}
\maketitle

\begin{abstract}
  The expected number of coded symbols drawn from a linear code until a desired information symbol can be decoded is studied. This quantity is of interest in DNA data storage, where it is called the random access expectation.
  If the generator matrix of the code is binary and fully symmetric, the process of sampling symbols from the code is equivalent to the encoder of an LT code. 
    Based on this observation, the random access expectation of binary fully symmetric codes under a peeling decoder is analyzed in the large blocklength limit.
    Asymptotically, a random access expectation, normalized by the number of information symbols, of $\approx 0.7869$ is achievable, while for a binary fully symmetric code under a peeling decoder it cannot be below $\pi/4 \approx 0.7854$. The analysis is extended to the case where a set of information symbols needs to be decoded.
\end{abstract}


\section{Introduction} \label{sec:introduction}

DNA shows promise as a storage medium~\cite{church2012next,goldman2013towards,grass2015robust,organick2018random,akram2018trends,panda2018dna,ceze2019molecular,gimpel2026comparison}, not least because of the high information density and the continued improvements in synthesis and sequencing technologies driven by applications in medicine and biology.
However, many technical challenges need to be overcome to make DNA data storage ready for widespread adoption.

Among them we focus on the natural tension between DNA data storage systems and efficient random access.
Due to limitations in synthesis technologies, information is encoded into many short strands of DNA, each with millions of copies, which are then stored in an unordered manner~\cite{shomorony2022information}.
As a result, it is challenging to access a specific piece of information without wasting costly reads on unrelated strands~\cite{organick2018random}.

The use of linear codes was proposed as a potential solution in~\cite{bar-lev2024cover}.
A number $k$ of information symbols is encoded into $n$ coded symbols (representing the synthesized strands) using a $k \times n$ generator matrix $\bm{G}$. 
Coded symbols are drawn with replacement one-by-one until a desired information symbol can be recovered. 
The \emph{absolute random access expectation} is the maximum (over all information symbols) expected number of required draws until the information symbol can be decoded. 
By \emph{(relative) random access expectation} we refer to the same quantity normalized by $k$, which we denote as $T(\bm{G})$.

The fundamental question we are concerned with is: How small can the random access expectation be?
Formally, what is $\optT \defeq \inf_{q, k, n} \min_{\bm{G}\in \mathbb{F}_q^{k\times n}} T(\bm{G})$?

Following~\cite{bar-lev2024cover}, the recent works~\cite{gruica2025combinatoriala,gruica2025geometry,boruchovsky2025making,bodur2025random,wang2026random,van2026weak,bar-lev2026coded} have studied the random access expectation from various perspectives. 
It has been found that many families of algebraic codes perform no better than uncoded storage~\cite{bar-lev2024cover,gruica2025combinatoriala}, achieving $T(\bm{G})=1$.
Generator matrices leading to a lower random access expectation typically contain multiple copies of columns with low Hamming weight, such as systematic MDS codes with their systematic part duplicated multiple times~\cite[Section~VI-A]{gruica2025combinatoriala}.
Additionally, generator matrices proposed in the literature typically treat all information symbols equally~\cite{bar-lev2024cover,gruica2025combinatoriala,gruica2025geometry,boruchovsky2025making,bodur2025random,wang2026random,van2026weak}. Specifically, unless $n$ is restricted (as in~\cite[Construction~1]{bar-lev2024cover}), the number of columns with a given support only depends on the size of the support.
We call such matrices \emph{\property}. This pattern was first observed in~\cite[Section~5]{boruchovsky2025making}, where \property matrices are conjectured to be optimal (at least for column weights up to $2$).
The lowest random access expectation achieved in prior literature is $\frac{\pi^2}{12} \approx 0.8225$ for $k, n, q \to \infty$~\cite[Corollary~3]{boruchovsky2025making}.

\begin{figure}[t]
  \centering
  \begin{minipage}{\figwidth}
  {

  \begin{tikzpicture}[/tikz/background rectangle/.style={fill={rgb,1:red,1.0;green,1.0;blue,1.0}, fill opacity={1.0}, draw opacity={1.0}}, show background rectangle] \begin{axis}[point meta max={nan}, point meta min={nan}, legend cell align={left}, legend columns={1}, title={}, title style={at={{(0.5,1)}}, anchor={south}, font={{\fontsize{14 pt}{18.2 pt}\selectfont}}, color={rgb,1:red,0.0;green,0.0;blue,0.0}, draw opacity={1.0}, rotate={0.0}, align={center}}, legend style={color={rgb,1:red,0.0;green,0.0;blue,0.0}, draw opacity={1.0}, line width={1}, solid, fill={rgb,1:red,1.0;green,1.0;blue,1.0}, fill opacity={1.0}, text opacity={1.0}, font={{\fontsize{8 pt}{10.4 pt}\selectfont}}, text={rgb,1:red,0.0;green,0.0;blue,0.0}, cells={anchor={center}}, at={(0.98, 0.98)}, anchor={north east}}, axis background/.style={fill={rgb,1:red,1.0;green,1.0;blue,1.0}, opacity={1.0}}, anchor={north west}, xshift={1.0mm}, yshift={-1.0mm}, width={0.9\linewidth}, height={0.6\linewidth}, scaled x ticks={false}, xlabel={$d$}, x tick style={color={rgb,1:red,0.0;green,0.0;blue,0.0}, opacity={1.0}}, x tick label style={color={rgb,1:red,0.0;green,0.0;blue,0.0}, opacity={1.0}, rotate={0}}, xlabel style={at={(ticklabel cs:0.5)}, anchor=near ticklabel, at={{(ticklabel cs:0.5)}}, anchor={near ticklabel}, font={{\fontsize{11 pt}{14.3 pt}\selectfont}}, color={rgb,1:red,0.0;green,0.0;blue,0.0}, draw opacity={1.0}, rotate={0.0}}, xmode={log}, log basis x={10}, xmajorgrids={true}, xmin={1}, xmax={100000.0}, xticklabels={{$10^{0}$,$10^{1}$,$10^{2}$,$10^{3}$,$10^{4}$,$10^{5}$}}, xtick={{1.0,10.0,100.0,1000.0,10000.0, 100000.0}}, xtick align={inside}, xticklabel style={font={{\fontsize{8 pt}{10.4 pt}\selectfont}}, color={rgb,1:red,0.0;green,0.0;blue,0.0}, draw opacity={1.0}, rotate={0.0}}, x grid style={color={rgb,1:red,0.0;green,0.0;blue,0.0}, draw opacity={0.1}, line width={0.5}, solid}, axis x line*={left}, x axis line style={color={rgb,1:red,0.0;green,0.0;blue,0.0}, draw opacity={1.0}, line width={1}, solid}, scaled y ticks={false}, ylabel={$\optpdT_d$}, y tick style={color={rgb,1:red,0.0;green,0.0;blue,0.0}, opacity={1.0}}, y tick label style={color={rgb,1:red,0.0;green,0.0;blue,0.0}, opacity={1.0}, rotate={0}}, ylabel style={at={(ticklabel cs:0.5)}, anchor=near ticklabel, at={{(ticklabel cs:0.5)}}, anchor={near ticklabel}, font={{\fontsize{11 pt}{14.3 pt}\selectfont}}, color={rgb,1:red,0.0;green,0.0;blue,0.0}, draw opacity={1.0}, rotate={0.0}}, ymajorgrids={true}, ymin={0.7848}, ymax={0.7941895058262817}, yticklabels={{$\pi/4$,$0.786$,$0.787$,$0.788$,$0.789$,$0.790$,$0.791$,$0.792$,$0.793$,$0.794$}}, ytick={{0.7853981633974483,0.786,0.787,0.788,0.789,0.79,0.791,0.792,0.793,0.794}}, ytick align={inside}, yticklabel style={font={{\fontsize{8 pt}{10.4 pt}\selectfont}}, color={rgb,1:red,0.0;green,0.0;blue,0.0}, draw opacity={1.0}, rotate={0.0}}, y grid style={color={rgb,1:red,0.0;green,0.0;blue,0.0}, draw opacity={0.1}, line width={0.5}, solid}, axis y line*={left}, y axis line style={color={rgb,1:red,0.0;green,0.0;blue,0.0}, draw opacity={1.0}, line width={1}, solid}, colorbar={false}]
    \addplot[log x thin = 100, color=vikblue, line width={1}, solid] table[x=x,y=y] {tikz/ovdata_1e5.dat}; 

    \addlegendentry {rel. random access expectation $\optpdT_d$}
    \addplot[color={rgb,1:red,0.0;green,0.0;blue,0.0}, name path={15}, draw opacity={1.0}, line width={1}, dashed]
        table[row sep={\\}]
        {
            \\
            1.0  0.7853981633974483  \\
            100000.0  0.7853981633974483  \\
        }
        ;
        \addlegendentry {\cref{lem:lowerbound} (lower bound)}
\end{axis}
\end{tikzpicture}
  }
  \caption{Achievable random access expectations for maximum column weight $2 \leq d \leq 10^5$. The lowest value from prior literature for $k\to \infty$ is $\pi^2/12 \approx 0.8225$~\cite[Corollary~3]{boruchovsky2025making}. 
    A lower bound for binary \property codes under a peeling decoder is marked by the dashed line. 
  }
  \label{fig:objval}
  \end{minipage}
\end{figure}

We find that generating reads from a binary \property generator matrix in the DNA random access setting is equivalent to generating coded symbols with the encoder of an LT code~\cite{luby2002lt}.
LT codes are a class of fountain codes~\cite{mackay2005fountain}, which are well studied in the communications and coding theory literature, especially in combination with a simple \emph{peeling decoder}~\cite{luby2001efficient}. 
For instance, appropriately designed LT codes (using a \emph{soliton distribution}) are universally capacity-approaching for binary erasure channels~\cite{luby2002lt}.

Using a result from~\cite{darling2005structure}, previously applied to LT codes in~\cite{maneva2005new,sanghavi2007intermediate}, we bound $\optpdT$, the smallest random access expectation of a binary \property generator matrix under a peeling decoder. For tractability and to control encoding and decoding cost, we constrain the maximum Hamming weight in a column of $\bm{G}$ by some integer $d$, denoting the analogous quantity by $\optpdT_d$. We determine its value to high precision for all $d$ up to $10^5$. The results can be seen in \cref{fig:objval}. We further show that $\optpdT_d \geq \pi/4$ for all $d$, as indicated by the dashed line.

Following an earlier version of the present work, \cite{bar-lev2026lower} extended this bound to optimal decoders and symmetric linear ensembles over arbitrary fields. For $d=2$, they showed that an optimal decoder asymptotically achieves the same value as the peeling decoder, and that block constructions of positive rate can approach it arbitrarily closely.

Fountain codes have been applied to DNA data storage before~\cite{erlich2017dna,lazzarotto2024storing,schwarz2024optimizing} in a different way and for a different purpose. As opposed to previous applications where the fountain codes are used as a component of the system in their capacity as an erasure-correcting code, we model the encoder of a (fixed-rate) linear block code together with the channel, which corresponds to sampling strands with replacement, jointly as the encoder of an LT code.

\emph{Organization:} 
We state the problem, define LT codes and peeling decoders, and present relevant results from the literature in \cref{sec:preliminaries}.
The equivalence between random access and LT encoders is established and applied to obtain the minimum number of draws for probabilistic recovery in \cref{sec:prelimresults}. In \cref{sec:mainresults}, the minimal random access expectation is formulated as an optimization problem, which is solved numerically and subsequently analyzed in \cref{sec:analysisoptprob}.
In \cref{sec:setrecovery} the analysis is extended to the case where multiple information symbols need to be decoded. 
Conclusions and open problems are stated in \cref{sec:conclusion}.

\begin{figure*}[t]
  \begin{subfigure}{0.32\textwidth}
    \includegraphics[width=\linewidth]{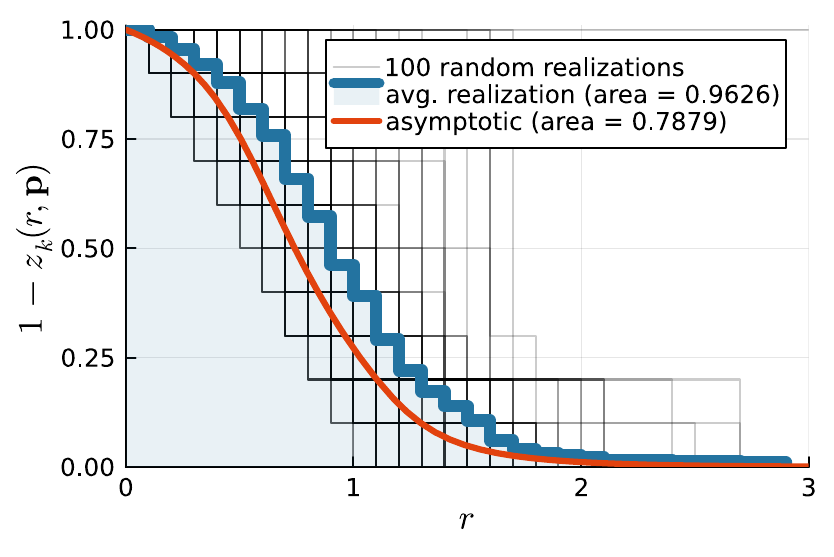}
  \caption{$k=10$}
  \end{subfigure}
  \begin{subfigure}{0.32\textwidth}
    \includegraphics[width=\linewidth]{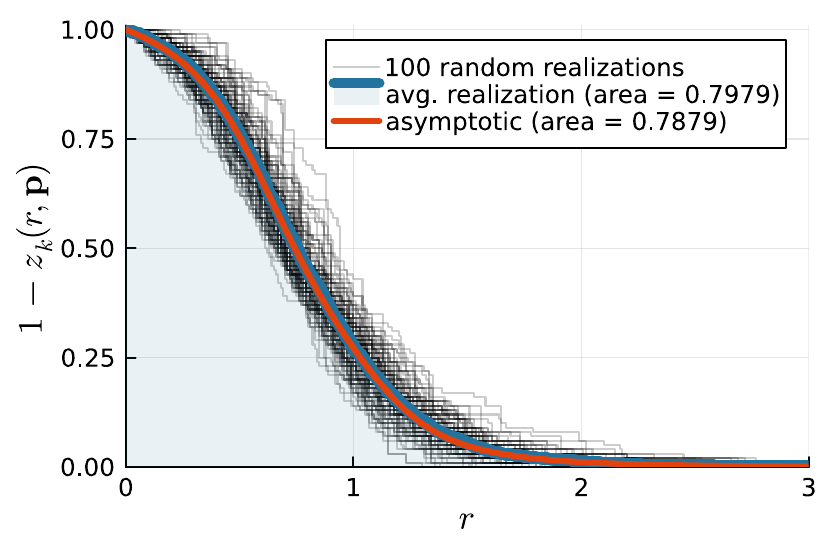}
  \caption{$k=100$}
  \end{subfigure}
  \begin{subfigure}{0.32\textwidth}
    \includegraphics[width=\linewidth]{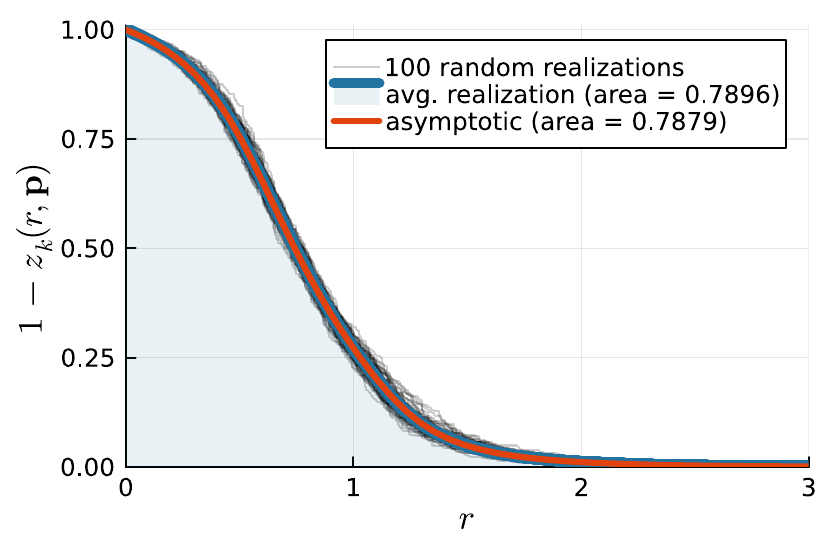}
  \caption{$k=1000$}
  \end{subfigure}
  \caption{
    Simulated decoding trajectories for an LT code with a peeling decoder. The number of undecoded information symbols over the number of received coded symbols, both normalized by $k$.
    The degree distribution that minimizes the asymptotic random access expectation for $d=10$ is used. 
    The non-zero probability values are $p_1\approx 0.205, p_2 \approx 0.727$ and $p_{10}\approx 0.067$.
    Random decoding trajectories are shown in black, their average in blue, and the asymptotic limit of the decoding curve in red.
    The areas under the blue and red curves are stated in the corresponding legend entries.
  As $k$ increases, the decoding trajectories of the LT code concentrate around their expected value~\cite{maneva2005new} and the area under the curve approaches the asymptotic random access expectation $\optpdT_{10}\approx 0.7879$.}
  \label{fig:decodingcurves}
\end{figure*}

\section{Preliminaries} \label{sec:preliminaries}
We introduce definitions and notation on the random access expectation of linear codes in the DNA random access setting, as well as notation for the asymptotic analysis of LT codes with a peeling decoder. We reproduce results from the literature on both pertaining to our results.

\subsection{Random Access Expectation}

For a given generator matrix $\bm{G} \in \mathbb{F}_q^{k \times n}$, let $\tau^{(i)}(\bm{G})$ denote the number of coded symbols, normalized by $k$, drawn one-by-one with replacement until the $i$th information symbol can be decoded under a peeling decoder, and let $\pdT(\bm{G}) = \max_i \E[\tau^{(i)}(\bm{G})]$. 
The equivalent quantity under an optimal decoder is denoted as $T(\bm{G})$. Clearly, $T(\bm{G}) \leq \pdT(\bm{G})$.

\subsection{Fountain Codes}
We reproduce descriptions of the encoder of an LT code and a peeling decoder in the following.

\emph{LT Code:}
An LT code on $k$ information symbols is defined by a probability distribution $\bm{p} = (p_1, \dots, p_d)$ over $\setint{1}{k}$, called the \emph{degree distribution}, where we use $d\leq k$ to denote the maximum degree. The encoder of the LT code generates random coded symbols independently as follows:\begin{enumerate}
  \item a degree $w$ is drawn from the degree distribution; 
  \item $w$ distinct information symbols are chosen uniformly at random among the $\binom{k}{w}$ possible choices; the coded symbol is their sum.
\end{enumerate}

\emph{Peeling Decoder:}
The peeling decoder is an efficient suboptimal decoder for LT codes under erasures. It operates as follows. Construct a bipartite graph with one left node per information symbol and one right node per received coded symbol. An edge connects a left node to a right node whenever the corresponding information symbol is included as a summand in the corresponding coded symbol. Iteratively: \begin{enumerate}
  \item select a right node of degree $1$;
  \item decode its neighbor, i.e., assign the value of the right node to it;
  \item subtract the value of the newly decoded left node from all its neighbors and remove all its edges from the graph. 
\end{enumerate}
The decoder terminates successfully if all left nodes are decoded and it fails if at step 1) there is no degree-$1$ right node.

\subsection{Asymptotic Analysis of LT Codes under Peeling}

We use the framework of~\cite{darling2005structure, maneva2005new, sanghavi2007intermediate}.
For a given positive real number $r$, assume the number of received coded symbols is distributed
as $\mathrm{Poisson}(rk)$.
For a given degree distribution $\bmp$, let the random variable $\zrv_k(r, \bmp)$ denote the fraction of decoded information symbols. Let $\zE_k(r, \bmp) = \E[\zrv_k(r, \bmp)]$ and $\zE(r, \bmp) \defeq \lim_{k\to\infty} \zE_k(r, \bmp)$, provided the limit exists.

We call a realization of $1-\zrv_k(r, \bmp)$ over $r$ a \emph{decoding trajectory} and $1-\zE_k(r, \bmp)$ the \emph{decoding curve} of the code. 

We reproduce the relevant results from the literature here for completeness. Let $p(\zgen) \defeq p_1\zgen + p_2\zgen^2 + \dots + p_d \zgen^d$ and denote its derivative as $p'(\zgen) = p_1 + 2p_2\zgen + \dots + d p_d \zgen^{d-1}.$ Define
\begin{align}
  \zs(r,\bmp) \defeq &\inf\{\zgen \in [0, 1)\!: r p'(\zgen) + \log(1-\zgen)\!<\!0\} \wedge 1, \label{eq:dandn}
\end{align}
i.e., the smallest non-negative solution to $rp'(\zgen) < -\log(1-\zgen)$, or $1$ if none exists.
The following theorem is originally due to~\cite{darling2005structure}, and was applied to LT codes in~\cite{maneva2005new,sanghavi2007intermediate}.
\begin{theorem}[{\cite[Theorem~1]{sanghavi2007intermediate}}] \label{thm:dandn}
  Let $\bpk$ be a sequence of degree distributions with limiting distribution $\bmp$ such that \begin{align}
    r p'(\zgen)+\log(1-\zgen) >0 \text{ for } 0 \leq \zgen < \zs(r, \bmp),  \label{eq:dandncond}
  \end{align}
    then, as $k\to \infty$, $\zrv_k(r, \bpk) \to \zs(r, \bmp)$ in probability.
\end{theorem}
In special cases, the result simplifies, by way of the following definitions. Define the function
\begin{align*}
  g(\zgen, \bmp) \defeq \frac{-\log(1-\zgen)}{p'(\zgen)}, 
\end{align*}
and denote its generalized inverse function by
\begin{align*}
  g^{-1}(r, \bmp) \defeq \inf\left\{\zgen \in [0, 1)\!: r \!<\!g(\zgen, \bmp)\right\}.
\end{align*}
Note that, if $g(t, \bmp)$ is strictly increasing, $g^{-1}(r, \bmp)$ is its inverse function and in general $g^{-1}(r, \bmp)$ is also the generalized inverse function of the cumulative maximum of $g(t, \bmp)$.

\begin{restatable}{corollary}{cordandn} \label{cor:dandn}
  Let $\bpk$ be a sequence of degree distributions with limiting distribution $\bmp$ with $p_1>0$, then as $k\to \infty$, $\zrv_k(r, \bpk) \to g^{-1}(r, \bmp)$ in probability, except for at most $d-1$ values of $r$.
\end{restatable}
A proof can be found in \cref{app:proofcordand}.

The following result, due to~\cite{sanghavi2007intermediate}, extends the applicability of \cref{thm:dandn} to a larger class of degree distributions by shifting a small amount of probability mass from higher degrees to degree $1$, such that it has positive mass.
For a given $\bmp$ and a small perturbation $\delta$, define the distribution $\bmq^{(\delta)}$ with $q^{(\delta)}_1 = p_1+\delta (1-p_1)$ and $q^{(\delta)}_i = (1-\delta) p_i$ for $i>1$.
\begin{lemma}[{\cite[Lemma~1]{sanghavi2007intermediate}}] \label{lem:sanghavi}
  Given $\varepsilon >0$ there exists $\delta^\prime>0$ such that for any $\delta<\delta^\prime$ it holds that\begin{align*}
    \zs(r,\bmp) \leq \zE\left(\frac{r}{1-\delta}, \bmq^{(\delta)}\right) \leq \zs(r, \bmp) + \varepsilon.
  \end{align*}
\end{lemma}

Even if \cref{thm:dandn} is not directly applicable for a given $\bmp$, a distribution arbitrarily close to $\bmp$ has an asymptotic expected fraction of decoded symbols arbitrarily close to $\zs(r, \bmp)$ at a slightly increased $r$. 

\section{Equivalence to LT Codes} \label{sec:prelimresults}

We observe an equivalence between drawing coded symbols from a binary \property code and the encoder of an LT code, and frame the random access expectation in terms of the area under the decoding curve. As an immediate consequence of the equivalence we characterize the minimum number of coded symbols that need to be drawn such that a prespecified information symbol or set of information symbols can be recovered with a certain probability, using a result from the literature on LT codes.

\subsection{Equivalence of LT Codes and Binary Fully Symmetric Codes}

Let a binary \property matrix be defined as follows.
\begin{definition}
  A binary generator matrix $\bm{G} \in \{0,1\}^{k \times n}$ is called \property if for each $w\in \setint{1}{k}$ it holds that each binary vector $\bm{c} \in \mathbb{F}_2^{k}$ with $\wt{\bm{c}} = w$ appears equally often as a column of $\bm{G}$. 
\end{definition}
A binary \property matrix is determined, up to column permutations, by the empirical distribution of Hamming weights in its columns.
We denote this distribution as $\bmp = (p_1, \dots, p_d),$ for some maximum weight $d\leq k$, where $p_i$ is the relative fraction of columns with weight $i$. 
We justify the use of $\bmp$ for both a degree distribution and an empirical column-weight distribution with the following observation.

\begin{observation} \label{obs:equivalence}
  The process of obtaining reads from a binary fully symmetric generator matrix with empirical column-weight distribution $\bmp$ is identical to an LT encoder with degree distribution $\bmp$.
\end{observation}
Drawing a column from $\bm{G}$ amounts to: 1) drawing a column weight $w$ from $\bmp$ and then 2) drawing a weight-$w$ vector uniformly at random, where the ones in the drawn column select information symbols which are then summed. These steps correspond exactly to steps 1) and 2) of the LT encoder.

Any degree distribution $\bmp$ can be approximated arbitrarily closely by a rational distribution and thus realized by a binary \property generator matrix $\bm{G}$, provided $n$ is large enough, cf.~\cite[Problem~2]{boruchovsky2025making}.
Accordingly, in the following we focus on degree distributions rather than generator matrices.

\subsection{Random Access Expectation and the Decoding Curve}

Consider a binary \property generator matrix $\bm{G}$ with column-weight distribution $\bmp$.
Since the encoding is symmetric in the information symbols, we focus on any particular information symbol. Let $\tau_k(\bmp)$ denote the time when a particular information symbol can be decoded by a peeling decoder, when the number of received coded symbols is distributed as $\mathrm{Poisson}(rk)$. 
Since $\tau_k(\bmp)$ is the sum of $k\tau^{(i)}(\bm{G})$ independent $\mathrm{Exp}(k)$-distributed random variables, that are also independent of $\tau^{(i)}(\bm{G})$, it holds that 
\begin{align*}
  \pdT_k(\bmp) \defeq \E[\tau_k(\bmp)] = \E[k\tau^{(i)}(\bm{G})] \cdot \frac{1}{k} = \pdT(\bm{G}).
\end{align*}
If the fraction of decoded information symbols is $\zrv_k$, the probability that the particular information symbol is decoded equals $\zrv_k$, by symmetry. By marginalizing over $\zrv_k$ we have $\Pr[\tau_k(\bmp)\leq r] = \E[\Pr[\tau_k(\bmp)\leq r \;|\; \zrv_k]] = \E[\zrv_k] = \zE_k(r, \bmp)$. And so by the tail-sum formula for expectation, we have
\begin{align}
  \pdT_k(\bm{p}) = \int_0^\infty \Pr[\tau_k(\bmp)>r] \dr = \int_0^\infty 1-\zE_k(r, \bmp) \dr. \label{eq:tailsum}
\end{align}
In other words, the random access expectation is the area under the decoding curve.
A numerical experiment illustrating these concepts is shown in \cref{fig:decodingcurves}.

\subsection{Probabilistic Recovery from Binary Fully Symmetric Codes}

As a brief intermission, we point out a direct consequence of \cref{obs:equivalence}, namely that the result of~\cite{sanghavi2007intermediate} can be applied to DNA data storage.

The work~\cite{sanghavi2007intermediate} derives the minimum (approximate asymptotic) number of required received symbols from an LT code to recover a given fraction of the information symbols. By the identity $\zE_k(r, \bmp) = \Pr[\tau_k(\bmp)\leq r]$, this quantity is equivalent to the minimum required number of draws $r$ (normalized by $k$) such that a desired information symbol can be decoded with a given probability $\zgen$.
In~\cite{sanghavi2007intermediate}, this number is determined exactly for $t \leq 2/3$ and up to a small margin for all other values.
Up to $t=1/2$, a degree distribution with $p_1=1$ is optimal, i.e., uncoded storage. For $1/2 \leq t \leq 2/3$ a degree distribution with $p_2 \to 1$ is optimal, i.e., all pairwise sums in equal proportion (with a positive but asymptotically negligible fraction of degree-$1$ symbols). The near-optimal results for $t>2/3$ are achieved using truncated soliton distributions.

The equivalent result for when a set of $s$ prespecified information symbols needs to be recovered is obtained from the results of~\cite{sanghavi2007intermediate} by applying the map $\zgen \to \zgen^{s}$ to the probability, since asymptotically membership of information symbols in the set of decoded symbols is independent. This relation is justified more rigorously in \cref{sec:setrecovery} (see \eqref{eq:inclusionindependent}), where the recovery of sets of information symbols is studied in terms of the random access expectation.

\Cref{fig:sanghavisets} shows the maximum asymptotic probability of decoding a prespecified set of $s$ information symbols using a peeling decoder for $s$ up to $10$.

\begin{figure}
  \centering
  \begin{minipage}{\figwidth}
  \input{tikz/sanghavi_set_plot.tex}
  \caption{The maximum asymptotic probability of decoding a prespecified set $\mathcal{S}$ of $s$ symbols using a peeling decoder. The probability is maximized for each value of $r$ separately. The curves are exact for $r \leq \frac{3}{4}\log 3 \approx 0.82$, where the $s=1$ curve reaches $2/3$, and close approximations thereafter. The dark blue curve for $s=1$ is the main result of~\cite{sanghavi2007intermediate}; the other curves are the $s$th powers of the dark blue curve.}
  \label{fig:sanghavisets}
  \end{minipage}
\end{figure}

\section{Minimizing the Asymptotic Random Access Expectation} \label{sec:mainresults}

We formulate an optimization problem and show that its solution corresponds to the optimal asymptotic random access expectation under peeling over all distributions $\bmp$ with bounded maximum degree, if certain conditions are met.  

Let the smallest asymptotic random access expectation among binary \property generator matrices with maximum column weight $d$ under a peeling decoder be defined as
\begin{align*}
  \optpdT_d \defeq \inf\left\{\liminf_{k\to\infty} \pdT_k(\bpk) : \bpk \in \Delta^d \text{ for all } k\right\},
\end{align*}
where $\Delta^d = \{\bm{p} \in \mathbb{R}^d : \sum_i p_i = 1, p_i\geq 0 \text{ for } 1\leq i\leq d\}$ denotes the $d$-dimensional probability simplex, and let $\optpdT \defeq \inf_{d} \optpdT_d$.

\subsection{Asymptotic Random Access Expectation}
Using \cref{cor:dandn} and perturbations as in \cref{lem:sanghavi} we establish the following result.
Recall that $g(\zgen, \bm{p}) = \frac{-\log(1-\zgen)}{p'(\zgen)}$ and define the function
\begin{align*}
  f(\bm{p}) \defeq \int_0^1 g(\zgen, \bm{p}) \dzgen.
\end{align*}

\begin{restatable}{lemma}{lemseqconv} \label{lem:seqconv}
  Let $\bpk$ be a sequence of degree distributions with $\lim_{k\to\infty} \bpk = \bm{p}$.
    If $p_1>0$ and $g(\zgen, \bmp)$ is strictly increasing for $0<\zgen<1$, then
      \begin{align}
        \lim_{k\to\infty} \pdT_k(\bpk) = f(\bm{p}). \label{eq:achiev}
      \end{align}
    For any $\bm{p}$ it holds that\begin{align}
        \liminf_{k\to\infty} \pdT_k(\bpk) \geq f(\bm{p}). \label{eq:converse}
      \end{align}
\end{restatable}
The proof is provided in \cref{app:prooflemseqconv}.
Informally, the asymptotic decoding curve is given by $1-g^{-1}(r, \bmp)$, the area under which equals the integral of the cumulative maximum of $g(t, \bmp)$ from $0$ to $1$ (see \cref{fig:decodingcurves}). This integral equals $f(\bmp)$ exactly when $g(t, \bmp)$ is non-decreasing and exceeds $f(\bmp)$ otherwise.

\begin{remark}
  The analysis is in terms of sequences of distributions to ensure the lower bounds apply to the case $p_1\to0$, such as the soliton distribution~\cite{luby2002lt}. However, $p_1\to0$ generally does not yield a low random access expectation. For instance, soliton distributions have a square asymptotic decoding curve and give $\pdT_k(\bpk) \to 1$, like uncoded storage.
\end{remark}

\newcommand{\floatextracurves}{%
\begin{figure*}[t]
  \begin{minipage}{\figwidth}
  \input{tikz/extracurves.tex}
  \captionsetup{margin={0pt,1pc}}
  \caption{Asymptotic decoding curves achieving $\optpdT_d$ for various values of $d$. We observe similar S-shaped, strictly decreasing curves across $d$.}
  \label{fig:extracurves}
  \end{minipage} \hfill
  \begin{minipage}{\figwidth}
  \begin{tikzpicture}
\begin{axis}[point meta max={nan}, point meta min={nan}, legend cell align={left}, legend columns={1}, title={}, title style={at={{(0.5,1)}}, anchor={south}, font={{\fontsize{14 pt}{18.2 pt}\selectfont}}, color={rgb,1:red,0.0;green,0.0;blue,0.0}, draw opacity={1.0}, rotate={0.0}, align={center}}, 
  legend style={color={rgb,1:red,0.0;green,0.0;blue,0.0}, draw opacity={1.0}, line width={1}, solid, fill={rgb,1:red,1.0;green,1.0;blue,1.0}, fill opacity={1.0}, text opacity={1.0}, font={{\fontsize{8 pt}{10.4 pt}\selectfont}}, text={rgb,1:red,0.0;green,0.0;blue,0.0}, cells={anchor={center}}, at={(0.98, 0.1)}, anchor={south east}}, 
  axis background/.style={fill={rgb,1:red,1.0;green,1.0;blue,1.0}, opacity={1.0}}, anchor={north west}, xshift={1.0mm}, yshift={-1.0mm}, width={0.8\linewidth}, height={0.7\linewidth}, scaled x ticks={false}, xlabel={$i$}, x tick style={color={rgb,1:red,0.0;green,0.0;blue,0.0}, opacity={1.0}}, x tick label style={color={rgb,1:red,0.0;green,0.0;blue,0.0}, opacity={1.0}, rotate={0}}, xlabel style={at={(ticklabel cs:0.5)}, anchor=near ticklabel, at={{(ticklabel cs:0.5)}}, anchor={near ticklabel}, font={{\fontsize{11 pt}{14.3 pt}\selectfont}}, color={rgb,1:red,0.0;green,0.0;blue,0.0}, draw opacity={1.0}, rotate={0.0}}, xmode={log}, log basis x={10}, xmajorgrids={false}, xmin={0.781160263981159}, xmax={116.36536597078177}, xticklabels={{$10^{0}$,$10^{1}$,$10^{2}$}}, xtick={{1.0,10.0,100.0}}, xtick align={inside}, xticklabel style={font={{\fontsize{8 pt}{10.4 pt}\selectfont}}, color={rgb,1:red,0.0;green,0.0;blue,0.0}, draw opacity={1.0}, rotate={0.0}}, x grid style={color={rgb,1:red,0.0;green,0.0;blue,0.0}, draw opacity={0.1}, line width={0.5}, solid}, axis x line*={left}, x axis line style={color={rgb,1:red,0.0;green,0.0;blue,0.0}, draw opacity={1.0}, line width={1}, solid}, scaled y ticks={false}, ylabel={derivative + objective}, y tick style={color={rgb,1:red,0.0;green,0.0;blue,0.0}, opacity={1.0}}, y tick label style={color={rgb,1:red,0.0;green,0.0;blue,0.0}, opacity={1.0}, rotate={0}}, ylabel style={at={(ticklabel cs:0.5)}, anchor=near ticklabel, at={{(ticklabel cs:0.5)}}, anchor={near ticklabel}, font={{\fontsize{11 pt}{14.3 pt}\selectfont}}, color={rgb,1:red,0.0;green,0.0;blue,0.0}, draw opacity={1.0}, rotate={0.0}}, ymajorgrids={false}, ymin={-0.006853472784658677}, ymax={0.006853472784658677}, yticklabels={{$-0.0050$,$-0.0025$,$0.0000$,$0.0025$,$0.0050$}}, ytick={{-0.005,-0.0025,0.0,0.0025,0.005}}, ytick align={inside}, yticklabel style={font={{\fontsize{8 pt}{10.4 pt}\selectfont}}, color={rgb,1:red,0.0;green,0.0;blue,0.0}, draw opacity={1.0}, rotate={0.0}}, y grid style={color={rgb,1:red,0.0;green,0.0;blue,0.0}, draw opacity={0.1}, line width={0.5}, solid}, axis y line*={left}, y axis line style={color={rgb,1:red,0.0;green,0.0;blue,0.0}, draw opacity={1.0}, line width={1}, solid}, colorbar={false}]
    \addlegendimage{legend image code/.code={\fill[vikred] (0.26cm,-0.1cm) rectangle (0.34cm,0.1cm);}}
    \addlegendentry {$p^{\star}_i$}
    \addplot[log x thin=100, color=vikblue, draw opacity={1.0}, line width={1}, solid]
        table {tikz/derivs_d100_data/block_000.dat}
        ;
        \addlegendentry {$\frac{\partial f(\bm{p}^\star)}{\partial p_i} + f(\bm{p}^\star)$}
    \addplot[color={rgb,1:red,0.0;green,0.0;blue,0.0}, draw opacity={1.0}, line width={1}, solid, forget plot]
        table {tikz/derivs_d100_data/block_001.dat}
        ;
    \addplot[color=vikblue, only marks, draw opacity={1.0}, line width={0}, solid, mark={*}, mark size={1.5 pt}, mark repeat={1}, mark options={color={rgb,1:red,0.0;green,0.0;blue,0.0}, draw opacity={1.0}, fill=vikblue, fill opacity={1.0}, line width={0.0}, rotate={0}, solid}, forget plot]
        table {tikz/derivs_d100_data/block_002.dat}
        ;
\end{axis}
\begin{axis}[point meta max={nan}, point meta min={nan}, legend cell align={left}, legend columns={1}, title={}, title style={at={{(0.5,1)}}, anchor={south}, font={{\fontsize{14 pt}{18.2 pt}\selectfont}}, color={rgb,1:red,0.0;green,0.0;blue,0.0}, draw opacity={1.0}, rotate={0.0}, align={center}}, legend style={color={rgb,1:red,0.0;green,0.0;blue,0.0}, draw opacity={1.0}, line width={1}, solid, fill={rgb,1:red,1.0;green,1.0;blue,1.0}, fill opacity={1.0}, text opacity={1.0}, font={{\fontsize{8 pt}{10.4 pt}\selectfont}}, text={rgb,1:red,0.0;green,0.0;blue,0.0}, cells={anchor={center}}, at={(0.98, 0.98)}, anchor={north east}}, axis background/.style={fill={rgb,1:red,0.0;green,0.0;blue,0.0}, opacity={0.0}}, anchor={north west}, xshift={1.0mm}, yshift={-1.0mm}, width={0.8\linewidth}, height={0.7\linewidth}, scaled x ticks={false}, xlabel={$i$}, x tick style={color={rgb,1:red,0.0;green,0.0;blue,0.0}, opacity={1.0}}, x tick label style={color={rgb,1:red,0.0;green,0.0;blue,0.0}, opacity={1.0}, rotate={0}}, xlabel style={at={(ticklabel cs:0.5)}, anchor=near ticklabel, at={{(ticklabel cs:0.5)}}, anchor={near ticklabel}, font={{\fontsize{11 pt}{14.3 pt}\selectfont}}, color={rgb,1:red,0.0;green,0.0;blue,0.0}, draw opacity={1.0}, rotate={0.0}}, xmode={log}, log basis x={10}, xmajorticks={false}, xmajorgrids={false}, xmin={0.781160263981159}, xmax={116.36536597078177}, axis x line*={left}, separate axis lines, x axis line style={{draw opacity = 0}}, xlabel={}, scaled y ticks={false}, ylabel={$p_i$}, y tick style={color={rgb,1:red,0.0;green,0.0;blue,0.0}, opacity={1.0}}, y tick label style={color={rgb,1:red,0.0;green,0.0;blue,0.0}, opacity={1.0}, rotate={0}}, ylabel style={at={(ticklabel cs:0.5)}, anchor=near ticklabel, at={{(ticklabel cs:0.5)}}, anchor={near ticklabel}, font={{\fontsize{11 pt}{14.3 pt}\selectfont}}, color={rgb,1:red,0.0;green,0.0;blue,0.0}, draw opacity={1.0}, rotate={0.0}}, ymajorgrids={true}, ymin={-0.8342248854491865}, ymax={0.8342248854491865}, ytick={{0.0,0.3,0.6}}, ytick align={inside}, yticklabel style={font={{\fontsize{8 pt}{10.4 pt}\selectfont}}, color={rgb,1:red,0.0;green,0.0;blue,0.0}, draw opacity={1.0}, rotate={0.0}}, y grid style={color={rgb,1:red,0.0;green,0.0;blue,0.0}, draw opacity={0.1}, line width={0.5}, solid}, axis y line*={right}, y axis line style={color={rgb,1:red,0.0;green,0.0;blue,0.0}, draw opacity={1.0}, line width={1}, solid}, colorbar={false}]
    \addplot[color=vikred, draw opacity={1.0}, line width={4}, solid, ycomb, forget plot]
        table {tikz/derivs_d100_data/stems.dat}
        ;
\end{axis}
\end{tikzpicture}
  \captionsetup{margin={0pt,1pc}}
  \caption{The degree distribution $\bmp^{\star,100}$ and the sum of the objective value and the derivative of the objective w.r.t. $i$ (extended to the reals). The KKT conditions state that the support of $\bm{p}$ is limited to those $i$, where this sum crosses zero.
    The degree distribution is given by 
    $p^{\star,100}_1\approx     0.19363$, 
    $p^{\star,100}_2\approx     0.75839$, 
    $p^{\star,100}_{14}\approx  0.00004$, 
    $p^{\star,100}_{15}\approx  0.04198$, and
    $p^{\star,100}_{100}\approx 0.00596$. 
  }
\label{fig:derivs100}
  \end{minipage}
\end{figure*}
}
\begin{figure*}[t]
  \begin{minipage}{\figwidth}
  \input{tikz/extracurves.tex}
  \captionsetup{margin={0pt,1pc}}
  \caption{Asymptotic decoding curves achieving $\optpdT_d$ for various values of $d$. We observe similar S-shaped, strictly decreasing curves across $d$.}
  \label{fig:extracurves}
  \end{minipage} \hfill
  \begin{minipage}{\figwidth}
  \begin{tikzpicture}
\begin{axis}[point meta max={nan}, point meta min={nan}, legend cell align={left}, legend columns={1}, title={}, title style={at={{(0.5,1)}}, anchor={south}, font={{\fontsize{14 pt}{18.2 pt}\selectfont}}, color={rgb,1:red,0.0;green,0.0;blue,0.0}, draw opacity={1.0}, rotate={0.0}, align={center}}, 
  legend style={color={rgb,1:red,0.0;green,0.0;blue,0.0}, draw opacity={1.0}, line width={1}, solid, fill={rgb,1:red,1.0;green,1.0;blue,1.0}, fill opacity={1.0}, text opacity={1.0}, font={{\fontsize{8 pt}{10.4 pt}\selectfont}}, text={rgb,1:red,0.0;green,0.0;blue,0.0}, cells={anchor={center}}, at={(0.98, 0.1)}, anchor={south east}}, 
  axis background/.style={fill={rgb,1:red,1.0;green,1.0;blue,1.0}, opacity={1.0}}, anchor={north west}, xshift={1.0mm}, yshift={-1.0mm}, width={0.8\linewidth}, height={0.7\linewidth}, scaled x ticks={false}, xlabel={$i$}, x tick style={color={rgb,1:red,0.0;green,0.0;blue,0.0}, opacity={1.0}}, x tick label style={color={rgb,1:red,0.0;green,0.0;blue,0.0}, opacity={1.0}, rotate={0}}, xlabel style={at={(ticklabel cs:0.5)}, anchor=near ticklabel, at={{(ticklabel cs:0.5)}}, anchor={near ticklabel}, font={{\fontsize{11 pt}{14.3 pt}\selectfont}}, color={rgb,1:red,0.0;green,0.0;blue,0.0}, draw opacity={1.0}, rotate={0.0}}, xmode={log}, log basis x={10}, xmajorgrids={false}, xmin={0.781160263981159}, xmax={116.36536597078177}, xticklabels={{$10^{0}$,$10^{1}$,$10^{2}$}}, xtick={{1.0,10.0,100.0}}, xtick align={inside}, xticklabel style={font={{\fontsize{8 pt}{10.4 pt}\selectfont}}, color={rgb,1:red,0.0;green,0.0;blue,0.0}, draw opacity={1.0}, rotate={0.0}}, x grid style={color={rgb,1:red,0.0;green,0.0;blue,0.0}, draw opacity={0.1}, line width={0.5}, solid}, axis x line*={left}, x axis line style={color={rgb,1:red,0.0;green,0.0;blue,0.0}, draw opacity={1.0}, line width={1}, solid}, scaled y ticks={false}, ylabel={derivative + objective}, y tick style={color={rgb,1:red,0.0;green,0.0;blue,0.0}, opacity={1.0}}, y tick label style={color={rgb,1:red,0.0;green,0.0;blue,0.0}, opacity={1.0}, rotate={0}}, ylabel style={at={(ticklabel cs:0.5)}, anchor=near ticklabel, at={{(ticklabel cs:0.5)}}, anchor={near ticklabel}, font={{\fontsize{11 pt}{14.3 pt}\selectfont}}, color={rgb,1:red,0.0;green,0.0;blue,0.0}, draw opacity={1.0}, rotate={0.0}}, ymajorgrids={false}, ymin={-0.006853472784658677}, ymax={0.006853472784658677}, yticklabels={{$-0.0050$,$-0.0025$,$0.0000$,$0.0025$,$0.0050$}}, ytick={{-0.005,-0.0025,0.0,0.0025,0.005}}, ytick align={inside}, yticklabel style={font={{\fontsize{8 pt}{10.4 pt}\selectfont}}, color={rgb,1:red,0.0;green,0.0;blue,0.0}, draw opacity={1.0}, rotate={0.0}}, y grid style={color={rgb,1:red,0.0;green,0.0;blue,0.0}, draw opacity={0.1}, line width={0.5}, solid}, axis y line*={left}, y axis line style={color={rgb,1:red,0.0;green,0.0;blue,0.0}, draw opacity={1.0}, line width={1}, solid}, colorbar={false}]
    \addlegendimage{legend image code/.code={\fill[vikred] (0.26cm,-0.1cm) rectangle (0.34cm,0.1cm);}}
    \addlegendentry {$p^{\star}_i$}
    \addplot[log x thin=100, color=vikblue, draw opacity={1.0}, line width={1}, solid]
        table {tikz/derivs_d100_data/block_000.dat}
        ;
        \addlegendentry {$\frac{\partial f(\bm{p}^\star)}{\partial p_i} + f(\bm{p}^\star)$}
    \addplot[color={rgb,1:red,0.0;green,0.0;blue,0.0}, draw opacity={1.0}, line width={1}, solid, forget plot]
        table {tikz/derivs_d100_data/block_001.dat}
        ;
    \addplot[color=vikblue, only marks, draw opacity={1.0}, line width={0}, solid, mark={*}, mark size={1.5 pt}, mark repeat={1}, mark options={color={rgb,1:red,0.0;green,0.0;blue,0.0}, draw opacity={1.0}, fill=vikblue, fill opacity={1.0}, line width={0.0}, rotate={0}, solid}, forget plot]
        table {tikz/derivs_d100_data/block_002.dat}
        ;
\end{axis}
\begin{axis}[point meta max={nan}, point meta min={nan}, legend cell align={left}, legend columns={1}, title={}, title style={at={{(0.5,1)}}, anchor={south}, font={{\fontsize{14 pt}{18.2 pt}\selectfont}}, color={rgb,1:red,0.0;green,0.0;blue,0.0}, draw opacity={1.0}, rotate={0.0}, align={center}}, legend style={color={rgb,1:red,0.0;green,0.0;blue,0.0}, draw opacity={1.0}, line width={1}, solid, fill={rgb,1:red,1.0;green,1.0;blue,1.0}, fill opacity={1.0}, text opacity={1.0}, font={{\fontsize{8 pt}{10.4 pt}\selectfont}}, text={rgb,1:red,0.0;green,0.0;blue,0.0}, cells={anchor={center}}, at={(0.98, 0.98)}, anchor={north east}}, axis background/.style={fill={rgb,1:red,0.0;green,0.0;blue,0.0}, opacity={0.0}}, anchor={north west}, xshift={1.0mm}, yshift={-1.0mm}, width={0.8\linewidth}, height={0.7\linewidth}, scaled x ticks={false}, xlabel={$i$}, x tick style={color={rgb,1:red,0.0;green,0.0;blue,0.0}, opacity={1.0}}, x tick label style={color={rgb,1:red,0.0;green,0.0;blue,0.0}, opacity={1.0}, rotate={0}}, xlabel style={at={(ticklabel cs:0.5)}, anchor=near ticklabel, at={{(ticklabel cs:0.5)}}, anchor={near ticklabel}, font={{\fontsize{11 pt}{14.3 pt}\selectfont}}, color={rgb,1:red,0.0;green,0.0;blue,0.0}, draw opacity={1.0}, rotate={0.0}}, xmode={log}, log basis x={10}, xmajorticks={false}, xmajorgrids={false}, xmin={0.781160263981159}, xmax={116.36536597078177}, axis x line*={left}, separate axis lines, x axis line style={{draw opacity = 0}}, xlabel={}, scaled y ticks={false}, ylabel={$p_i$}, y tick style={color={rgb,1:red,0.0;green,0.0;blue,0.0}, opacity={1.0}}, y tick label style={color={rgb,1:red,0.0;green,0.0;blue,0.0}, opacity={1.0}, rotate={0}}, ylabel style={at={(ticklabel cs:0.5)}, anchor=near ticklabel, at={{(ticklabel cs:0.5)}}, anchor={near ticklabel}, font={{\fontsize{11 pt}{14.3 pt}\selectfont}}, color={rgb,1:red,0.0;green,0.0;blue,0.0}, draw opacity={1.0}, rotate={0.0}}, ymajorgrids={true}, ymin={-0.8342248854491865}, ymax={0.8342248854491865}, ytick={{0.0,0.3,0.6}}, ytick align={inside}, yticklabel style={font={{\fontsize{8 pt}{10.4 pt}\selectfont}}, color={rgb,1:red,0.0;green,0.0;blue,0.0}, draw opacity={1.0}, rotate={0.0}}, y grid style={color={rgb,1:red,0.0;green,0.0;blue,0.0}, draw opacity={0.1}, line width={0.5}, solid}, axis y line*={right}, y axis line style={color={rgb,1:red,0.0;green,0.0;blue,0.0}, draw opacity={1.0}, line width={1}, solid}, colorbar={false}]
    \addplot[color=vikred, draw opacity={1.0}, line width={4}, solid, ycomb, forget plot]
        table {tikz/derivs_d100_data/stems.dat}
        ;
\end{axis}
\end{tikzpicture}
  \captionsetup{margin={0pt,1pc}}
  \caption{The degree distribution $\bmp^{\star,100}$ and the sum of the objective value and the derivative of the objective w.r.t. $i$ (extended to the reals). The KKT conditions state that the support of $\bm{p}$ is limited to those $i$, where this sum crosses zero.
    The degree distribution is given by 
    $p^{\star,100}_1\approx     0.19363$, 
    $p^{\star,100}_2\approx     0.75839$, 
    $p^{\star,100}_{14}\approx  0.00004$, 
    $p^{\star,100}_{15}\approx  0.04198$, and
    $p^{\star,100}_{100}\approx 0.00596$. 
  }
\label{fig:derivs100}
  \end{minipage}
\end{figure*}

\subsection{Optimality of the Minimizer}

We formulate the optimization problem\begin{align}
  \min_{\bm{p}\in \Delta^d} f(\bm{p}). \label{eq:optprob}
\end{align}

It is shown in \cref{sec:analysisoptprob} that \eqref{eq:optprob} is a convex optimization problem (see \cref{lem:convex}). For a given $d$, let $\bm{p}^{\star,d}$ denote the unique optimum of \eqref{eq:optprob}. Then, under certain conditions, which we verify case-by-case, $\bm{p}^{\star,d}$ minimizes the random access expectation among all degree distributions as stated in the following.
\begin{theorem} \label{thm:ifgoodthenperfect}
  For any $d$, if it holds that $p_1^{\star,d}>0$ and $g(\zgen, \bm{p}^{\star,d})$ is strictly increasing for $0<\zgen<1$, then\begin{align*}
    \optpdT_d = f(\bm{p}^{\star,d}).
  \end{align*}
  Conversely, for any sequence of probability distributions $\bpk$ with limiting distribution $\bmp \neq \bm{p}^{\star,d}$ it holds that 
  \begin{align*}
    \liminf_{k\to\infty} \pdT_k(\bpk) > f(\bm{p}^{\star,d}).
  \end{align*}
\end{theorem}
\begin{proof}
  We apply \cref{lem:seqconv}. The equality follows from \eqref{eq:achiev} applied to $\bmp^{\star,d}$ (formally, to the constant sequence $\bpk = \bmp^{\star,d}$). For any sequence $\bpk$ with limiting distribution $\bmp$ we have $\liminf_{k\to\infty} \pdT_k(\bpk)\geq f(\bmp)$ by \eqref{eq:converse}, and $f(\bmp)> f(\bmp^{\star,d})$ if $\bmp \neq \bmp^{\star,d}$ by the strict convexity of \eqref{eq:optprob}.
\end{proof}

Whenever $p_1>0$ and $g(\zgen,\bmp)$ is strictly increasing at the optimum, the corresponding code attains the minimum random access expectation among all binary \property codes under peeling. Additionally, any lower bound on $f(\bmp^{\star,d})$ is a lower bound on $\optpdT_d$. 

We find $\bmp^{\star,d}$ for all $d$ up to $10^5$ by convex optimization and verify that $p_1^{\star,d}>0$ and $g(\zgen, \bm{p}^{\star,d})$ is strictly increasing for $0<\zgen<1$ in all cases\footnote{For large $d$, $\min_t g'(t, \bmp^{\star, d}) \approx 0.87$ at $t\approx0.385$.}. 
This can also be seen in the optimal decoding curves: a strictly decreasing decoding curve corresponds to a strictly increasing $g(\zgen, \bmp^{\star,d})$.
Generally, we observe S-shaped and strictly decreasing decoding curves across values of $d$, as seen in \cref{fig:extracurves}.
The random access expectations $\optpdT_d$ are shown in \cref{fig:objval}, and decrease to $\optpdT_{10^5} \approx 0.7869$ at $d=10^5$.
The numerical solution of the optimization problem is detailed in \cref{sec:numericalsol}.

\section{Analysis of the Optimization Problem} \label{sec:analysisoptprob}

We analyze the optimization problem \eqref{eq:optprob} to characterize the degree distributions that solve it. We conjecture that there is no closed-form expression for $\bmp^{\star,d}$ for finite $d$.

\subsection{Convexity of the Optimization Problem}

Recall the objective function\begin{align*}
  f(\bmp) = \int_0^1 \frac{-\log(1-\zgen)}{p'(\zgen)} \dzgen.
\end{align*}

The domain of the optimization problem \eqref{eq:optprob} is the $d$-dimensional probability simplex, i.e., a convex set. 
Thus, to show that \eqref{eq:optprob} is a convex optimization problem, it suffices to show that the objective function is convex.
The first and second derivatives of $f(\bm{p})$ are given by
  \begin{align*}
    \frac{\partial f(\bmp )}{\partial p_i} &= \int_0^1 \frac{i \zgen^{i-1} \log(1-\zgen)}{p'(\zgen)^2} \dzgen \\
    \frac{\partial^2 f(\bmp )}{\partial p_i \partial p_j} &= - \int_0^1 \frac{2 i \zgen^{i-1} j \zgen^{j-1} \log(1-\zgen)}{p'(\zgen)^3} \dzgen.
  \end{align*}
\begin{restatable}{lemma}{lemconvex} \label{lem:convex}
  The function $f(\bmp )$ is strictly convex.
\end{restatable}
\begin{proof}
  Let $\bm{H}$ be the Hessian matrix of $f(\bmp )$, i.e., the $d\times d$ matrix given by $H_{i, j} = \frac{\partial^2 f(p_1, \dots, p_d)}{\partial p_i \partial p_j}.$
  To see that $\bm{H}$ is positive definite, let $\bm{a} \in \mathbb{R}^d$ be a non-zero vector, define $\bm{v}(\zgen)\in \mathbb{R}^{d}$ as $v_i(\zgen)=i\zgen^{i-1}$, and consider
  \begin{align*}
    \bm{a}^T \bm{H} \bm{a} &= \sum_{i, j} a_i a_j \frac{\partial^2 f(\bmp )}{\partial p_i \partial p_j} \\
            &= \sum_{i, j} a_i a_j  \int_0^1 \frac{- 2 i \zgen^{i-1} j \zgen^{j-1} \log(1-\zgen)}{p'(\zgen)^3} \dzgen \\
            &= \int_0^1 \frac{-2\log(1-\zgen)}{p'(\zgen)^3} \sum_{i, j} a_i a_j i \zgen^{i-1} j \zgen^{j-1} \dzgen \\
            &= \int_0^1 \frac{-2\log(1-\zgen)}{p'(\zgen)^3} \bm{a}^T (\bm{v}(\zgen) \bm{v}(\zgen)^T) \bm{a} \,\dzgen \\
            &= \int_0^1 \frac{-2\log(1-\zgen)}{p'(\zgen)^3} (\bm{a}^T \bm{v}(\zgen))^2 \dzgen \\
            &> 0,
  \end{align*}
  where the final inequality holds since $\frac{-2\log(1-\zgen)}{p'(\zgen)^3}>0$ for $t>0$ and $\bm{a}^T \bm{v}(t)$ is a non-zero polynomial, so it cannot be zero on the whole interval $(0, 1)$.
\end{proof}

\subsection{Characterization of the Optimal Degree Distributions}

The optimization problem~\eqref{eq:optprob} is convex and Slater's condition~\cite[Chapter~5.2.3]{boyd2004convex} is fulfilled, since the vector $(1/d, 1/d, \dots, 1/d)^T$ is feasible.
As a result, the Karush--Kuhn--Tucker (KKT) conditions are necessary and sufficient conditions on the unique optimum of the problem, cf.~\cite[Chapter~5.5]{boyd2004convex}.
We apply them as follows.

Define the KKT multipliers $\bm{\mu} = (\mu_1, \dots, \mu_d)^T \in \mathbb{R}^d$ and $\lambda \in \mathbb{R}$. 
We have 
  \begin{align*}
    \mathcal{L}(\bm{p}, \bm{\mu}, \lambda) &= f(\bm{p}) - \bm{\mu}^T \bm{p} - \lambda \left(1-\sum_i p_i \right).
  \end{align*}
  Then, the KKT conditions are given by\begin{enumerate} 
    \item $\frac{\partial \mathcal{L}(\bm{p}, \bm{\mu}, \lambda)}{\partial p_i} = \frac{\partial f(\bm{p})}{\partial p_i} - \mu_i + \lambda = 0$ 
    \item $\sum_i p_i = 1$
    and $p_i \geq 0$
    \item $\mu_i \geq 0$
    \item $\mu_i p_i = 0$.
  \end{enumerate}
  We summarize these conditions in the following lemma.  \begin{restatable}{lemma}{lemcharact} \label{lem:charact}
    For a given integer $d\geq 2$, a distribution $\bm{p}^\star \in \mathbb{R}^d$,
    is the unique solution to the optimization problem \eqref{eq:optprob} if and only if
    $$\forall i \in \mathcal{A}\!: -\frac{\partial f(\bm{p}^\star)}{\partial p_i} = f(\bm{p}^\star) \,\text{and}\, \forall i \notin \mathcal{A}\!: -\frac{\partial f(\bm{p}^\star)}{\partial p_i} \leq f(\bm{p}^\star),$$
    where $\mathcal{A}\defeq \{i\in \{1, \dots, d\}: p^\star_i > 0\}$. 
\end{restatable}
  \begin{proof}
    Consider the KKT conditions. 
    If $p^\star_i>0$, then by 4), we have $\mu_i=0$, i.e., by 1) $- \frac{\partial f(\bm{p}^{\star})}{\partial p_i}= \lambda$. Otherwise, $\mu_i\geq0$ and $- \frac{\partial f(\bm{p}^{\star})}{\partial p_i} \leq \lambda$.
  To conclude,\begin{align*}
    \lambda &= \sum_i p^\star_i \lambda \\
            &= \sum_i p^\star_i \left(-\frac{\partial f(\bmp^\star)}{\partial p_i} +\mu_i \right)  \\
            &= -\sum_i p^\star_i \frac{\partial f(\bmp^\star)}{\partial p_i} +\sum_i p^\star_i\mu_i \\
            &= -\sum_i p^\star_i \frac{\partial f(\bmp^\star)}{\partial p_i} \\
            &= - \sum_i p^\star_i \int_0^1 \frac{i \zgen^{i-1} \log(1-\zgen)}{p^{\star\prime}(\zgen)^2} \dzgen \\
            &= - \int_0^1 \left(\sum_i p^\star_i i \zgen^{i-1}\right) \frac{\log(1-\zgen)}{{p^{\star\prime}}(\zgen)^2} \dzgen \\
            &= - \int_0^1 p^{\star\prime}(\zgen) \frac{\log(1-\zgen)}{{p^{\star\prime}}(\zgen)^2} \dzgen \\
            &= f(\bmp^\star).
    \end{align*}
  \end{proof}

An illustration of the conditions in \cref{lem:charact} is depicted in \cref{fig:derivs100}.
\Cref{lem:charact} characterizes the optimal degree distribution for a given $d$ as the solution to a system of equations. However, the equations involve integrals over fractions with logarithms and are typically solved by transcendental functions. See \cref{sec:exampled2} for a worked example for $d=2$ exemplifying this point. 
We believe that the optimal solutions cannot be expressed in terms of elementary functions and thus resort to numerical methods to find them.

We observe that, consistent with \cref{lem:charact}, the optimal degree distributions have sparse support. A depiction can be found in \cref{fig:psheatmap}. 
  We observe that the support of the optimal degree distributions includes $1, 2,$ and $d$ as well as a small number of roughly geometrically spaced support points (or adjacent pairs of support points) between $2$ and $d$. 
As $d$ grows large the results hint at a bifurcation structure in the support of the optimal degree distributions.

\newcommand{\floatpsheatmap}{%
\begin{figure}
  \centering
  \begin{minipage}{\figwidth}
  \includegraphics[width=\linewidth]{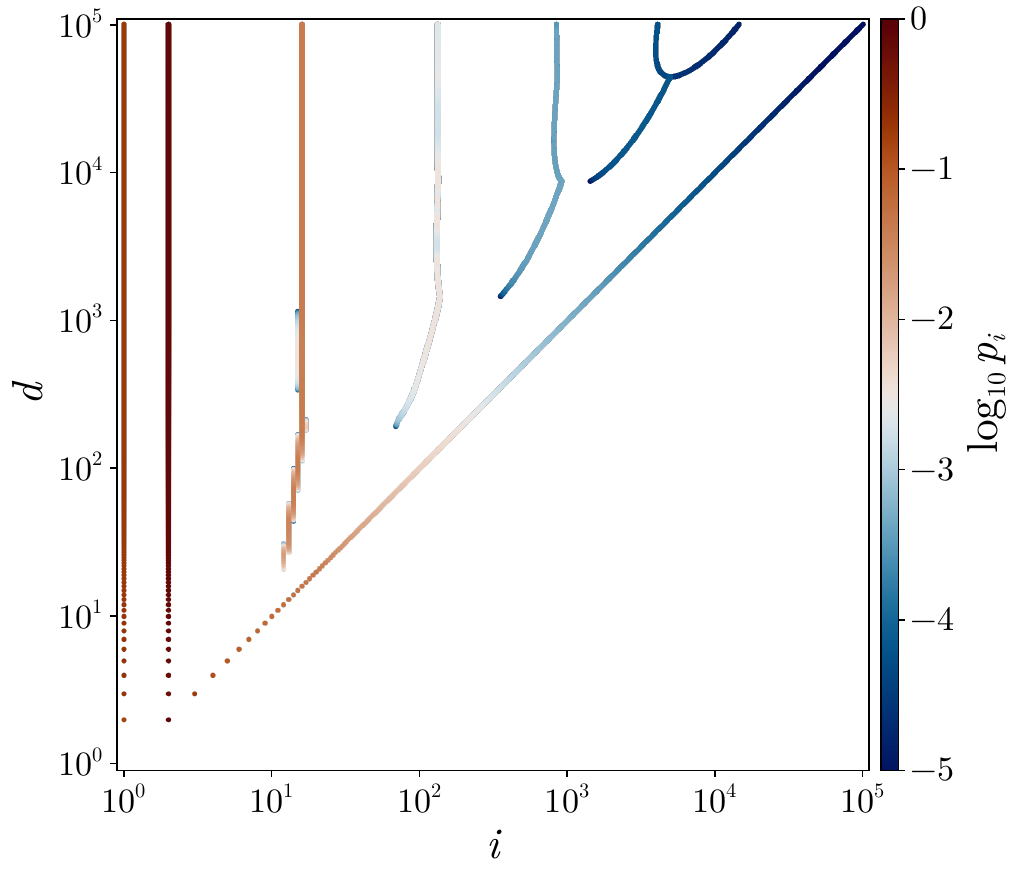}
  \caption{Depiction of the optimal degree distributions for all values of $d$ between $2$ and $10^5$. For each value of $d$ the dots indicate the indices $i$ where $p_i>0$. The corresponding probability value is indicated by color.}
\label{fig:psheatmap}
  \end{minipage}
\end{figure}
}
\floatpsheatmap
\newcommand{\floatdecconstr}{%
\begin{figure*}[ht]

  \begin{subfigure}{0.47\linewidth}
    \centering
    \providecommand{\decodingconstraintcsv}{tikz/decoding_constraint_data/d10_s2.csv}
    \providecommand{\decodingconstraintareas}{tikz/decoding_constraint_data/d10_s2_areas.csv}
    \pgfplotstableread[col sep=comma]{\decodingconstraintareas}\decconstrareas
\pgfplotstablegetelem{0}{Aconstr}\of\decconstrareas \let\decconstrAconstr\pgfplotsretval
\pgfplotstablegetelem{0}{A}\of\decconstrareas \let\decconstrA\pgfplotsretval
\pgfplotstablegetelem{0}{Acummax}\of\decconstrareas \let\decconstrAcummax\pgfplotsretval
\begin{tikzpicture}
\begin{axis}[
    width={0.98\linewidth}, axis equal image,
    xlabel={$g^{(s)}(\zgen, \bm{p})$}, ylabel={$1-\zgen$},
    xmin=0, xmax=2.25, ymin=-0.03, ymax=1.03,
    xtick={0,0.5,...,2.5}, ytick={0,0.25,...,1},
    xticklabel={$\pgfmathprintnumber[fixed,fixed zerofill,precision=1]{\tick}$},
    yticklabel={$\pgfmathprintnumber[fixed,fixed zerofill,precision=2]{\tick}$},
    axis x line*=left, axis y line*=left,
    axis line style={line width=1pt},
    xmajorgrids, ymajorgrids,
    grid style={line width=0.5pt, draw opacity=0.1},
    tick align=inside,
    tick label style={font=\fontsize{8pt}{10.4pt}\selectfont},
    label style={font=\fontsize{11pt}{14.3pt}\selectfont},
    xlabel style={at={(ticklabel cs:0.5)}, anchor=near ticklabel},
    ylabel style={at={(ticklabel cs:0.5)}, anchor=near ticklabel},
    legend cell align={left},
    legend style={font=\fontsize{8pt}{9pt}\selectfont, line width=1pt, at={(0.99,0.6)}, anchor=east,
      inner sep=1pt, nodes={inner xsep=1.5pt, inner ysep=0.5pt}},
    legend image code/.code={\draw[#1] (0cm,0cm) -- (0.45cm,0cm);},
  ]
  \addplot[color=vikblue, line width=1pt, solid]
    table[col sep=comma, x=gconstr, y=oneminusz] {\decodingconstraintcsv};
    \addlegendentry{w/\phantom{o} \eqref{eq:addconstr},~~~~area $=\pgfmathprintnumber[fixed,fixed zerofill,precision=4]{\decconstrAconstr}$}
  \addplot[color=vikred, line width=1pt, dashed]
    table[col sep=comma, x=g, y=oneminusz] {\decodingconstraintcsv};
    \addlegendentry{w/o \eqref{eq:addconstr},~~~~area $=\pgfmathprintnumber[fixed,fixed zerofill,precision=4]{\decconstrA}$}
  \addplot[color=vikred, line width=1pt, dotted]
    table[col sep=comma, x=gcummax, y=oneminusz] {\decodingconstraintcsv};
    \addlegendentry{cum. max, area $=\pgfmathprintnumber[fixed,fixed zerofill,precision=4]{\decconstrAcummax}$}
\end{axis}
\end{tikzpicture}
    \caption{$s=2, d=10$.}
    \label{fig:decconstr10_2}
  \end{subfigure}\hfill
  \begin{subfigure}{0.47\linewidth}
    \centering
    \providecommand{\decodingconstraintcsv}{tikz/decoding_constraint_data/d10_s3.csv}
    \providecommand{\decodingconstraintareas}{tikz/decoding_constraint_data/d10_s3_areas.csv}
    \pgfplotstableread[col sep=comma]{\decodingconstraintareas}\decconstrareas
\pgfplotstablegetelem{0}{Aconstr}\of\decconstrareas \let\decconstrAconstr\pgfplotsretval
\pgfplotstablegetelem{0}{A}\of\decconstrareas \let\decconstrA\pgfplotsretval
\pgfplotstablegetelem{0}{Acummax}\of\decconstrareas \let\decconstrAcummax\pgfplotsretval
\begin{tikzpicture}
\begin{axis}[
    width={0.98\linewidth}, axis equal image,
    xlabel={$g^{(s)}(\zgen, \bm{p})$}, ylabel={$1-\zgen$},
    xmin=0, xmax=2.25, ymin=-0.03, ymax=1.03,
    xtick={0,0.5,...,2.5}, ytick={0,0.25,...,1},
    xticklabel={$\pgfmathprintnumber[fixed,fixed zerofill,precision=1]{\tick}$},
    yticklabel={$\pgfmathprintnumber[fixed,fixed zerofill,precision=2]{\tick}$},
    axis x line*=left, axis y line*=left,
    axis line style={line width=1pt},
    xmajorgrids, ymajorgrids,
    grid style={line width=0.5pt, draw opacity=0.1},
    tick align=inside,
    tick label style={font=\fontsize{8pt}{10.4pt}\selectfont},
    label style={font=\fontsize{11pt}{14.3pt}\selectfont},
    xlabel style={at={(ticklabel cs:0.5)}, anchor=near ticklabel},
    ylabel style={at={(ticklabel cs:0.5)}, anchor=near ticklabel},
    legend cell align={left},
    legend style={font=\fontsize{8pt}{9pt}\selectfont, line width=1pt, at={(0.99,0.6)}, anchor=east,
      inner sep=1pt, nodes={inner xsep=1.5pt, inner ysep=0.5pt}},
    legend image code/.code={\draw[#1] (0cm,0cm) -- (0.45cm,0cm);},
  ]
  \addplot[color=vikblue, line width=1pt, solid]
    table[col sep=comma, x=gconstr, y=oneminusz] {\decodingconstraintcsv};
    \addlegendentry{w/\phantom{o} \eqref{eq:addconstr},~~~~area $=\pgfmathprintnumber[fixed,fixed zerofill,precision=4]{\decconstrAconstr}$}
  \addplot[color=vikred, line width=1pt, dashed]
    table[col sep=comma, x=g, y=oneminusz] {\decodingconstraintcsv};
    \addlegendentry{w/o \eqref{eq:addconstr},~~~~area $=\pgfmathprintnumber[fixed,fixed zerofill,precision=4]{\decconstrA}$}
  \addplot[color=vikred, line width=1pt, dotted]
    table[col sep=comma, x=gcummax, y=oneminusz] {\decodingconstraintcsv};
    \addlegendentry{cum. max, area $=\pgfmathprintnumber[fixed,fixed zerofill,precision=4]{\decconstrAcummax}$}
\end{axis}
\end{tikzpicture}
    \caption{$s=3, d=10$.}
    \label{fig:decconstr10_3}
  \end{subfigure}
  \caption{Solutions to the optimization problem \eqref{eq:optprobs}, with and without the additional constraint \eqref{eq:addconstr}. The solid blue line represents our achievability result, and uses the additional constraint. The red dashed line corresponds to naively solving the optimization problem without the additional constraint and does not in general correspond to an achievable random access expectation. The dotted line corresponds to the cumulative maximum of the red dashed line (from right to left), which is the actual asymptotic (inverse) decoding trajectory for the $\bm{p}$ corresponding to the red dashed line. 
}
    \label{fig:decconstr}
\end{figure*}
}
\begin{figure*}[ht]

  \begin{subfigure}{0.47\linewidth}
    \centering
    \providecommand{\decodingconstraintcsv}{tikz/decoding_constraint_data/d10_s2.csv}
    \providecommand{\decodingconstraintareas}{tikz/decoding_constraint_data/d10_s2_areas.csv}
    \pgfplotstableread[col sep=comma]{\decodingconstraintareas}\decconstrareas
\pgfplotstablegetelem{0}{Aconstr}\of\decconstrareas \let\decconstrAconstr\pgfplotsretval
\pgfplotstablegetelem{0}{A}\of\decconstrareas \let\decconstrA\pgfplotsretval
\pgfplotstablegetelem{0}{Acummax}\of\decconstrareas \let\decconstrAcummax\pgfplotsretval
\begin{tikzpicture}
\begin{axis}[
    width={0.98\linewidth}, axis equal image,
    xlabel={$g^{(s)}(\zgen, \bm{p})$}, ylabel={$1-\zgen$},
    xmin=0, xmax=2.25, ymin=-0.03, ymax=1.03,
    xtick={0,0.5,...,2.5}, ytick={0,0.25,...,1},
    xticklabel={$\pgfmathprintnumber[fixed,fixed zerofill,precision=1]{\tick}$},
    yticklabel={$\pgfmathprintnumber[fixed,fixed zerofill,precision=2]{\tick}$},
    axis x line*=left, axis y line*=left,
    axis line style={line width=1pt},
    xmajorgrids, ymajorgrids,
    grid style={line width=0.5pt, draw opacity=0.1},
    tick align=inside,
    tick label style={font=\fontsize{8pt}{10.4pt}\selectfont},
    label style={font=\fontsize{11pt}{14.3pt}\selectfont},
    xlabel style={at={(ticklabel cs:0.5)}, anchor=near ticklabel},
    ylabel style={at={(ticklabel cs:0.5)}, anchor=near ticklabel},
    legend cell align={left},
    legend style={font=\fontsize{8pt}{9pt}\selectfont, line width=1pt, at={(0.99,0.6)}, anchor=east,
      inner sep=1pt, nodes={inner xsep=1.5pt, inner ysep=0.5pt}},
    legend image code/.code={\draw[#1] (0cm,0cm) -- (0.45cm,0cm);},
  ]
  \addplot[color=vikblue, line width=1pt, solid]
    table[col sep=comma, x=gconstr, y=oneminusz] {\decodingconstraintcsv};
    \addlegendentry{w/\phantom{o} \eqref{eq:addconstr},~~~~area $=\pgfmathprintnumber[fixed,fixed zerofill,precision=4]{\decconstrAconstr}$}
  \addplot[color=vikred, line width=1pt, dashed]
    table[col sep=comma, x=g, y=oneminusz] {\decodingconstraintcsv};
    \addlegendentry{w/o \eqref{eq:addconstr},~~~~area $=\pgfmathprintnumber[fixed,fixed zerofill,precision=4]{\decconstrA}$}
  \addplot[color=vikred, line width=1pt, dotted]
    table[col sep=comma, x=gcummax, y=oneminusz] {\decodingconstraintcsv};
    \addlegendentry{cum. max, area $=\pgfmathprintnumber[fixed,fixed zerofill,precision=4]{\decconstrAcummax}$}
\end{axis}
\end{tikzpicture}
    \caption{$s=2, d=10$.}
    \label{fig:decconstr10_2}
  \end{subfigure}\hfill
  \begin{subfigure}{0.47\linewidth}
    \centering
    \providecommand{\decodingconstraintcsv}{tikz/decoding_constraint_data/d10_s3.csv}
    \providecommand{\decodingconstraintareas}{tikz/decoding_constraint_data/d10_s3_areas.csv}
    \pgfplotstableread[col sep=comma]{\decodingconstraintareas}\decconstrareas
\pgfplotstablegetelem{0}{Aconstr}\of\decconstrareas \let\decconstrAconstr\pgfplotsretval
\pgfplotstablegetelem{0}{A}\of\decconstrareas \let\decconstrA\pgfplotsretval
\pgfplotstablegetelem{0}{Acummax}\of\decconstrareas \let\decconstrAcummax\pgfplotsretval
\begin{tikzpicture}
\begin{axis}[
    width={0.98\linewidth}, axis equal image,
    xlabel={$g^{(s)}(\zgen, \bm{p})$}, ylabel={$1-\zgen$},
    xmin=0, xmax=2.25, ymin=-0.03, ymax=1.03,
    xtick={0,0.5,...,2.5}, ytick={0,0.25,...,1},
    xticklabel={$\pgfmathprintnumber[fixed,fixed zerofill,precision=1]{\tick}$},
    yticklabel={$\pgfmathprintnumber[fixed,fixed zerofill,precision=2]{\tick}$},
    axis x line*=left, axis y line*=left,
    axis line style={line width=1pt},
    xmajorgrids, ymajorgrids,
    grid style={line width=0.5pt, draw opacity=0.1},
    tick align=inside,
    tick label style={font=\fontsize{8pt}{10.4pt}\selectfont},
    label style={font=\fontsize{11pt}{14.3pt}\selectfont},
    xlabel style={at={(ticklabel cs:0.5)}, anchor=near ticklabel},
    ylabel style={at={(ticklabel cs:0.5)}, anchor=near ticklabel},
    legend cell align={left},
    legend style={font=\fontsize{8pt}{9pt}\selectfont, line width=1pt, at={(0.99,0.6)}, anchor=east,
      inner sep=1pt, nodes={inner xsep=1.5pt, inner ysep=0.5pt}},
    legend image code/.code={\draw[#1] (0cm,0cm) -- (0.45cm,0cm);},
  ]
  \addplot[color=vikblue, line width=1pt, solid]
    table[col sep=comma, x=gconstr, y=oneminusz] {\decodingconstraintcsv};
    \addlegendentry{w/\phantom{o} \eqref{eq:addconstr},~~~~area $=\pgfmathprintnumber[fixed,fixed zerofill,precision=4]{\decconstrAconstr}$}
  \addplot[color=vikred, line width=1pt, dashed]
    table[col sep=comma, x=g, y=oneminusz] {\decodingconstraintcsv};
    \addlegendentry{w/o \eqref{eq:addconstr},~~~~area $=\pgfmathprintnumber[fixed,fixed zerofill,precision=4]{\decconstrA}$}
  \addplot[color=vikred, line width=1pt, dotted]
    table[col sep=comma, x=gcummax, y=oneminusz] {\decodingconstraintcsv};
    \addlegendentry{cum. max, area $=\pgfmathprintnumber[fixed,fixed zerofill,precision=4]{\decconstrAcummax}$}
\end{axis}
\end{tikzpicture}
    \caption{$s=3, d=10$.}
    \label{fig:decconstr10_3}
  \end{subfigure}
  \caption{Solutions to the optimization problem \eqref{eq:optprobs}, with and without the additional constraint \eqref{eq:addconstr}. The solid blue line represents our achievability result, and uses the additional constraint. The red dashed line corresponds to naively solving the optimization problem without the additional constraint and does not in general correspond to an achievable random access expectation. The dotted line corresponds to the cumulative maximum of the red dashed line (from right to left), which is the actual asymptotic (inverse) decoding trajectory for the $\bm{p}$ corresponding to the red dashed line. 
}
    \label{fig:decconstr}
\end{figure*}

\subsection{Lower Bound}
  The random access expectation of any binary \property code decoded with a peeling decoder can be bounded from below as follows.
  \begin{lemma} \label{lem:lowerbound}
    For any value of $d$ and any $\bmp$ it holds that\begin{align*}
      f(\bm{p}) \geq \frac{\pi}{4} \approx 0.7854. 
    \end{align*}
  \end{lemma}
  \begin{proof}
    We apply the  Cauchy--Schwarz inequality\footnote{This decomposition of the integrand was proposed by Google Gemini, while we were searching for other bounds on integrals of this form.} to the functions $\varphi(\zgen) = \sqrt{\frac{-\log(1-\zgen)}{p'(\zgen)}}$ and $\psi(\zgen) = \sqrt{p'(\zgen)}$.

    Note that $\int_0^1 p'(\zgen) \dzgen = 1$. We have\begin{align*} 
      f(\bm{p}) &= \int_0^1 \frac{-\log(1-\zgen)}{p'(\zgen)} \dzgen \\ 
                &= \left(\int_0^1 \frac{-\log(1-\zgen)}{p'(\zgen)} \dzgen\right) \cdot \left(\int_0^1 p'(\zgen) \dzgen\right) \\ 
                &= \left(\int_0^1 \varphi(\zgen)^2 \dzgen\right)\cdot \left(\int_0^1 \psi(\zgen)^2 \dzgen\right) \\
      &\geq \left(\int_0^1 \varphi(\zgen)\psi(\zgen) \dzgen \right)^2  \\
      &= \left(\int_0^1 \sqrt{-\log(1-\zgen)} \dzgen\right)^2 \\
      &= \frac{\pi}{4}.
    \end{align*}
  \end{proof}

  The Cauchy--Schwarz inequality holds with equality only when $\varphi(\zgen) \propto \psi(\zgen)$, i.e., when $p'(\zgen) \propto \sqrt{-\log(1-\zgen)}$. However, $p'(\zgen)$ is convex, whereas $\sqrt{-\log(1-\zgen)}$ is strictly concave for $\zgen< 1-e^{-\frac{1}{2}}\approx0.39$. Thus, the bound is loose. Based on the numerical results for large $d$, the gap is $\approx0.0015$.

  It was shown in \cite{bar-lev2026lower} that the bound also holds under optimal decoding over arbitrary fields for symmetric linear ensembles.

\newcommand{\floatoptboundovers}{%
\begin{figure*}[t]
  \begin{subfigure}{0.47\linewidth}
    \centering

      \input{tikz/opt_bound_over_s.tex}%
      \caption{$f^{(s)}(\bm{p})$.}
      \label{fig:opt_bound_over_s_a}
  \end{subfigure}\hfill
  \begin{subfigure}{0.47\linewidth}
    \centering

      \input{tikz/opt_bound_over_s_log.tex}%
      \caption{The gap between $f^{(s)}(\bm{p})$ and $1$, where positive.}
      \label{fig:opt_bound_over_s_b}
  \end{subfigure}
  \caption{Lower bounds and achievable values for $f^{(s)}(\bm{p})$. For each value of $d$ and $s$, the degree distribution is the solution to the augmented optimization problem with the additional constraint \eqref{eq:addconstr}. The extension to real values of $s$ is shown; only integer values of $s$ correspond to the random access to subsets. As $s$ exceeds $1$, the random access expectation quickly approaches $1$. The dotted line represents the Cauchy--Schwarz lower bound in \cref{lem:lowerboundgeneral}, and the dashed line its cumulative maximum (cf. \cref{cor:lowerbound}).}
  \label{fig:opt_bound_over_s}
\end{figure*}
}
\begin{figure*}[t]
  \begin{subfigure}{0.47\linewidth}
    \centering

      \input{tikz/opt_bound_over_s.tex}%
      \caption{$f^{(s)}(\bm{p})$.}
      \label{fig:opt_bound_over_s_a}
  \end{subfigure}\hfill
  \begin{subfigure}{0.47\linewidth}
    \centering

      \input{tikz/opt_bound_over_s_log.tex}%
      \caption{The gap between $f^{(s)}(\bm{p})$ and $1$, where positive.}
      \label{fig:opt_bound_over_s_b}
  \end{subfigure}
  \caption{Lower bounds and achievable values for $f^{(s)}(\bm{p})$. For each value of $d$ and $s$, the degree distribution is the solution to the augmented optimization problem with the additional constraint \eqref{eq:addconstr}. The extension to real values of $s$ is shown; only integer values of $s$ correspond to the random access to subsets. As $s$ exceeds $1$, the random access expectation quickly approaches $1$. The dotted line represents the Cauchy--Schwarz lower bound in \cref{lem:lowerboundgeneral}, and the dashed line its cumulative maximum (cf. \cref{cor:lowerbound}).}
  \label{fig:opt_bound_over_s}
\end{figure*}

\subsection{Worked Example for \texorpdfstring{$d=2$}{d=2}} \label{sec:exampled2}

  We explicitly apply the preceding results for $d=2$ to derive the optimal ratio of weight-$1$ and weight-$2$ columns. Apart from illustrating the results,
    we also recover the $\optT\leq \pi^2/12$ result of~\cite{boruchovsky2025making}, and 
    find that for $d=2$ the optimal degree distribution can likely not be expressed in closed form. We suspect this holds in general. 

  To start, we explore the extreme cases $p_1= 1$ and $p_1\to 0$.
  For $p_1=1, p_2=0$, we have $f(\bm{p}) = \int_0^1 -\log(1-\zgen) \dzgen = 1$, which corresponds to the classic coupon collector's problem. 
  In the other extreme case, $p_1\to0$, $p_2\to1$, we have 
  \begin{align*}
    f(\bm{p}) = \int_0^1 \frac{-\log(1-\zgen)}{2\zgen} \dzgen &= \frac{\Li_2(1)}{2} = \frac{\pi^2}{12} \approx 0.8225, 
  \end{align*}
  where the integral $\Li_2(x) \defeq \int_0^x \frac{-\log(1-\zgen)}{\zgen} \dzgen$ is known as the dilogarithm function.
  Note that the derivative of $g(\zgen, \bmp) = \frac{-\log(1-\zgen)}{2\zgen}$ is $g^\prime(\zgen, \bmp) = \frac{1}{2\zgen^2}\left( \frac{\zgen}{1-\zgen} + \log(1-\zgen) \right)>0$ for $0<\zgen<1$ as $\zgen > -(1-\zgen)\log(1-\zgen)$. 
  Accordingly, a minimally perturbed degree distribution approaches this random access expectation asymptotically.

  This result numerically matches~\cite[Corollary~3]{boruchovsky2025making} exactly. 
  There is also a structural connection. 
  In~\cite{boruchovsky2025making}, an explicit construction of a (non-binary) \property generator matrix is presented.
  Specifically, for an integer $k$ and sufficiently large $q$ and $n$, all possible supports of weight $1$ or $2$ are included equally often as columns of the generator matrix. 
  The values of the non-zero positions are assigned based on MDS matrices, in a way that guarantees linear independence of any set of columns, where the supports allow it. 
  The construction is analyzed combinatorially under optimal decoding, and its random access expectation is found to tend to $\pi^2/12$ in the limit $k\to \infty$. 

  Interestingly, in this case binary \property matrices under peeling achieve the same asymptotic random access expectation as a $q$-ary code (with $q\to\infty$) with the same limit distribution under optimal decoding. 
  In \cite{bar-lev2026lower} it has since been proven that this is the case for all $p_1, p_2$, for all fields, and all distributions of the values of the non-zero entries.

  \Cref{lem:charact} shows that $p_1\to 0$ is not the solution of the optimization problem \eqref{eq:optprob} as $-\frac{\partial f(\bm{p})}{\partial p_1} \to \int_0^1 \frac{-\log(1-\zgen)}{4\zgen^2} \dzgen$ diverges (rather than converging to the objective value $\pi^2/12$).

  For general non-zero $p_1$ and $p_2$ we have $p'(\zgen) = p_1 + 2p_2\zgen = 2p_2 (\zgen+\frac{p_1}{2p_2})$. Applying \cref{lem:charact} and solving the integrals\footnote{E.g., by using a computer algebra system.}, we find that the following need to be equal at the optimum: \begin{align*}
    f(\bm{p}) &= \int_0^1 \frac{-\log(1-\zgen)}{2p_2 (\zgen+\frac{1-p_2}{2p_2})} \dzgen \\
               &= \frac{1}{2p_2} \Li_2\left(\frac{2p_2}{1+p_2}\right),  \\
    -\frac{\partial f(\bm{p})}{\partial p_1} &= \int_0^1 \frac{-\log(1-\zgen)}{(2p_2 (\zgen+\frac{1-p_2}{2p_2}))^2} \dzgen \\
                                             &= \frac{1}{2p_2(1+p_2)} \log(\frac{1+p_2}{1-p_2}), \\ 
    -\frac{\partial f(\bm{p})}{\partial p_2} &= \int_0^1 2\zgen \frac{-\log(1-\zgen)}{(2p_2 (\zgen+\frac{1-p_2}{2p_2}))^2} \dzgen  \\
                                             &= \frac{1}{2p^2_2} \Li_2\left(\frac{2p_2}{1+p_2}\right)-\frac{1}{2p^2_2}\frac{1-p_2}{1+p_2} \log\left(\frac{1+p_2}{1-p_2}\right),
  \end{align*}
  where the last equation based on $-\frac{\partial f(\bm{p})}{\partial p_2}$ is redundant.

  The dilogarithm prevents a closed form solution.
  We find numerically that $p^{\star,2}_1\approx0.15547$ and $p^{\star,2}_2 \approx 0.84453$ with $f(\bm{p}^{\star,2}) \approx 0.7939$. The conditions of \cref{thm:ifgoodthenperfect} are fulfilled since: 1) we have $p^{\star,2}_1 >0$; and 2) $g'(\zgen, \bmp^{\star,2})>0$ becomes $p^{\star,2}_1 +  2p^{\star,2}_2 \zgen > 2 p_2^{\star, 2} (1-\zgen) (-\log(1-\zgen))$, which holds since $\zgen > (1-\zgen) (-\log(1-\zgen))$ for $0<\zgen<1$. Hence, $\optpdT_2=f(\bmp^{\star,2})\approx 0.7939$.
  The value of $\optpdT_2$ differs from the lower bound $\pi/4$ by approximately $1\%$.

  \subsection{Suboptimality of Finite Degrees \texorpdfstring{$d$}{d}}

Despite the optimum for $d=2$ being close to the lower bound, $\optpdT_d$ can always be further decreased by increasing $d$.

\begin{restatable}{lemma}{lemmoreismore} \label{lem:moreismore}
    For integers $d, D \geq 2$, let $\bm{p}^{\star,d}$ and $\bm{p}^{\star,D}$ denote the respective solutions to \eqref{eq:optprob}.
    For any value of $d$ there exists a value $D>d$ s.t. $f(\bm{p}^{\star,d}) > f(\bm{p}^{\star,D})$. 
\end{restatable}
  The statement can be proven by showing that when $d$ is increased for a fixed $\bm{p}$, the KKT conditions are eventually violated. 
 A proof is provided in \cref{app:moreismore}.

\section{Random Access to a Set of Information Symbols} \label{sec:setrecovery}

The analysis extends to the case where a set of $s$ prespecified information symbols needs to be decoded rather than a single symbol, which has been proposed as an open direction in~\cite{bar-lev2024cover}, and was first studied in~\cite{gruica2026generalized} for general linear codes at finite $k$. Previous works~\cite{abraham2024covering,bar-lev2026coded} have analyzed the case where specific sets of information symbols, known before encoding, called files, need to be recovered. 

We apply the same methodology as in \cref{sec:mainresults,sec:analysisoptprob}, which, however, does not give the optimal degree distribution under peeling for general $s$. Nonetheless, we obtain valid asymptotic achievability and converse results for binary \property codes under a peeling decoder.

Let $\mathcal{S}$ be a subset of information symbols of cardinality $s = |\mathcal{S}| \leq k$. 
When exactly $kz \geq s$ information symbols are decoded, the probability that $\mathcal{S}$ is in the set of decoded symbols $\mathcal{D}$ is given by\begin{align*}
\Pr[\mathcal{S} \subseteq \mathcal{D}] = \frac{\binom{kz}{s}}{\binom{k}{s}} = \prod_{i=0}^{s-1} \frac{kz-i}{k-i}.
\end{align*}

For fixed $s$ and $k\to \infty$ we have \begin{align}
  \Pr[\mathcal{S} \subseteq \mathcal{D}] =  \prod_{i=0}^{s-1} \frac{z-i/k}{1-i/k} \to z^s,  \label{eq:inclusionindependent}
\end{align}
i.e., asymptotically information symbols are independently included or not included in $\mathcal{D}$.

Analogously to \eqref{eq:tailsum}, the $s$-random access expectation is given by \begin{align*}
  \pdT^{(s)}(\bm{p}) \defeq \int_0^\infty 1- g^{-1}(r, \bmp)^s \dr,
\end{align*}
by \cref{cor:dandn} and the convergence in \eqref{eq:inclusionindependent}.
This expression equals $\int_0^1  g(\sqrt[s]{\zgen}, \bmp) \dzgen$ if $g(\zgen, \bmp)$ is non-decreasing, and exceeds it otherwise.  
Note that $\zgen$ now represents the asymptotic probability that the whole set $\mathcal{S}$ is decoded, i.e., $\zgen = \zrv^s$, when a fraction $\zrv$ of information symbols is decoded.
Thus, for general $s$ we are interested in minimizing the area under a non-linearly scaled version of the decoding curve.
We have\begin{align*}
  f^{(s)}(\bm{p}) &\defeq \int_0^1 \frac{-\log(1-\sqrt[s]{\zgen})}{p'(\sqrt[s]{\zgen})} \dzgen \\
                  &= \int_0^1 s u^{s-1} \frac{-\log(1-u)}{p'(u)} \du.
\end{align*}

\begin{remark}
  For large $s$ the part of the decoding curve near $z=1$, i.e., the bottom part in \cref{fig:decodingcurves} or \cref{fig:extracurves}, is weighted more heavily in the integral. In the limit $s\to\infty$ the point of intersection of the decoding curve with the x-axis determines $\pdT^{(s)}$. The soliton distribution~\cite{luby2002lt} leads to a square decoding curve (see, e.g.,~\cite{mackay2005fountain,sanghavi2007intermediate}) and achieves $\pdT^{(s)}$ arbitrarily close to $1$ for all $s$.
\end{remark}

The analysis of the optimization problem follows the same steps as for the case $s=1$ in \cref{sec:analysisoptprob}, except for the substitution $\zgen \to \sqrt[s]{\zgen}$.
Convexity can be proven in the same way.
The Cauchy--Schwarz lower bound can be generalized as follows.
\begin{lemma} \label{lem:lowerboundgeneral}
It holds that\begin{align*} 
  f^{(s)}(\bm{p}) &\geq s \left(\int_0^1 \sqrt{- u^{s-1}\log(1-u)} \du\right)^2.
    \end{align*}
\end{lemma}
The proof is analogous to that of \cref{lem:lowerbound}.

    The lower bound is maximized\footnote{Relaxing the integrality constraint on $s$ gives a maximum at  $s\approx 2.84$ with an insubstantially larger value.} for $s=3$ with a value of $\approx 0.9846$.
    The fact that $\pdT^{(s)}(\bmp)$ is non-decreasing in $s$ yields the following corollary.
    \begin{corollary} \label{cor:lowerbound}
      It holds that $\pdT^{(2)}(\bm{p}) \geq 0.9603$ and for $s\geq 3$ it holds that\begin{align*}
        \pdT^{(s)}(\bm{p}) &\geq 0.9846.
      \end{align*}
    \end{corollary}

    Thus, for $s>1$ at most a small gain is possible over the soliton distribution, which achieves a random access expectation arbitrarily close to $1$ for all $s$ simultaneously.

    Define \begin{align*}
      g^{(s)}(t, \bm{p}) \defeq \frac{-\log(1-\sqrt[s]{t})}{p'(\sqrt[s]{t})}.
    \end{align*}

    Solving the generalized optimization problem \begin{align}
      \min_{\bm{p}\in \Delta^d} f^{(s)}(\bm{p}) \label{eq:optprobs}
    \end{align}
    numerically often yields solutions corresponding to  $g^{(s)}(t, \bm{p})$ that are locally decreasing at some $t$. 
    Then, by \cref{cor:dandn} (the inverse of) the asymptotic decoding curve is the cumulative maximum of $g^{(s)}(t, \bm{p})$, which makes the function non-convex.
    Instead, we augment the optimization problem \eqref{eq:optprobs} with the additional constraint
    \begin{gather}
      {g^{(s)}}'(t, \bm{p}) >0 \text{ for } 0 <t<1. \label{eq:addconstr}
    \end{gather}
    For fixed $s$ and $t$, \eqref{eq:addconstr} is linear in the $p_i$ and thus preserves the convexity of the problem. The constraint ensures that the minimum objective value corresponds directly to an achievable $s$-random access expectation. The minimum objective value is only guaranteed to be the minimum $s$-random access expectation for binary \property codes under peeling when it coincides with the solution of the optimization problem without the additional constraint \eqref{eq:addconstr}. 
    This issue is illustrated in \cref{fig:decconstr}, where numerical solutions with and without the additional constraint are shown.

    Like for $s=1$, for $s=2$ we generally observe that the solutions with and without \eqref{eq:addconstr} coincide, as seen in \cref{fig:decconstr10_2}. Numerically we find that for large $d$ the optimal $2$-random access expectation for binary \property codes under peeling is $\approx 0.9603$.

    For $s>2$ we observe that the solutions to \eqref{eq:optprobs} with and without the constraint \eqref{eq:addconstr} generally differ, as exemplified in \cref{fig:decconstr10_3}, where the dashed and dotted red lines do not coincide. Our achievability result is built on the solution with the constraint \eqref{eq:addconstr}, represented by the solid line.
    A numerical comparison is provided in \cref{fig:opt_bound_over_s}, and shows that the achievable results are close to the lower bound.

    Compared to the constructions for fixed $k$ in~\cite{gruica2026generalized}, our achievability results show an $s$-random access expectation below $1$ at $s\geq2$, and our converse result, though much less general, is tighter for the specific case of binary \property codes under peeling in the large blocklength limit.

\section{Conclusion} \label{sec:conclusion}

We have shown an equivalence between binary \property codes for DNA data storage and the encoders of LT codes. Based on this equivalence we have identified codes that achieve improved asymptotic random access expectation over the current literature, by finding the optimal degree distributions under peeling.
The achievability result comes close to an analytic lower bound derived for all binary \property codes under peeling, since extended to optimal decoding~\cite{bar-lev2026lower}. 
We have extended our methodology to cases where a set of information symbols is recovered.
Based on our results, we believe the column-support structure of a code is a promising perspective for analyzing its asymptotic random access expectation. The extension of the converse results to codes that are not \property, as well as a finite-length analysis~\cite{karp2004finite,maneva2005new}, are left for future work.

\section*{Acknowledgments}
The authors thank Avital Boruchovsky and Anina Gruica for helpful discussions. 
The authors acknowledge the use of large language models in the preparation of this manuscript. Notably, Google Gemini suggested the decomposition in \cref{lem:lowerbound}. Anthropic's Claude Opus and Fable models were consulted on writing, numerical methods, and mathematical derivations, and were used for feedback on the manuscript, but did not directly generate any of the manuscript's content, for which the authors take full responsibility.

\appendices
\crefalias{section}{appendix}
\section{Proof of \texorpdfstring{\Cref*{cor:dandn}}{Corollary~\ref{cor:dandn}}} \label{app:proofcordand}

\cordandn*
\begin{proof}
  Since $p_1>0$, we have $p'(t)>0$ for all $t \in [0,1)$, and \eqref{eq:dandncond} can be rewritten as \begin{align*}
    r > g(t, \bmp) \text{ for } 0 \leq t < g^{-1}(r, \bmp), 
  \end{align*}
  which can only be violated at $r$ corresponding to stationary points of $g(t, \bmp)$, i.e., $r$ for which there exists a $t_0$ such that $g(t_0, \bmp)=r$ and $g'(t_0, \bmp)=0$. 

  If $p_1=1$, then no stationary point exists. Otherwise, at any such point it must hold that\begin{gather*}
    g'(t_0, \bmp) = \frac{p'(t_0)\frac{1}{1-t_0} + \log(1-t_0)p''(t_0)}{p'(t_0)^2} = 0, 
    \end{gather*}
    i.e., that 
    \begin{gather*}
    w(t_0) \defeq \frac{p'(t_0)}{(1-t_0) p''(t_0)} +\log(1-t_0) = 0.
  \end{gather*}
  Between any two points where a differentiable function assumes the same value (in this case $w(t)=0$), there must be a stationary point, by Rolle's theorem.
  Differentiating $w(t)$ yields
  \begin{align*}
    w'(t)
          &= \frac{- p'(t)(p'''(t)(1-t)-p''(t))}{((1-t) p''(t))^2}, 
  \end{align*}
  which is zero only at the roots of $p'''(t)(1-t)-p''(t)$, which is a non-zero polynomial of degree at most $d-2$.
  Thus, by Rolle's theorem, $w(t)$ has at most one more zero than $w'(t)$, i.e., $g'(t, \bmp)=0$ at no more than $d-1$ points.
  At all other $r$, \cref{thm:dandn} applies with $\zs(r, \bmp) = g^{-1}(r, \bmp)$.
\end{proof}

\section{Proof of \texorpdfstring{\Cref*{lem:seqconv}}{Lemma~\ref{lem:seqconv}}} \label{app:prooflemseqconv}

\lemseqconv*

\begin{proof}
  We first prove~\eqref{eq:achiev}.
  By \cref{cor:dandn}, $\zrv_k(r, \bpk)$ converges to $g^{-1}(r, \bmp)$, in probability, except for finitely many $r$.
  As $\zrv_k$ is bounded in $[0,1]$, the expectation $\zE_k(r, \bpk)$ converges to $g^{-1}(r, \bm{p})$, except for finitely many $r$, by bounded convergence.
  We have\begin{align*}
    \lim_{k\to\infty} \pdT_k(\bpk) = \lim_{k\to\infty} \int_0^\infty 1 - \zE_k(r, \bpk) \dr.
  \end{align*}
  We change the order of limit and integration using dominated convergence, cf.~\cite[Theorem~1.34]{rudin1987real}. For all $k$, $1-\zE_k(r, \bpk) \leq e^{-\pk_1r}$ as the probability of not decoding a given information symbol is bounded by the probability of drawing the corresponding degree-$1$ coded symbol directly. The number of such received degree-$1$ coded symbols is distributed as $\mathrm{Poisson}(\pk_1r)$ by the splitting property of Poisson processes, cf.~\cite[Section~2.3]{gallager2013stochastic}. As $\pk_1\to p_1$, from some $k=k_0$ onwards, $\pk_1>p_1/2$, i.e., $e^{-r\pk_1}<e^{-rp_1/2}$. Further, $\int_0^\infty e^{-r p_1/2} \dr = \frac{2}{p_1}$ so dominated convergence applies from $k_0$ onwards.
  Accordingly, we have\begin{align*}
    \lim_{k\to \infty} \pdT_k(\bpk) = \int_0^\infty 1-g^{-1}(r, \bm{p}) \dr = \int_0^1 g(\zgen, \bm{p}) \dzgen,
  \end{align*} by the layer cake representation~\cite[Theorem~1.13]{lieb2010analysis}. 

  Next we prove \eqref{eq:converse}. 
  If $p_1>0$ and $g(\zgen, \bm{p})$ is strictly increasing for $0<\zgen<1$, \eqref{eq:achiev} applies directly. Otherwise, we consider a perturbed distribution $\bmqd$, in the style of \cref{lem:sanghavi}, where a small amount of mass is shifted from higher degrees to degree $1$, such that $\qd_1>0$, and \cref{cor:dandn} applies. Specifically, for a small positive $\delta$, let $\qd_1=p_1 + \delta(1-p_1)$  and $\qd_i = (1-\delta)p_i$ for $i\in\{2,\dots,d\}$. 
  By the Poisson splitting property, draws from $\bqkd$ at rate $\frac{r}{1-\delta}$ can be seen as draws from two independent Poisson processes. One Poisson process produces coded symbols according to $\bpk$ at rate $r$, and another produces degree-$1$ coded symbols at rate $\frac{r\delta}{1-\delta}$.
  Thus, it holds that $\pdT_k(\bpk) \geq (1-\delta) \pdT_k(\bqkd)$. Note that any $\delta>0$ ensures that all $\qkd_1>0$. Next, we apply the same steps as in the proof of \eqref{eq:achiev} to $\bqkd$. As $g(\zgen, \bmqd)$ may not be strictly increasing, the layer cake representation becomes an inequality. Accordingly, we have $\lim_{k\to\infty}\pdT_k(\bqkd)\geq \int_0^1 g(\zgen, \bmqd) \dzgen$. 

  As $\delta \to 0$, $\bmqd \to \bmp$. Applying Fatou's lemma~\cite[Theorem~1.28]{rudin1987real} yields\begin{align*}
    \liminf_{\delta\to0} \int_0^1 g(\zgen, \bmqd) \dzgen \geq \int_0^1 \liminf_{\delta \to 0} g(\zgen, \bmqd) \dzgen = f(\bm{p}).
  \end{align*}
  Finally, we have\begin{align*}
    \liminf_{k\to\infty} \pdT_k(\bpk) \geq \liminf_{\delta\to0} (1-\delta) \liminf_{k\to\infty} \pdT_k(\bqkd) \geq f(\bm{p}).
  \end{align*}
\end{proof}

\section{Proof of \texorpdfstring{\Cref*{lem:moreismore}}{Lemma~\ref{lem:moreismore}}} \label{app:moreismore}
\lemmoreismore*

  \begin{proof}
    Toward a contradiction assume there exists a value $d$ such that $f(\bm{p}^{\star,d}) \leq f(\bm{p}^{\star,D})$ for all $D > d$. 
    The concatenation of $\bm{p}^{\star,d}$  and $(D-d)$ zeros, denoted as $\bm{p}^{\star,d} || \bm{0}^{(D-d)},$ is a feasible solution for the problem with degree constraint $D$ and achieves an objective value of $f(\bm{p}^{\star,d} || \bm{0}^{(D-d)}) = f(\bm{p}^{\star,d})$.
    Thus, $f(\bm{p}^{\star,d}) = f(\bm{p}^{\star,D})$ and since minimizing vectors are unique by \cref{lem:convex}, we must have that $\bm{p}^{\star,D} = \bm{p}^{\star,d} || \bm{0}^{(D-d)}$.
    Then, by \cref{lem:charact}, it holds that
    \begin{align}
      \frac{\partial f(\bm{p}^{\star,d} || \bm{0}^{(D-d)})}{\partial p_D} &\geq -f(\bm{p}^{\star,d}) \nonumber \\
      D \int_0^1 \zgen^{D-1}\frac{ -\log(1-\zgen)}{p'(\zgen)^2} \dzgen &\leq f(\bm{p}^{\star,d}). \label{eq:mim}
  \end{align}
  In the remainder we show a contradiction by demonstrating that as $D$ grows large, the left-hand side tends to infinity while the right-hand side remains constant. Denoting an upper bound on $p'(\zgen)^2$ by $m_p \leq d^2$ and using the series expansion $-\log(1-\zgen) = \sum_{i=1}^\infty \frac{\zgen^i}{i}$, we have that 
  \begin{align*}
    D \int_0^1 \zgen^{D-1}\frac{ -\log(1-\zgen)}{p'(\zgen)^2} \dzgen &\geq \frac{D}{m_p} \int_0^1 - \zgen^{D-1} \log(1-\zgen) \dzgen \\
                                                     &= \frac{D}{m_p} \int_0^1 \zgen^{D-1} \sum_{i=1}^\infty \frac{\zgen^i}{i}  \dzgen \\
                                                     &= \frac{D}{m_p} \sum_{i=1}^\infty \int_0^1 \frac{\zgen^{i+D-1}}{i}  \dzgen \\
                                                     &= \frac{D}{m_p} \sum_{i=1}^\infty \frac{1}{i(i+D)}.
  \end{align*}
  It can be readily verified that the partial fraction decomposition of $\frac{1}{i(i+D)}$ is $\frac{1}{Di} - \frac{1}{D(i+D)}$.
  Thus, we have
  \begin{align*}
    D \int_0^1 \zgen^{D-1}\frac{ -\log(1-\zgen)}{p'(\zgen)^2} \dzgen &\geq \frac{1}{m_p} \sum_{i=1}^\infty \frac{1}{i} - \frac{1}{i+D}  \\
                                                     &= \frac{1}{m_p} \left(\sum_{i=1}^\infty \frac{1}{i} - \sum_{i=D+1}^\infty \frac{1}{i}\right)  \\
                                                     &= \frac{1}{m_p} \sum_{i=1}^D \frac{1}{i} \\
                                                     &= \frac{H_D}{m_p}. 
  \end{align*}
  It is well known that the harmonic numbers $H_D$ grow logarithmically without bound. More precisely, for large $D$ it holds that $H_D \approx \gamma + \log(D)$, where $\gamma \approx 0.577$ is the Euler--Mascheroni constant.
  Thus, for large enough $D$, $\frac{H_D}{m_p}$ exceeds $f(\bmp^{\star,d})$, contradicting \eqref{eq:mim}.
\end{proof}

\section{Notes on the Numerical Solution of the Optimization Problem} \label{sec:numericalsol}

The optimization problem \eqref{eq:optprob} is convex, but solving it numerically nonetheless presents several difficulties. The objective function is an integral, where the integrand has singularities at both endpoints, and the polynomial $p'(\zgen)$ has degree $d-1$, leading to ill-conditioned Hessian matrices for large $d$.
Accordingly, the numerical procedure to solve it requires some care.

We address these difficulties through: 1) a careful choice of quadrature rule for numerical integration; 2) warm-starting $d+1$ from the solution for $d$; and 3) high-precision arithmetic, as is elaborated in the following.

\subsection{Algorithmic Details}

We use the tanh-sinh quadrature rule~\cite{takahasi1974double} for numerical integration of $g(t, \bm{p})$ to address the singularities at the endpoints of the integration interval.

For $d=2$ the solution is computed directly, see \cref{sec:exampled2}.  
We derive an initial guess $\bm{p}^{0,d}$ for the distribution of $\bm{p}^{\star,d}$ based on the distribution of $\bm{p}^{\star,d-1}$, by setting ${p}^{0,d}_i = p^{\star,d-1}_i$ for $1\leq i \leq d-2$, $p^{0,d}_{d-1} = 0$, and $p^{0,d}_{d}=p^{\star,d-1}_{d-1}$, i.e., we shift the probability mass of the highest degree up by one.
Then, we refine the probability masses at the support using Newton's method. 

The result of Newton's method informs changes to the support. Probability masses in the support that tend to zero are removed and indices outside the support with large residuals are added to the support.
Changes to the support and runs of Newton's method are iterated until the residual is sufficiently small.

In our experience, this form of warm-starting Newton's method with a close guess and restricting updates to the support of the distribution is required due to the Hessian matrix being ill-conditioned far from the optimum. Additionally, we employ high-precision floating-point arithmetic, specifically quadruple-precision floating-point numbers and arbitrary-precision floating-point types with even higher precision in the Julia programming language~\cite{bezanson2017julia}.

\subsection{Quality of the Numerical Solutions}

We assess the quality of the numerical solutions by how close they are to fulfilling \cref{lem:charact}. 
Specifically, we consider the residuals
\begin{align*}
  \res_i \defeq \begin{cases} 
    |\frac{\partial f(\bm{p}^\star)}{\partial p_i} + f(\bm{p}^\star)| & \text{ if }  i \in \mathcal{A}\\
    \max(0, -\frac{\partial f(\bm{p}^\star)}{\partial p_i} - f(\bm{p}^\star)) & \text{otherwise,}
  \end{cases}
\end{align*}
and the maximum residual $\maxres \defeq \max_{i\in \setint{1}{d}} \res_i$.

For all $i\not\in\mathcal{A}$ and all $d$ the computed value of $\frac{\partial f(\bm{p}^\star)}{\partial p_i} + f(\bm{p}^\star)$ is positive, so $\res_i=0$. 
For $i \in \mathcal{A}$, a small residual always remains. The maximum residuals for $d \leq 10^5$ are shown in \cref{fig:residuals}.

\begin{figure}
  \centering
  \begin{minipage}{\figwidth}

  \input{tikz/residualhist.tex}
  \caption{Histogram of the maximum residuals of the numerical solutions for all $d$ between $2$ and $10^5$. The highest value is $10^{-16.18}$.}
  \label{fig:residuals}
  \end{minipage}
\end{figure}

\bibliographystyle{IEEEtran}
\bibliography{references_journal}

\end{document}